%% file: main_MST_latex.tex
\documentclass[11pt]{article}

\usepackage{tikz}
\usetikzlibrary{backgrounds,positioning,fit,patterns,shadows,calc}

\usepackage{epsfig}
\usepackage{amsfonts}
\usepackage{amssymb}
\usepackage{amstext}
\usepackage{amsmath}

\usepackage{calligra}
\usepackage[T1]{fontenc}

\usepackage{amsthm, thmtools}

\usepackage{paralist}

\usepackage{nameref}
\definecolor{ForestGreen}{rgb}{0.1333,0.5451,0.1333}
\definecolor{DarkRed}{rgb}{0.8,0,0}
\definecolor{Red}{rgb}{1,0,0}
\usepackage[linktocpage=true,
pagebackref=true,colorlinks,
linkcolor=DarkRed,citecolor=ForestGreen,
bookmarks,bookmarksopen,bookmarksnumbered]
{hyperref}
\renewcommand*\backref[1]{\ifx#1\relax \else (cit. on p. #1) \fi} %http://latex.org/forum/viewtopic.php?t=3670

\usepackage{cleveref}

\usepackage{thm-restate} % See section 1.4 of the pdf above

\usepackage{xspace}
\usepackage{color}
\usepackage{enumitem}
\usepackage{comment}
\usepackage{caption}
\usepackage{subcaption}

\usepackage[noend]{algorithmic}
\usepackage[section,boxed]{algorithm}

\usepackage{array}
\usepackage{multirow}
\usepackage{tablefootnote}

\usepackage{wrapfig}

\providecommand{\tabularnewline}{\\}
\makeatletter
\def\thmt@refnamewithcomma #1#2#3,#4,#5\@nil{%
	\@xa\def\csname\thmt@envname #1utorefname\endcsname{#3}%
	\ifcsname #2refname\endcsname
	\csname #2refname\expandafter\endcsname\expandafter{\thmt@envname}{#3}{#4}%
	\fi
}
\makeatother

\declaretheorem[numberwithin=section,refname={Theorem,Theorems},Refname={Theorem,Theorems},name={Theorem}]{thm}
\declaretheorem[numberlike=thm,refname={Lemma,Lemmas},Refname={Lemma,Lemmas},name={Lemma}]{lem}
\declaretheorem[numberlike=thm,refname={Corollary,Corollaries},Refname={Corollary,Corollaries},name={Corollary}]{cor}
\declaretheorem[numberlike=thm,refname={Fact,Facts},Refname={Fact,Facts},name={Fact}]{fact}

\declaretheorem[numberlike=thm,refname={Proposition,Propositions},Refname={Proposition,Propositions},name={Proposition}]{prop}

\declaretheorem[numberlike=thm,refname={Definition,Definitions},Refname={Definition,Definitions},name={Definition}]{defn}
\declaretheorem[style=remark,numberlike=thm,refname={Remark,Remarks},Refname={Remark,Remarks},name={Remark}]{rem}
\declaretheorem[style=remark,numberlike=thm,refname={Claim,Claims},Refname={Claim,Claims}]{claim}

\newcommand{\squishlist}{
	\begin{list}{$\bullet$}
		{ \setlength{\itemsep}{0pt}
			\setlength{\parsep}{2pt}
			\setlength{\topsep}{2pt}
			\setlength{\partopsep}{0pt}
			\setlength{\leftmargin}{1.5em}
			\setlength{\labelwidth}{1em}
			\setlength{\labelsep}{0.5em} } }
	\newcommand{\squishend}{
\end{list}  }

\def\*#1*\ {}

\ifdefined\ShowComment

\def\danupon#1{\marginpar{$\leftarrow$\fbox{D}}\footnote{$\Rightarrow$~{\sffamily #1 --Danupon}}}
\def\thatchaphol#1{\marginpar{$\leftarrow$\fbox{T}}\footnote{$\Rightarrow$~{\sffamily #1 --Thatchaphol}}}

\else

\def\danupon#1{}
\def\thatchaphol#1{}

\fi

\newboolean{short}
\setboolean{short}{true} % set to true if the paper is the short version

\newcommand{\shortOnly}[1]{\ifthenelse{\boolean{short}}{#1}{}}
\newcommand{\longOnly}[1]{\ifthenelse{\boolean{short}}{}{#1}}

\usepackage[left=1in,top=1in,right=1in,bottom=1in]{geometry} % Does NOT work right now
\usepackage{booktabs}
\usepackage{threeparttable}
\makeatletter
\renewcommand{\paragraph}{%
	\@startsection{paragraph}{4}%
	{\z@}{1ex \@plus 1ex \@minus .2ex}{-1em}%
	{\normalfont\normalsize\bfseries}%
}
\makeatother

\newcommand{\patrascu}{P{\v a}tra{\c s}cu\xspace}

\AtBeginDocument{%
	\let\ref\Cref
}

\usepackage{authblk}
\author[1]{Danupon Nanongkai}
\author[1]{Thatchaphol Saranurak}
\author[2]{Christian Wulff-Nilsen}
\affil[1]{KTH Royal Institute of Technology, Sweden}
\affil[2]{University of Copenhagen, Denmark}

\begin{document}
	\global\long\def\cA{{\cal A}\xspace}
	\global\long\def\cB{{\cal B}\xspace}
	\global\long\def\cC{{\cal C}\xspace}
	\global\long\def\cD{{\cal D}\xspace}
	\global\long\def\cE{{\cal E}\xspace}
	\global\long\def\opt{\textsf{OPT}\xspace}
	\global\long\def\decomp{\textsf{decomp}\xspace}
	\global\long\def\msfdecomp{\textsf{MSFdecomp}\xspace}	
	\global\long\def\cP{{\cal P}\xspace}
	\global\long\def\cG{{\cal G}\xspace}
	\global\long\def\sf{\textsf{SF}\xspace}
	\global\long\def\msf{\textsf{MSF}\xspace}
	\global\long\def\cA{{\cal A}\xspace}
	\global\long\def\cB{{\cal B}\xspace}
	\global\long\def\cC{{\cal C}\xspace}
	\global\long\def\cT{{\cal T}\xspace}
	\global\long\def\cS{{\cal S}\xspace}
	\global\long\def\cQ{{\cal Q}\xspace}
	\global\long\def\cH{{\cal H}\xspace}
	\global\long\def\cR{{\cal R}\xspace}
	\global\long\def\cF{{\cal F}\xspace}
	\global\long\def\cL{{\cal L}\xspace}
	\global\long\def\cM{{\cal M}\xspace}
	\global\long\def\cN{{\cal N}\xspace}
	\global\long\def\rtensor{\otimes_{R}}
	\global\long\def\one{\mathbf{1}}
	\global\long\def\cO{{\cal O}\xspace}	
	\global\long\def\polylog{\mbox{polylog}}
	\global\long\def\poly{\mbox{poly}}
	\global\long\def\total#1{|#1(\cdot)|}
	\global\long\def\vol#1{vol(#1)}
	\global\long\def\contract{\textsf{Contract}}
	\global\long\def\disjunion{\textsf{\ensuremath{\dot{\cup}}}}
	\global\long\def\Disjunion{\textsf{\ensuremath{\dot{\bigcup}}}}
	\global\long\def\pruning{\textsf{\textsf{Pruning}}}
	\global\long\def\incident{\overline{E}}
	\global\long\def\largetext{\textnormal{large}}
	\global\long\def\smalltext{\textnormal{small}}

	\newcommand*\samethanks[1][\value{footnote}]{\footnotemark[#1]}
	
\title{Dynamic Minimum Spanning Forest with Subpolynomial Worst-case Update Time}
%	\title{Dynamic Minimum Spanning Forest with $n^{o(1)}$ Worst-case Update Time}

%	
%	\author{
%		Danupon Nanongkai\\{KTH Royal Institute of Technology}
%		\and
%		Thatchaphol Saranurak\\ {KTH Royal Institute of Technology}
%		\and
%		Christian Wulff-Nilsen\\ {University of Copenhagen}
%	}

	\date{}
	\pagenumbering{roman}
	\maketitle
	\input{abstract}

	\pagebreak{}
	
	\tableofcontents{}
	
	\pagebreak{}
	\pagenumbering{arabic}

\input{intro.tex}

\input{prelim.tex}

\input{flow.tex}

\input{balancedcut.tex}

\input{pruning.tex}

\input{pruning_lasvegas.tex}

\input{contraction.tex}

	\input{decomposition.tex}

\input{MSF_algorithm.tex}

\input{MSF_correct.tex}

\input{MSF_time.tex}

\input{openproblems.tex}

\paragraph{Acknowledgement.} 
This project has received funding from the European Research Council (ERC) under the European Union's Horizon 2020 research and innovation programme under grant agreement No 715672. Nanongkai and Saranurak were also partially supported by the Swedish Research Council (Reg. No. 2015-04659).
	
	\bibliographystyle{plain}
	\bibliography{references}

	\pagebreak{}
	
	\appendix

\input{local_pruning.tex}

\input{omitted.tex}

\end{document}

%% file: abstract.tex
% VERSION 2 by Danupon

\begin{abstract}
We present a Las Vegas algorithm for dynamically maintaining a minimum	spanning forest of an $n$-node graph undergoing edge insertions and
deletions. Our algorithm guarantees an \emph{$O(n^{o(1)})$} \emph{worst-case} update time with high probability. This significantly improves the	two recent Las Vegas algorithms by Wulff-Nilsen \cite{Wulff-Nilsen16a} with update time $O(n^{0.5-\epsilon})$ for some constant $\epsilon>0$ and, independently, by Nanongkai and Saranurak \cite{NanongkaiS16} with update time $O(n^{0.494})$ (the latter works only for maintaining a spanning forest).
	
Our result is obtained by identifying the common framework that
	both two previous algorithms rely on, and then improve and combine
	the ideas from both works. There are two main algorithmic components
	of the framework that are newly improved and critical for obtaining
	our result. First, we improve the update time from $O(n^{0.5-\epsilon})$ in \cite{Wulff-Nilsen16a}
	to $O(n^{o(1)})$ for decrementally removing all low-conductance cuts in an expander undergoing edge deletions. 
	Second, by revisiting the ``contraction technique'' by Henzinger and King \cite{HenzingerK97b} and Holm et al. \cite{HolmLT01},
	we show a new approach for maintaining a minimum spanning
	forest in connected graphs with very few (at most $(1+o(1))n$) edges. This significantly improves the previous approach in \cite{Wulff-Nilsen16a,NanongkaiS16} which is based on Frederickson's 2-dimensional topology tree \cite{Frederickson85} and illustrates a new application to this old technique.
\end{abstract}

%% file: intro.tex
\section{Introduction}

In the \emph{dynamic minimum spanning forest ($\msf$)} problem, we
want to maintain a minimum spanning forest $F$ of an undirected edge-weighted
graph $G$ undergoing edge insertions and deletions. In particular,
we want to construct an algorithm that supports the following
operations.
\begin{itemize}[noitemsep]
\item {\sc Preprocess($G$)}: Initialize the algorithm with an input
graph $G$. After this operation, the algorithm outputs a minimum
spanning forest $F$ of $G$. 
\item {\sc Insert($u, v, w$)}: Insert edge $(u,v)$ of weight $w$ to $G$. After this operation,
the algorithm outputs changes to $F$ (i.e. edges to be added to or removed from $F$), if any.
\item {\sc Delete($u, v$)}: Delete edge $(u,v)$ from $G$. After this
operation, the algorithm outputs changes to $F$, if any.
\end{itemize}
The goal is to minimize the \emph{update time}, i.e., the time needed
for outputting the changes to $F$ given each edge update. We call
an algorithm for this problem a \emph{dynamic $\msf$ algorithm}.
Below, we denote respectively by $n$ and $m$ the upper bounds of
the numbers of nodes and edges of $G$, and use $\tilde{O}$ to hide
$\polylog(n)$ factors.

%The dynamic $\msf$ problem is one of the most fundamental dynamic
%graph problems.
%%and has a wide range of applications. 
%Dynamic $\msf$
%algorithms have been used as a main component in several other dynamic
%graph algorithms such as dynamic $k$-connectivity certificate \cite{EppsteinGIN97},
%dynamic min cut \cite{Thorup07mincut} and dynamic cut sparsifier
%\cite{AbrahamDKKP16}. Also in static problems, it can be used for
%approximating the tree packing value and the edge connectivity of
%a graph in near linear time \cite{ThorupK00}.\thatchaphol{Should I also talk about Chekuri's paper that also computes tree packing value?}

The dynamic $\msf$ problem is one of the most fundamental dynamic
graph problems.
Its solutions have been used as a main subroutine for several static and dynamic
graph algorithms, such as tree packing value and edge connectivity approximation~\cite{ThorupK00}, dynamic $k$-connectivity certificate \cite{EppsteinGIN97},
dynamic minimum cut \cite{Thorup07mincut} and dynamic cut sparsifier
\cite{AbrahamDKKP16}.
%
%Also in static problems, it can be used for
%approximating the tree packing value and the edge connectivity of
%a graph in near linear time \cite{ThorupK00}.\thatchaphol{Should I also talk about Chekuri's paper that also computes tree packing value?}
%
More importantly, this problem together with its weaker
variants -- \emph{dynamic connectivity} and \emph{dynamic spanning forest ($\sf$)}%
%and \emph{dynamic $k$-weight $\msf$} ($k$-\msf)
\footnote{The \emph{dynamic $\sf$} problem is the same as the dynamic \msf problem but
	we only need to maintain some spanning forest of the graph.
	In the \emph{dynamic connectivity} problem, we need not to explicitly
maintain a spanning forest. We only need to answer the query,
given any nodes $u$ and $v$, whether $u$ and $v$ are connected
in the graph.}
-- have played a central role in the development in the area of dynamic
graph algorithms for more than three decades. 
The first dynamic $\msf$
algorithm dates back to Frederickson's algorithm from 1985 \cite{Frederickson85},
which provides an $O(\sqrt{m})$ update time. This bound,
combined with the general sparsification technique of Eppstein~et~al. from 1992 \cite{EppsteinGIN97}, implies an $O(\sqrt{n})$ 
update time.

%The first dynamic $\msf$
%algorithm dates back to Frederickson's algorithm from 1985 \cite{Frederickson85},
%which provides an $O(\sqrt{m})$ worst-case update time. This bound,
%combined with the general sparsification technique of Eppstein et
%al. from 1992 \cite{EppsteinGIN97}, implies a $O(\sqrt{n})$ worst-case
%update time. 

Before explaining progresses after the above, it is important to note that the update time can be categorized into two types: An update time that holds for every single update is called {\em worst-case update time}. This is to contrast with an {\em amortized update time} which holds ``on average''\footnote{In particular, for any $t$, an algorithm is said to have an amortized update time
of $t$ if, for any $k$, the total time it spends to process the
first $k$ updates (edge insertions/deletions) is at most $kt$. Thus,
roughly speaking an algorithm with a small amortized update time is
fast \textquotedblleft on average\textquotedblright{} but may take
a long time to respond to a single update.}. 
Intuitively, worst-case update time bounds are generally more
preferable since in 
some applications, such as real-time systems, hard guarantees are needed to process a request before the next request arrives.
The $O(\sqrt{n})$ bound of Frederickson and  Eppstein~et~al.~\cite{Frederickson85,EppsteinGIN97} holds in the worst case. 
By allowing the update time to be \emph{amortized}, 
this bound was significantly improved: Henzinger and King \cite{HenzingerK99}
in 1995 showed Las Vegas randomized algorithms with $O(\log^{3}n)$ amortized update
time for the dynamic $\sf$.
%and $O(k\log^{3}n)$ 
%and dynamic $k$-\msf problems, respectively. 
%
The same authors
\cite{HenzingerK97} in 1997 provided an $O(\sqrt[3]{n}\log n)$ amortized
update time for the more general case of dynamic $\msf$.
Finally,
Holm et al. \cite{HolmLT01} in 1998 presented deterministic dynamic
$\sf$ and $\msf$ algorithms with $O(\log^{2}n)$ and $O(\log^{4}n)$
amortized update time respectively. 
Thus by the new millennium we already knew that, with amortization,
the dynamic $\msf$ problem admits an algorithm with polylogarithmic update
time.
In the following decade, this result has been refined in many
ways, including faster dynamic $\sf$ algorithms (see, e.g. \cite{HenzingerT97,Thorup00,HuangHKP-SODA17}
for randomized ones and \cite{Wulff-Nilsen13a} for a deterministic
one), a faster dynamic $\msf$ algorithm~\cite{HolmRW15}, and an $\Omega(\log n)$
lower bound for both problems \cite{PatrascuD06}.
\begin{comment}
There are improved dynamic $\sf$ algorithms (see, e.g. \cite{HenzingerT97,Thorup00,HuangHKP-SODA17}
for randomized algorithms, \cite{Wulff-Nilsen13a} for a deterministic
one, and \cite{PatrascuD06} for lower bounds). For the dynamic $\msf$
problem, Holm, Rotenberg and Wulff-Nilsen \cite{HolmRW15} recently
showed in 2015 an improved algorithm with $O(\log^{4}n/\log\log n)$
amortized update time.
\end{comment}

Given that these problems were fairly well-understood from the perspective
of amortized update time, many researchers have turned their attention
back to the worst-case update time in a quest to reduce gaps between amortized and worst-case update time 
(one sign of this trend is the 2007 work of \patrascu and Thorup \cite{PatrascuT07}).
%
%\footnote{
%Attempts at reducing the gaps between amortized and worst-case update time are
%not limited to the dynamic $\msf$ problem and its variants, e.g. algorithms
%for dynamic transitive closure \cite{Sankowski04}, dynamic all pair
%shortest paths \cite{Thorup05,AbrahamCK17}, dynamic spanner \cite{BodwinK16},
%and dynamic matching \cite{BhattacharyaHN17soda}.}
%
This quest was not limited to dynamic \msf and its variants (e.g. \cite{Sankowski04,Thorup05,AbrahamCK17,BodwinK16,BhattacharyaHN17soda}), 
but overall the progress was still limited and it has become a big technical challenge whether one can close the gaps.
%
%
%\footnote{One concrete reason why worst-case update time bounds are generally
%	preferable is because of the applications in real-time systems where
%	hard guarantees on the update time are needed before the next request
%	comes. }.
%
In the context of dynamic $\msf$, the $O(\sqrt{n})$ worst-case
update time of \cite{Frederickson85,EppsteinGIN97} has remained the
best for decades until the breakthrough in 2013 by Kapron, King and
Mountjoy \cite{KapronKM13} who showed a Monte Carlo randomized algorithm
with polylogarithmic worst-case bound for the dynamic connectivity
problem (the bound was originally $O(\log^{5}n)$ in \cite{KapronKM13}
and was later improved to $O(\log^{4}n)$ in \cite{GibbKKT15}). Unfortunately,
the algorithmic approach in \cite{KapronKM13,GibbKKT15} seems insufficient for harder problems like dynamic $\sf$
and $\msf$\footnote{Note that the algorithms in \cite{KapronKM13,GibbKKT15} actually maintain
	a spanning forest; however, they cannot output such forest. In particular,  \cite{KapronKM13,GibbKKT15} assume the so-called \emph{oblivious adversary}.
	Thus,  \cite{KapronKM13,GibbKKT15}  do not solve dynamic \sf as we define here, as we require algorithms to report how the spanning forest changes.  
	See further discussions on the oblivious
	adversary in \cite{NanongkaiS16}.},
%\footnote{Although the algorithm in \cite{KapronKM13,GibbKKT15} actually maintains
%a spanning forest as an algorithm for answering the connectivity
%queries, the correctness of the algorithm heavily relies on the fact
%that the spanning forest is never outputted to the adversary. So this
%does not solve the dynamic $\sf$ problem. We note that, however,
%if the adversary is assumed to be \emph{oblivious}, then their spanning
%forest can be outputted, but the adversary can be \emph{adaptive}
%in general. See the detailed discussion about oblivious and adaptive
%adversaries in \cite{NanongkaiS16}.}, 
and the $O(\sqrt{n})$ barrier remained unbroken for both problems.

It was only very recently that the \emph{polynomial improvement} to the $O(\sqrt{n})$ worst-case update time
bound was presented \cite{Wulff-Nilsen16a,NanongkaiS16}\footnote{Prior to this, Kejlberg-Rasmussen et al. \cite{Kejlberg-Rasmussen16}
improved the bound slightly to $O(\sqrt{n(\log\log n)^{2}/\log n})$ for dynamic $\sf$ using word-parallelism. Their algorithm is deterministic.}.
%
%Until very recently, the \emph{polynomial improvement} to the $O(\sqrt{n})$ bound is shown. 
%
Wulff-Nilsen \cite{Wulff-Nilsen16a} showed a
Las Vegas algorithm with $O(n^{0.5-\epsilon})$ update time for some
constant $\epsilon>0$ for the dynamic $\msf$ problem. 
Independently, Nanongkai and Saranurak \cite{NanongkaiS16} presented
two dynamic $\sf$ algorithms: one is Monte Carlo with $O(n^{0.4+o(1)})$
update time and another is Las Vegas with $O(n^{0.49306})$ update
time. 
Nevertheless,
the large gap between polylogarithmic amortized update time and the best
worst-case update time remains.

\paragraph{Our Result.}

We significantly reduce the gap by showing the dynamic $\msf$ algorithm
with \emph{subpolynomial} ($O(n^{o(1)})$) update time:

%\begin{thm}
%\label{thm:main intro}For any $p>0$, there is a Las Vegas dynamic
%$\msf$ algorithm on an $n$-node graph that, for each update, has
%worst-case update time $O(n^{o(1)}\log\frac{1}{p})$ with probability
%at least $1-p$. The preprocessing time is ... 
%\end{thm}

%\begin{thm}
%	\label{thm:main intro}There is a Las Vegas dynamic
%	$\msf$ algorithm on an $n$-node graphs that has  $O(m^{1+o(1)})$ preprocessing time
%	 and can answer each update in $O(n^{o(1)}\log\frac{1}{p})$ time. 
%\end{thm}

\begin{thm}
\label{thm:main intro} There is a Las Vegas randomized dynamic
$\msf$ algorithm on an $n$-node graph that can answer each update in 
$O(n^{o(1)})$ time both in expectation and with high probability.  
\end{thm}

Needless to say, the above result completely subsumes the result in \cite{Wulff-Nilsen16a,NanongkaiS16}.
The $o(1)$ term above hides a $O(\log\log\log n/\log\log n)$ factor.\footnote{Note that by starting from an empty graph and inserting one edge at a time, the preprocessing time of our algorithm is clearly $O(m^{1+o(1)})$, where $m$ is the number of edges in the initial graph. However, note further that the $o(1)$ term in our preprocessing time can be slightly reduced to $O(\sqrt{\log\log m/\log m})$ if we analyze the preprocessing time explicitly instead.}
Recall that Las Vegas randomized algorithms always return correct
answers and the time guarantee is randomized. 
Also recall that an event holds {\em with high probability} (w.h.p.) if it holds with probability at least $1-1/n^c$, where $c$ is an arbitrarily large constant. 

%This result immediately
%implies Las Vegas dynamic connectivity and $\sf$ with the same update
%time. Moreover, our guarantee of update time is not just in expectation
%but with high probability (by setting $p=1/n^{c}$ for a large constant
%$c$). 
%
%So, using a standard trick, we also obtain a Monte Carlo algorithm
%with the same update time that is correct with high probability (details
%omitted). 

%\subsection
\paragraph{Key Technical Contribution and Organization.}
We prove \ref{thm:main intro} 
by identifying the common framework
behind the results of Nanongkai-Saranurak~\cite{NanongkaiS16} and Wullf-Nilsen~\cite{Wulff-Nilsen16a} (thereafter NS and WN), and
significantly improving some components within this framework. In particular, in retrospect it can be said 
that at a high level NS~\cite{NanongkaiS16} and WN~\cite{Wulff-Nilsen16a} share the following three components: 
\begin{enumerate}[noitemsep]
	\item  \emph{Expansion decomposition:}
	This component decomposes the input graph into several expanders
	and the ``remaining'' part with few ($o(n)$) edges.
	\item \emph{Expander pruning:} This component helps maintaining an \msf/\sf in expanders from the first component
	by decrementally removing all low-conductance cuts in an expander undergoing edge deletions.
    \item 
	\emph{Dynamic \msf/\sf on ultra-sparse  graphs:} This components maintains \msf/\sf in the ``remaining'' part obtained from  the first component
	by exploiting the fact that this part has few
	edges\footnote{For the reader who are familiar with the
		results in \cite{NanongkaiS16} and \cite{Wulff-Nilsen16a}.
		The first component are shown in Theorem 4 in \cite{Wulff-Nilsen16a}
		and Theorem 5.1 in \cite{NanongkaiS16}. The second are shown in Theorem
		5 in \cite{Wulff-Nilsen16a} and Theorem 6.1 in \cite{NanongkaiS16}.
		The third are shown in Theorem 3 from \cite{Wulff-Nilsen16a} and
		Theorem 4.2 in \cite{NanongkaiS16}.	
%	Algorithms for graphs with few non-tree edges
%	that exploit the fact that the remaining part of the graph has few
%	edges\footnote{For the reader who are familiar with the
%		results in \cite{NanongkaiS16} and \cite{Wulff-Nilsen16a}.
%		The first component are shown in Theorem 4 in \cite{Wulff-Nilsen16a}
%		and Theorem 5.1 in \cite{NanongkaiS16}. The second are shown in Theorem
%		5 in \cite{Wulff-Nilsen16a} and Theorem 6.1 in \cite{NanongkaiS16}.
%		The third are shown in Theorem 3 from \cite{Wulff-Nilsen16a} and
%		Theorem 4.2 in \cite{NanongkaiS16}.		
%		We note that, in \cite{NanongkaiS16}, they call the algorithm in
%		Theorem 6.1 for the second component the \emph{local expansion decomposition}
%		algorithm. This algorithm is not exactly the dynamic expander pruning
%		algorithm. First, it does more than it needs: namely, it decomposes
%		the expander into components instead of just pruning all low-conductance
%		cuts in the expander. Second, it does less than it needs: it only
%		handle a single batch of edge deletions and not a sequence of edge
%		deletions.\label{fn:old alg}}.
}.
\end{enumerate}

The key difference is that while NS~\cite{NanongkaiS16} heavily relied on developing fast algorithms for these components using recent flow techniques (from, e.g., \cite{Peng14,OrecchiaZ14}),  WN~\cite{Wulff-Nilsen16a} focused on developing a sophisticated way to integrate all components together and used slower (diffusion-based) algorithms for the three components. In this paper we significantly improve algorithms for the second and third components from those in NS~\cite{NanongkaiS16}, and show how to adjust the integration method of WN~\cite{Wulff-Nilsen16a} to exploit these improvements; in particular, the method has to be carefully applied recursively.
%
%can be adjusted to fit with these components.
%further develop algorithms for the second and third components above that are much faster than those in NS~\cite{NanongkaiS16}, and show how they can compose with the integration method of WN~\cite{Wulff-Nilsen16a}. 
Below we discuss how we do this in more detail.

%In this paper we further develop algorithms for the second and third components above that are much faster than those in NS~\cite{NanongkaiS16}, and show how they can compose with the integration method of WN~\cite{Wulff-Nilsen16a}. Below we discuss how we do this in more detail. 

%The main difference in the two works is as follows. While in \cite{NanongkaiS16}
%Nanongkai and Saranurak showed the faster algorithms for the first
%and the second components using some new developments in flow-related
%algorithms from \cite{Peng14,OrecchiaZ14}, allowing the faster update
%time, Wulff-Nilsen \cite{Wulff-Nilsen16a} showed a much more sophisticated
%way in integrating all components together allowing him to solve dynamic
%$\msf$ and not only dynamic $\sf$. In the following, we will now
%list our technical contribution: how we improve some of the main components
%in this framework, and how we combine ideas from \cite{NanongkaiS16,Wulff-Nilsen16a}.

\medskip\noindent{\em (i) Improved expander pruning (Details in \ref{sec:Almost-Flow,sec:LBS cut,sec:pruning,sec:pruning_lasvegas}).}
We significantly improve the running time of the \emph{one-shot}
expander pruning algorithm
by NS \cite{NanongkaiS16}\footnote{In \cite{NanongkaiS16}, the authors actually
	show the local expansion decomposition algorithm which is the same
	as one-shot expander pruning but it does not only prune the graph but also decompose the graph into components. 
	In retrospect, we can see that it is enough to instead
	use the one-shot expander pruning algorithm in \cite{NanongkaiS16}.} and the \emph{dynamic} expander
pruning by WN \cite{Wulff-Nilsen16a}. For the one-shot case, given a \emph{single batch}
of $d$ edge deletions to an expander, the one-shot expander pruning
algorithm by NS \cite{NanongkaiS16} takes $O(d^{1.5+o(1)})$ time for removing all low-conductance
cuts. We improve the running time to $O(d^{1+o(1)})$. 
To do this, in \Cref{sec:Almost-Flow} we first extend a new \emph{local flow-based algorithm}\footnote{By local algorithms, we means algorithms that can output its answer without reading the whole input graph. }
for finding
a low-conductance cut by Henzinger, Rao and Wang \cite{HenzingerRW17}, and then use this extension in \Cref{sec:LBS cut} to get another algorithm
for finding a {\em locally balanced sparse (LBS) cut}.
Then in \ref{sec:static pruning} we apply the reduction from LBS cut algorithms by NS \cite{NanongkaiS16} and obtain an improved one-shot expander pruning algorithm.

For the dynamic case, given a \emph{sequence} of edge deletions to
an expander, the dynamic expander pruning algorithm by WN \cite{Wulff-Nilsen16a} dynamically
removes all low-conductance cuts and takes $O(n^{0.5-\epsilon})$
time for each update. We improve the update time to $O(n^{o(1)})$.
%Further, the algorithm in \cite{Wulff-Nilsen16a} is randomized while ours is deterministic. 
Our algorithm is also arguably simpler
and differ significantly because we do not need random sampling as
in \cite{Wulff-Nilsen16a}. To obtain the dynamic
expander pruning algorithm, we use many instances of the static ones, where 
each instance is responsible on finding low-conductance cuts of different
sizes. Each instance is called periodically with different frequencies (instances for finding
larger cuts are called less frequently). See \ref{sec:dynamic pruning} for details.

\medskip\noindent{\em (ii) Improved dynamic  $\msf$ algorithm on ``ultra-sparse'' graphs (Details in \ref{sec:contraction}).}
We show a new way to maintain dynamic $\msf$ in a graph with
few ($o(n)$) ``non-tree'' edges that can also handle a batch of edge insertions.
Both NS and WN \cite{NanongkaiS16,Wulff-Nilsen16a} used a variant
of Frederickson's 2-dimensional topology tree \cite{Frederickson85}
to do this task\footnote{Unlike \cite{Wulff-Nilsen16a}, the algorithm in \cite{NanongkaiS16} cannot handle inserting a batch of many non-tree edges.}. In this paper, we change the approach to
reduce this problem on graphs with few \emph{non-tree edges} to the same
problem on graphs with few \emph{edges} and {\em fewer nodes}; this allows us to apply recursions later in \Cref{sec:Dynamic MSF}.
We do this by applying the classic ``contraction technique'' of Henzinger and King \cite{HenzingerK97b}
and Holm et al. \cite{HolmLT01} {\em in a new way}: This technique was used extensively previously (e.g.  \cite{HenzingerK97,HenzingerK97b,HolmLT01,HolmRW15,Wulff-Nilsen16a}) to reduce fully-dynamic algorithms to decremental
algorithms (that can only handle deletions). Here, we use this technique so that we can recurse.

%To do this, we revisit the
%classic ``contraction technique'' by Henzinger and King \cite{HenzingerK97b}
%and Holm et al. \cite{HolmLT01}. 
%
%We emphasize that our purpose for
%using this technique is conceptually very different
%from all previous applications of the technique \cite{HenzingerK97,HenzingerK97b,HolmLT01,HolmRW15,Wulff-Nilsen16a}.
%The purpose of all previous applications is for reducing decremental
%algorithms to fully dynamic algorithms. However, this goal is not
%crucial for us. 
%Therefore, this work illustrates a new type of application of the old technique. 
%%
%In \ref{sec:contraction}, we show our reduction by extending the ``worst-case-preserving'' reduction
%by Wulff-Nilsen \cite{Wulff-Nilsen16a} (that also uses the contraction
%technique).

\medskip
\noindent
In \Cref{sec:MST Decomposition,sec:Dynamic MSF}, we take a close look into the integration method
in WN~\cite{Wulff-Nilsen16a} which is used to compose the three components. 
%
%  WN composes his three main components, and then s
%
We show
that it is possible to replace all the three components with the tools based on flow algorithms
from either this paper or from NS \cite{NanongkaiS16} instead. 

In particular, in \Cref{sec:MST Decomposition} we consider a subroutine implicit in WN~\cite{Wulff-Nilsen16a}, which is built on top of the expansion decomposition algorithm (the first component above).  
To make the presentation more modular, we explicitly state this subroutine and its needed properties and name it \emph{\msf decomposition} in \Cref{sec:MST Decomposition}. 
This subroutine can be used as it is constructed in \cite{Wulff-Nilsen16a}, but we further show that it can be slightly improved if 
we replace the diffusion-based expansion decomposition algorithm in \cite{Wulff-Nilsen16a}
with the flow-based expansion decomposition by NS \cite{NanongkaiS16} in the construction. This leads to a slight improvement in the $o(1)$ term in our claimed $O(n^{o(1)})$ update time.

%For the first component (details in \Cref{sec:MST Decomposition}),
%we replace the \emph{diffusion-based} expansion decomposition in \cite{Wulff-Nilsen16a}
%with the \emph{flow-based} expansion decomposition by NS \cite{NanongkaiS16}; the latter is faster
%and provides a better trade-off on the ``quality'' of the decomposition.
%In turn, this implies (by techniques in WN~\cite{Wulff-Nilsen16a}) an improvement over one subroutine implicit in WN~\cite{Wulff-Nilsen16a}. 
%%
%To make the presentation more modular, we explicitly state this subroutine and its properties that we need and name it \emph{\msf decomposition} in \Cref{sec:MST Decomposition}.
%

%
%
%which we call \emph{\msf decomposition} algorithm in this paper. 
%Note that this subroutine 
%
%
%\emph{\msf decomposition} algorithm, which
%is implicitly defined by WN \cite{Wulff-Nilsen16a} and is a crucial tool for preprocessing
%a graph in both \cite{Wulff-Nilsen16a} and in this paper. We explicitly state this decomposition and its properties that are needed in \Cref{sec:MST Decomposition} to make the presentation more modular.

%\footnote{We believe that this makes the  presentation more modular and hope that it is more accessible.}

%
%We make the presentation more
%modular by explicitly stating the decomposition and extracting all
%of its properties that are needed in \Cref{sec:MST Decomposition}.

Then, in \Cref{sec:Dynamic MSF}, we combine (using a method in WN~\cite{Wulff-Nilsen16a})
our improved \msf decomposition algorithm (from  \Cref{sec:MST Decomposition}) with our new dynamic expander pruning algorithm 
and our new dynamic \msf algorithm on ultra-sparse graphs (for the second and third components above). 
%As a result, we obtain a dynamic \msf algorithm with $O(n^{o(1)})$ update time.
As our new algorithm on ultra-sparse graphs is actually a reduction to the dynamic \msf problem on a smaller graph, we recursively apply our new dynamic MSF algorithm on that graph. 
By a careful time analysis of our recursive algorithm, we eventually obtain the $O(n^{o(1)})$ update time.

%% file: prelim.tex
\section{Preliminaries\label{sec:prelim}}

When the problem size is $n$, we denote $\tilde{O}(f(n))=O(f(n)\polylog(n))$,
for any function $f$. We denote by $\disjunion$ and $\Disjunion$
the disjoint union operations. We denote the set minus operation by
both $\setminus$ and $-$. For any set $S$ and an element $e$,
we write $S-e=S-\{e\}=S\setminus\{e\}$.

Let $G=(V,E,w)$ be any weighted graph where each edge $e\in E$ has
weight $w(e)$. We usually denote $n=|V|$ and $m=|E|$. We also just
write $G=(V,E)$ when the weight is clear from the context. We assume
that the weights are distinct. For any set $V'\subseteq V$ of nodes,
$G[V']$ denotes the subgraph of $G$ induced by $V'$. We denote
$V(G)$ the set of nodes in $G$ and $E(G)$ the set of edges in $G$.
In this case, $V(G)=V$ and $E(G)=E$. Let $\msf(G)$ denote the minimum
spanning tree of $G$. For any set $E'\subseteq E$, let $end(E')$
be the set of nodes which are endpoints of edges in $E'$. Sometimes,
we abuse notation and treat the set of edges in $E'$ as a graph $G'=(end(E'),E')$
and vice versa. For example, we have $\msf(E')=\msf(G')$ and $E-\msf(G')=E-E(\msf(G'))$.
The set of \emph{non-tree edges} of $G$ are the edges in $E-\msf(G)$.
However, when it is clear that we are talking about a forest $F$
in $G$, non-tree edges are edges in $E-F$. 

A cut $S\subseteq V$ is a set of nodes. A \emph{volume} of $S$ is
$vol(S)=\sum_{v\in S}\deg(v)$. The \emph{cut size} of $S$ is denoted
by $\delta(S)$ which is the number of edges crossing the cut $S$.
The \emph{conductance} of a cut $S$ is $\phi(S)=\frac{\delta(S)}{\min\{vol(S),vol(V-S)}$.
The conductance of a graph $G=(V,E)$ is $\phi(G)=\min_{\emptyset\neq S\subset V}\phi(S)$.

\danupon{@Thatchaphol: The remark below is changed. Please check.}
\begin{rem}[Local-style input]\label{rem:local input graph}
	Whenever a graph $G$ is given to any algorithm $A$ in this paper, we
	assume that a {\em pointer} to the adjacency list representing $G$ is given to $A$. This is necessary for some of our algorithms which are {\em local} in the sense that they do not even read the whole input graph.  
	Recall that in an adjacency list, for each node $v$  we have a list $\ell_v$ of edges incident to $v$ , and we can
	access the head $\ell_v$ in constant time.  (See details in, e.g., \cite[Section 22.1]{CLRS_book})
	Additionally, we assume that we have a list of nodes whose
	degrees are at least $1$ (so that we do not need to probe
	lists of single nodes). 
\end{rem}

We extensively use the following facts about $\msf$.
\begin{fact}
[\cite{EppsteinGIN97}]For any edge sets $E_{1}$ and $E_{2}$, $\msf(E_{1}\cup E_{2})\subseteq\msf(E_{1})\cup\msf(E_{2})$.\label{fact:sparsify}
\end{fact}
Let $G'=(V',E')$ be a graph obtained from $G$ by contracting some
set of nodes into a single node. We always keeps parallel edges in
$G'$ but sometimes we do not keep all the self loops. We will specify
which self loops are preserved in $G'$ when we use contraction in
our algorithms. We usually assume that each edge in $G'$ ``remember''
its original endpoints in $G$. That is, there are two-way pointers
from each edge in $E'$ to its corresponding edge in $E$. So, we
can treat $E'$ as a subset of $E$. For example, for a set $D\subseteq E$
of edges in $G$, we can write $E'-D$ and this means the set of edges
in $G'$ excluding the ones which are originally edges in $D$. With
this notation, we have the following fact about $\msf$:
\begin{fact}
For any graph $G$ and (multi-)graph $G'$ obtained from $G$ by contracting
two nodes of $G$, $\msf(G')\subseteq\msf(G)$.\label{fact:contract}\end{fact}
\begin{defn}
[Dynamic $\msf$]A \emph{(fully) dynamic $\msf$ algorithm} $\cA$
is given an initial graph $G$ to be preprocessed, and then $\cA$
must return an initial minimum spanning forest. Then there is an online
sequence of edge updates for $G$, both insertions and deletions.
After each update, $\cA$ must return the list of edges to be added
or removed from the previous spanning tree to obtain the new one.
We say $\cA$ is an \emph{incremental/decremental} $\msf$ algorithm
if the updates only contain insertions/deletions respectively.
\end{defn}
The time an algorithm uses for preprocessing the initial graph and
for updating a new $\msf$ is called \emph{preprocessing time }and\emph{
update time} respectively. In this paper, we consider the problem
where the update sequence is generated by an adversary\footnote{There are actually two kinds of adversaries: oblivious ones and adaptive
ones. In \cite{NanongkaiS16}, they formalize these definitions precisely
and discuss them in details. In this paper, however, we maintain $\msf$
which is uniquely determined by the underlying graph at any time (assuming
that the edge weights are distinct). So, there is no difference in
power of the two kinds of adversaries and we will not distinguish
them.}. We say that an algorithm has \emph{update time $t$ with probability
$p$}, if, \emph{for each }update, an algorithm need at most $t$
time to update the $\msf$ with probability at least $p$.

Let $G$ be a graph undergoing a sequence of edge updates. If we say
that $G$ has $n$ nodes, then $G$ has $n$ nodes at any time. However,
we say that $G$ has at most $m$ edges and $k$ non-tree edges, if
at any time, $G$ is updated in such a way that $G$ always has at
most $m$ edges and $k$ non-tree edges. We also say that $G$ is
an \emph{$m$-edge $k$-non-tree-edge graph}. Let $F=\msf(G)$. Suppose
that there is an update that deletes $e\in F$. We say that $f$ is
a \emph{replacement/reconnecting edge }if $F\cup f-e=\msf(G-e)$.

\subsection{Some Known Results for Dynamic $\protect\msf$}

We use the following basic ability of the top tree data structure
(see e.g. \cite{SleatorT83,AlstrupHLT05}).
\begin{lem}
There is an algorithm $\cA$ that runs on an $n$-node edge-weighted
forest $F$ undergoing edge updates. $\cA$ has preprocessing time
$O(n\log n)$ and update time $O(\log n)$. At any time, given two
nodes $u$ and $v$, then in time $O(\log n)$ $\cA$ can 1) return
the heaviest edge in the path from $u$ to $v$ in $F$, or 2) report
that $u$ and $v$ are not connected in $F$.\label{thm:top tree}
\end{lem}
A classic dynamic $\msf$ algorithm by Frederickson \cite{Frederickson85}
has $O(\sqrt{m})$ worst-case update time. Using the same
approach, it is easy to see the following algorithm which is useful
in a multi-graph where $m$ is much larger than $n^{2}$:
\begin{lem}
There is a deterministic fully dynamic $\msf$ algorithm for an $n$-node
graph with $m$ initial edges and has $\tilde{O}(m)$ preprocessing
time and $\tilde{O}(n)$ worst-case update time.\label{lem:MSF in multigraph}
\end{lem}
Next, Wulff-Nilsen \cite{Wulff-Nilsen16a} implicitly showed a decremental
$\msf$ algorithms for some specific setting. In \ref{sec:Dynamic MSF},
we will use his algorithms in the same way he used. The precise statement
is as follows:
\begin{lem}
Let $G=(V,E)$ be any $m$-edge graph undergoing edge deletions, and
let $S\subseteq V$ be a set of nodes such that, for any time step,
every non-tree edge in $E(G)-\msf(G)$ has exactly one endpoint in
$S$. Moreover, every node $u\in V\setminus S$ has constant degree.
Then, there is a decremental $\msf$ algorithm $\cA$ that can preprocess
$G$ and $S$ in time $\tilde{O}(m)$ and handle each edge deletion
in $\tilde{O}(|S|)$ time.\label{lem:non-tree cover}\footnote{For those readers who are familiar with \cite{Wulff-Nilsen16a}, Wulff-Nilsen
showed in Section 3.2.2 of \cite{Wulff-Nilsen16a} how to maintain
the $\msf$ in the graph $G_{2}(\overline{C})$ which has the same
setting as in \ref{lem:non-tree cover}. $G_{2}(\overline{C})$ is
defined in \cite{Wulff-Nilsen16a}, which is the same as $\overline{C}_{2}$
in \ref{sec:Dynamic MSF}. In \cite{Wulff-Nilsen16a}, the set of
\emph{large-cluster vertices }(or\emph{ super nodes}) in $G_{2}(\overline{C})$
corresponds to the set $S$ in \ref{lem:non-tree cover} and every
non-tree edge in $G_{2}(\overline{C})$ has exactly one endpoint as
a large cluster node. There, the number of large cluster vertices
is $|S|=O(n^{\epsilon})$ and the algorithm has update time $|S|=O(n^{\epsilon})$
for some constant $\epsilon>0$.}\end{lem}

%% file: flow.tex
\section{The Extended Unit Flow Algorithm\label{sec:Almost-Flow}}

In this section, we show an algorithm called \emph{Extended Unit Flow}
in \ref{thm:Extended Unit Flow}. It is the main tool for developing
an algorithm in \ref{sec:LBS cut} called\emph{ locally balanced sparse
cut}, which will be used in our dynamic algorithm. The theorem is
based on ideas of flow algorithms by Henzinger, Rao and Wang \cite{HenzingerRW17}.

\paragraph{Flow-related notions.}

We derive many notations from \cite{HenzingerRW17}, but note that
they are not exactly the same. (In particular, we do not consider
edge capacities, but instead use the notion of congestion.) A flow
is defined on an instance\emph{ }$\Pi=(G,\Delta,T)$ consisting of
(i) an unweighted undirected graph $G=(V,E)$, (ii) a \emph{source
function} $\Delta:V\rightarrow\mathbb{Z}_{\ge0}$, and (iii) a \emph{sink
function} $T:V\rightarrow\mathbb{Z}_{\ge0}$. A \emph{preflow} is
a function $f:V\times V\rightarrow\mathbb{Z}$ such that $f(u,v)=-f(v,u)$
for any $(u,v)\in V\times V$ and $f(u,v)=0$ for every \emph{$(u,v)\notin E$}.
Define $f(v)=\Delta(v)+\sum_{u\in V}f(u,v).$ A preflow $f$ is said
to be \emph{source-feasible} (respectively \emph{sink-feasible}) if,
for every node $v$, $\sum_{u}f(v,u)\le\Delta(v)$ (respectively $f(v)\leq T(v)$.).
If $f$ is both source- and sink-feasible, then we call it a \emph{flow.
}We define $cong(f)=\max_{(u,v)\in V\times V}f(u,v)$ as \emph{the
congestion of $f$.} We emphasize that the input and output functions
considered here (i.e. $\Delta,$ $T,$ $f$, and \emph{$cong$}) map
to \emph{integers. }

One way to view a flow is to imagine that each node $v$ initially
has $\Delta(v)$ units of \emph{supply} and an ability to \emph{absorb}
$T(v)$ units of supply. A preflow is a way to ``route'' the supply
from one node to another. Intuitively, in a valid routing the total
supply out of each node $v$ should be at most its initial supply
of $\Delta(v)$ (source-feasibility). A flow describes a way to route
such that all supply can be absorbed (sink feasibility); i.e. in the
end, each node $v$ has at most $T(v)$ units of supply. The congestion
measures how much supply we need to route through each edge.

\danupon{LATER: I use ``unit of supply'' as in {\tt http://encyclopedia2.thefreedictionary.com/Unit+of+Supply}. I don't think the way we use supply is really correct, but we can do it just to preserve the way used before. A better term might be ``commodity'' or ''mass''}

With the view above, we call $f(v)$ (defined earlier)\emph{ the amount
of supply ending at $v$ after $f$}. For every node $v$, we denote
$ex_{f}(v)=\max\{f(v)-T(v),0\}$ as \emph{the excess supply at $v$
after $f$ }and $ab_{f}(v)=\min\{T(v),f(v)\}$ as \emph{the absorbed
supply at $v$ after $f$}. Observe that $ex_{f}(v)+ab_{f}(v)=f(v)$,
for any $v$, and $f$ is a feasible flow iff $ex_{f}(v)=0$ for all
nodes $v\in V$. When $f$ is clear from the context, we simply use
$ex$ and $ab$ to denote $ex_{f}$ and $ab_{f}$. For convenience,
we denote $\total{\Delta}=\sum_{v}\Delta(v)$ as \emph{the total source
supply}, $\total T=\sum_{v}T(v)$ as \emph{the total sink capacity,
}$\total{ex_{f}}=\sum_{v}ex_{f}(v)$ as \emph{the total excess}, and
$\total{ab_{f}}=\sum_{v}ab_{f}(v)$ as \emph{the total supply absorbed}. 
\begin{rem}
[Input and output formats] \label{rem:compact souce sink}\label{rem:output flow}
The input graph $G$ is given to our algorithms as described in \ref{sec:prelim};
in particular, our algorithms do not need to read $G$ entirely. Functions
$\Delta$ and $T$ are input in the form of sets $\{(v,\Delta(v))\mid\Delta(v)>0\}$
and $\{(v,T(v))\mid T(v)<\deg(v)\}$, respectively. Our algorithms
will read both sets entirely.

Our algorithms output a preflow $f$ as a set $\{((u,v),f(u,v))\mid f(u,v)\neq0\}$.
When $f$ is outputted, we can assume that we also obtained functions
$ex_{f}$ and $ab_{f}$ which are represented as sets $\{(v,ex_{f}(v))\mid ex_{f}(v)>0\}$
and $\{(v,ab_{f}(v))\mid ab_{f}(v)>0\}$, respectively. This is because
the time for computing these sets is at most linear in the time for
reading $\Delta$ and $T$ plus the time for outputting~$f$.
\end{rem}
\begin{rem}
[$\overline{T}(\cdot)$]\label{rem:T bar} We need another notation
to state our result. Throughout, we only consider sink functions $T$
such that $T(v)\le\deg(v)$ for all nodes $v\in V$. When we compute
a preflow, we usually add to each node $v$ an \emph{artificial supply}
$\overline{T}(v)=\deg(v)-T(v)$ to both $\Delta(v)$ and $T(v)$ so
that $T(v)=\deg(v)$. Observe that adding the artificial apply does
not change the problem (i.e. a flow and preflow is feasible in the
new instance if and only if it is in the old one). We define $\total{\overline{T}}=\sum_{v}\overline{T}(v)=2m-\total T$
as the \emph{total artificial supply}. This term will appear in the
running time of our algorihtm. 
\end{rem}

\paragraph{The main theorem.}

Now, we are ready to state the main result of this section. 
\begin{thm}
[Extended Unit Flow Algorithm]\label{thm:Extended Unit Flow}There
exists an algorithm \emph{called Extended Unit Flow }which takes the
followings as input: 
\begin{itemize}
\item a graph $G=(V,E)$ with $m$ edges (possibly with parallel edges but
without self loops), 
\item positive integers $h\ge1$ and $F\ge1$, 
\item a source function $\Delta$ such that $\Delta(v)\le F\deg(v)$ for
all $v\in V$, and
\item a sink function $T$ such that $\total{\Delta}\le\total T$ (also
recall that $T(v)\le\deg(v)$, $\forall v\in V$, as in \ref{rem:T bar}).
\end{itemize}

In time $O(hF(\total{\Delta}+\total{\overline{T}})\log m)$ the algorithm
returns (i) a source-feasible preflow $f$ with congestion $cong(f)\leq2hF$
and (ii) $\total{ex_{f}}.$ Moreover, either
\begin{itemize}
\item (Case 1) $\total{ex_{f}}=0$, i.e. $f$ is a flow, or 
\item (Case 2) the algorithm returns a set $S\subseteq V$ such that $\phi_{G}(S)<\frac{1}{h}$
and $vol(S)\ge\frac{\total{ex_{f}}}{F}$. (All nodes in $S$ are outputted.) 
\end{itemize}
\end{thm}

\paragraph{Interpretation of \ref{thm:Extended Unit Flow}. }

One way to interpret \ref{thm:Extended Unit Flow} is the following.
(Note: readers who already understand \ref{thm:Extended Unit Flow}
can skip this paragraph.) Besides graph $G$ and source and sink functions,
the algorithm in \ref{thm:Extended Unit Flow} takes integers $h$
and $F$ as inputs. These integers indicate the input that we consider
``good'': (i) the source function $\Delta$ is not too big at each
node, i.e. $\forall v\in V,\ \Delta(v)\le F\deg(v)$, and (ii) the
graph $G$ has high conductance; i.e. $\phi(G)>1/h$. Note that for
the good input it is possible to find a flow of congestion $\tilde{O}(hF)$:
each set $S\subseteq V$ there can be $\sum_{v\in S}\Delta(v)\leq F\cdot vol(S)$
initial supply (by (i)), while there are $\delta(S)>vol(S)/h$ edges
to route this supply out of $S$ (by (ii)); so, on average there is
$\frac{\sum_{v\in S}\Delta(v)}{\delta(S)}\leq hF$ supply routed through
each edge. This is essentially what our algorithm achieves in Case
1. If it does not manage to compute a flow, it computes some source-feasible
preflow and outputs a ``certificate'' that the input is bad, i.e.
a low-conductance cut $S$ as in Case 2. Moreover, the larger the
excess of the preflow, the higher the volume of $S$; i.e. $vol(S)$
is in the order of $ex_{f}(\cdot)/F$. In fact, this volume-excess
relationship is the key property that we will need later. One way
to make sense of this relationship is to notice that if $vol(S)\geq\total{ex_{f}}/F$,
then we can put as much as $F\cdot vol(S)\geq\total{ex_{f}}$ initial
supply in $S$. With conductance of $S$ low enough ($\phi_{G}(S)\le\frac{1}{2h}$
suffices), we can force most of the initial supply to remain in $S$
and become an excess. Note that this explanation is rather inaccurate,
but might be useful to intuitively understand the interplay between
$vol(S)$, $\total{ex_{f}}$ and $F$. 

Finally, we note again that our algorithm is \emph{local} in the sense
that its running time is lower than the size of $G$. For this algorithm
to be useful later, it is important that the running time is \emph{almost-linear}
in $(\total{\Delta}+\total{\overline{T}})$. Other than this, it can
have any polynomial dependency on $h$, $F$ and the logarithmic terms. 

\danupon{Should we note that we can in fact make the algorithm faster (by eliminating F)?} \danupon{Modify or extend?}

The rest of this section is devoted to proving \ref{thm:Extended Unit Flow}.
The main idea is to slightly extend the algorithm called \emph{Unit
Flow} by Henzinger, Rao and Wang \cite{HenzingerRW17}.%
\begin{comment}
(see \ref{thm:orig unit flow}) This algorithm almost immediately
imply our result in \ref{thm:Extended Unit Flow} except that we do
not have a guarantee on $vol(S)$ as in Case 2 of \ref{thm:Extended Unit Flow}
(see \ref{thm:simplifed unit flow}). Because of this, we need to
modify this algorithm mainly to ensure that all excess is in the returned
set $S$ and $ex_{f}(v)\leq\Delta(v)$ for every $v\in S$. 
\end{comment}

\paragraph{The Unit Flow Algorithm.}

The following lemma, which is obtained by adjusting and simplifying
parameters from Theorem 3.1 and Lemma 3.1 in \cite{HenzingerRW17}.
See \ref{sec:proof unit flow} for the details.
\begin{lem}
\label{thm:simplifed unit flow}There exists an algorithm \emph{called
Unit Flow }which takes the same input as in \ref{thm:Extended Unit Flow}
(i.e. $G$, $h$, $F$, $\Delta$, and $T$) except an additional
condition that $T(v)=\deg(v)$ for all $v\in V$. In time $O(Fh\total{\Delta}\log m)$,
the algorithm returns a source-feasible preflow $f$ with congestion
at most $2hF$. Moreover, one of the followings holds. 
\begin{enumerate}
\item $\total{ex_{f}}=0$ i.e. $f$ is a flow.
\item A set $S$ is returned, where $\phi_{G}(S)<\frac{1}{h}$. Moreover,
$\forall v\in S$: $ex_{f}(v)\le(F-1)\deg(v)$ and $\forall v\notin S$:
$ex_{f}(v)=0$.
\end{enumerate}
\end{lem}
We also note the following fact which holds because the Unit Flow
algorithm is based on the push-relabel framework, where each node
$v$ sends supply to its neighbors only when $ex(v)>0$ and pushes
at most $ex(v)$ units of supply.\danupon{DISCUSS: This claim looks sloppy. We should pinpoint where the claimed property above can be found in Henzinger et al.}
\begin{fact}
\label{rem:push excess}The returned preflow $f$ in \ref{thm:simplifed unit flow}
is such that, for any $v\in V$, $\sum_{u\in V}f(v,u)\le0$ if $f(v)<T(v)$.
\end{fact}
\begin{comment}

\paragraph{The Extended Unit Flow Algorithm.}
\begin{cor}
\label{thm:mod Unit Flow}There exists an algorithm called \emph{Extended
Unit Flow }which takes the same input as in \ref{thm:Extended Unit Flow}
(i.e. $G$, $h$, $F$, $\Delta$, and $T$), and in time $O(Fh(\total{\Delta}+\total{\overline{T}})\log m)$,
returns the same output as in \ref{thm:simplifed unit flow}.\end{cor}
\end{comment}

\paragraph{The Extended Unit Flow Algorithm.}

Next, we slightly extend the Unit Flow algorithm so that it can handle
sink functions $T$ where for some node $v$, $T(v)<\deg(v)$.
\begin{proof}
[Proof of \ref{thm:Extended Unit Flow}]Given $G$, $h$, $F$, $\Delta$,
and $T$ as inputs, the Extended Unit Flow algorithm proceeds as follows.
Let $\Delta'$ and $T'$ be source and sink functions where, for all
nodes $v\in V,$ $\Delta'(v)=\Delta(v)+\overline{T}(v)$ and $T'(v)=T(v)+\overline{T}(v)=\deg(v)$,
respectively. Note that $\total{\Delta'}=\total{\Delta}+\total{\overline{T}}\le\total T+\total{\overline{T}}=2m$,
and $\Delta'(v)\le(F+1)\deg(v)$ for all $v\in V$. According to \ref{rem:compact souce sink},
we can construct the compact representation of $\Delta'$ and $T'$
in time $O(\total{\Delta}+\total{\overline{T}})$. 

Let $\cA$ denote the Unit Flow algorithm from \ref{thm:simplifed unit flow}.
We run $\cA$ with input $(G,h,F+1,\Delta',T')$. $\cA$ will run
in time $O(h(F+1)\total{\Delta'}\log m)=O(Fh(\total{\Delta}+\total{\overline{T}})\log m)$.
Then, $\cA$ returns a source-feasible preflow $f$ with congestion
$2hF$ for an instance $(G,\Delta',T')$ and the excess function $ex'_{f}$
(see \ref{rem:output flow}). If $\cA$ also returns a set $S$, then
the Extended Unit Flow algorithm also returns $S$; otherwise, we
claim that $f$ is a \emph{flow} (not just a preflow) for the instance\emph{
}$(G,\Delta,T)$, as follows. \danupon{TO DO: Above, the algorithm and proof are tied together. We should separate them.}

Recall that for any $v\in V$, the followings hold. 

\begin{align*}
f'(v) & =\Delta'(v)+\sum_{u}f(u,v)\\
f(v) & =\Delta(v)+\sum_{u}f(u,v)\\
ex'_{f}(v) & =\max\{f'(v)-T'(v),0\}\\
ex_{f}(v) & =\max\{f(v)-T(v),0\}
\end{align*}
We argue that $f$ is a source-feasible preflow for $(G,\Delta,T)$
with congestion $2hF$. As guaranteed by $\cA$, $f$ has congestion
$2hF$. The following claim shows that $f$ is source-feasible, i.e.
for all $v$, $\sum_{u}f(v,u)\le\Delta(v)$. 
\begin{claim}
$\sum_{u}f(v,u)\le\max\{0,\Delta'(v)-T'(v)\}\le\Delta(v)$ for any
$v\in V$.
\begin{proof}
There are two cases. If $f'(v)<T'(v)$, then $\sum_{u}f(v,u)\le0$
by \ref{rem:push excess}. If $f'(v)\ge T'(v)$, then $\sum_{u}f(v,u)=\Delta'(v)-f'(v)\le\Delta'(v)-T'(v)$.
So $\sum_{u}f(v,u)\le\max\{0,\Delta'(v)-T'(v)\}$. Since $\Delta(v)\ge0$,
we only need to show $\Delta'(v)-T'(v)\le\Delta(v)$. Indeed, $\Delta'(v)-T'(v)=\Delta(v)-T(v)\le\Delta(v)$. 
\end{proof}
\end{claim}
Therefore, we conclude that $f$ is a source-feasible preflow for
$(G,\Delta,T)$ with congestion $2hF$. 

Next, we show that the excess function $ex_{f}$ w.r.t. $\Delta$
and $T$ is the same function as the excess function $ex'_{f}$ w.r.t.
$\Delta'$ and $T'$, which is returned by $\cA$.
\begin{claim}
For any $v\in V$, $ex'_{f}(v)=ex_{f}(v)$.\label{claim:same excess}
\begin{proof}
It suffices to show $f'(v)-T'(v)=f(v)-T(v)$ for all $v\in V$. Indeed,
for every $v\in V$, $f'(v)-f(v)=\Delta'(v)-\Delta(v)=\overline{T}(v)=T(v)-T'(v)$.
\end{proof}
\end{claim}
Therefore, by \ref{claim:same excess}, in the first case where $\cA$
guarantees $\total{ex'_{f}}=0$, we have $\total{ex_{f}}=\total{ex'_{f}}=0$,
i.e. $f$ is a flow for $(G,\Delta,T)$. In the second case where
$\cA$ returns a cut $S$ where $\phi_{G}(S)<\frac{1}{h}$, we have
that $\forall v\in S$, $ex_{f}(v)=ex'_{f}(v)\le((F+1)-1)\deg(v)=F\deg(v)$
and $\forall v\notin S$, $ex_{f}(v)=ex'_{f}(v)=0$. By summing over
all nodes $v$, we have $\total{ex_{f}}=\sum_{v\in S}ex_{f}(v)+\sum_{v\notin S}ex_{f}(s)\le F\sum_{v\in S}\deg(v)=Fvol(S)$.
That is, $vol(S)\ge\total{ex_{f}}/F$.\end{proof}

%% file: balancedcut.tex
\section{Locally Balanced Sparse Cut\label{sec:LBS cut}}

In this section, we show an algorithm for finding a \emph{locally
balanced sparse cut}, which is a crucial tool in \ref{sec:pruning}.
The main theorem is \ref{thm:LBS cut alg}. First, we need this definition:

\begin{defn}
[Overlapping]For any graph $G=(V,E)$, set $A\subset V$, and real
$0\leq\sigma\leq1$, we say that a set $S\subset V$ is \emph{$(A,\sigma)$-overlapping}
in $G$ if $vol(S\cap A)/vol(S)\ge\sigma$.\label{def:overlapping}

\danupon{What's the difference between a "cut" $S$ and a "set" $S$?}
\end{defn}
Let $G=(V,E)$ be a graph. Recall that a cut $S$ is \emph{$\alpha$-sparse}
if it has conductance $\phi(S)=\frac{\delta(S)}{\min\{vol(S),vol(V-S)\}}<\alpha$.
Consider any set $A\subset V$, an overlapping parameter $0\leq\sigma\leq1$
and a conductance parameter $0\leq\alpha\leq1$. Let $S^{*}$ be the
set of largest volume that is $\alpha$-sparse $(A,\sigma)$-overlapping
and such that $vol(S^{*})\le vol(V-S^{*})$. We define $\opt(G,\alpha,A,\sigma)=vol(S^{*})$\emph{.}
If $S^{*}$ does not exist, then we define $\opt(G,\alpha,A,\sigma)=0$.
From this definition, observe that $\opt(G,\alpha,A,\sigma)\le\opt(G,\alpha',A,\sigma)$
for any $\alpha\le\alpha'$. Now, we define the locally balanced sparse
cut problem formally:

\danupon{DISCUSS: Should we simply define the volume of a cut to be $\min(vol(S), vol(V\setminus S))$? This will avoid the technical term $vol(S)\leq vol(V-S)$ many times. If this is the case, we should use a different notation to clearly distinguish set $S$ from cut $S$ (e.g. $(S, \bar{S})$).}
\begin{defn}
[Locally Balanced Sparse (LBS) Cut]\label{def:LBS cut} Consider
any graph $G=(V,E)$, a set $A\subset V$, and parameters $c_{size}\geq1,c_{con}\geq1,\sigma$
and $\alpha$. We say that a cut $S$ where $vol(S)\le vol(V-S)$
is a \emph{$(c_{size},c_{con})$-approximate locally balanced sparse
cut} \emph{with respect to $(G,\alpha,A,\sigma)$ }(in short, \emph{$(c_{size},c_{con},G,\alpha,A,\sigma)$-LBS
cut}) if 
\begin{align}
\phi(S) & <\alpha\quad\mbox{\quad\mbox{and\quad\quad}}c_{size}\cdot vol(S)\ge\opt(G,\alpha/c_{con},A,\sigma).\label{eq:LBS cut}
\end{align}
\begin{comment}
 That is, $S$ is an $\alpha$-sparse cut such that $c_{size}\cdot vol(S)$
is as large as the volume of the largest $(A,\sigma)$-overlapping
$(\alpha/c_{con}$)-sparse cut. 
\end{comment}

\end{defn}
In words, the \emph{$(c_{size},c_{con},G,\alpha,A,\sigma)$}-LBS\emph{
}cut can be thought of as a relaxed version of $\opt(G,\alpha,A,\sigma)$:
On the one hand, we define $\opt(G,\alpha,A,\sigma)$ to be a \emph{highest-volume
}cut with low enough conductance and high enough overlap with $A$
(determined by $\alpha$ and $\sigma$ respectively). On the other
hand, a \emph{$(c_{size},c_{con},G,\alpha,A,\sigma)$}-LBS cut does
\emph{not} need to overlap with $A$ at all; moreover, its volume
is only compared to $\opt(G,\alpha/c_{con},A,\sigma)$, which is at
most $\opt(G,\alpha,A,\sigma)$, and we also allow the gap of $c_{size}$
in such comparison. We note that the existence of a \emph{$(c_{size},c_{con},G,\alpha,A,\sigma)$}-LBS
cut $S$ implies that any $(A,\sigma)$-overlapping cut of volume
more than $c_{size}\cdot vol(S)$ must have conductance at least $\alpha/c_{con}$
(because any $(A,\sigma)$-overlapping cut with conductance less than
$\alpha/c_{con}$ has volume at most $c_{size}\cdot vol(S)$).

\begin{comment}
We note that $(c_{size},c_{con},G,\alpha,A,\sigma)$-LBS cut $S$
may not be $(A,\sigma)$-overlapping. 
\end{comment}

\begin{comment}
 
\end{comment}

\danupon{TO DO: In definitions above and below, are parameters reals or integers? How about their range? vol should be operatorname{vol}}
\danupon{TO DO: I changed "with parameter" to w.r.t. in definition above}
\danupon{TO DO: Did we have an explanation for the above in previous paper?}
\danupon{Again, we have to be clear why we use "cut" and not "set" in the definition above.}
\begin{defn}
[LBS Cut Algorithm]\label{def:most balanced sparse cut alg-1}For
any parameters $c_{size}$ and $c_{con}$, \emph{a $(c_{size},c_{con})$-approximate
algorithm} for the LBS cut problem (in short, \emph{$(c_{size},c_{con})$-approximate
LBS cut algorithm}) takes as input a graph $G=(V,E)$, a set $A\subset V$,
an overlapping parameter $0\leq\sigma\leq1$, and an conductance parameter
$0\leq\alpha\leq1$. Then, the algorithm either 
\begin{itemize}
\item (Case 1) finds a $(c_{size},c_{con},G,\alpha,A,\sigma)$-LBS cut $S$,
or
\item (Case 2) reports that there is no $(\alpha/c_{con})$-sparse $(A,\sigma)$-overlapping
cut, i.e. $\opt(G,\alpha/c_{con},A,\sigma)=0$.
\end{itemize}
\end{defn}
From \ref{def:most balanced sparse cut alg-1}, if there exists an
$(\alpha/c_{con})$-sparse $(A,\sigma)$-overlapping cut, then a $(c_{size},c_{con})$-approximate
LBS cut algorithm $\cA$ can only do Case 1, or if there is no $\alpha$-sparse
cut, then $\cA$ must do Case 2. However, if there is no $(\alpha/c_{con})$-sparse
$(A,\sigma)$-overlapping cut but there is an $\alpha$-sparse cut,
then $\cA$ can either do Case 2, or Case 1 (which is to find any
$\alpha$-sparse cut in this case). 

The main result of this section is the following:
\begin{thm}
\label{thm:LBS cut alg}Consider the special case of the LBS cut problem
where the input ($G$, $A$, $\sigma,$ $\alpha$) is always such
that (i) $2vol(A)\le vol(V-A)$ and (ii) $\sigma\in[\frac{2vol(A)}{vol(V-A)},1]$.
In this case, there is a $(O(1/\sigma^{2}),O(1/\sigma^{2}))$-approximate
LBS cut algorithm that runs in $\tilde{O}(\frac{vol(A)}{\alpha\sigma^{2}})$
time. \danupon{Is it the case that this algorithm works under some conditions while previous ones work under any condition? Eig: they all have there same kind of condition}
\end{thm}
We note that in our later application it is enough to have an algorithm
with $\poly(\frac{\log n}{\alpha\sigma})$ approximation guarantees
and running time almost linear in $vol(A)$ (possibly with $\poly(\frac{\log n}{\alpha\sigma})$).
\begin{comment}
Also note that it is possible to improve the running time of the above
algorithm to $\tilde{O}(\frac{vol(A)}{\alpha\sigma})$\footnote{See \ref{sec:improved flow} for details.},
but this is not necessary for our application.
\end{comment}

Before proving the above theorem, let us compare the above theorem
to related results in the literature. Previously, Orecchia and Zhu
\cite{OrecchiaZ14} show two algorithms for a problem called \emph{local
cut improvement}. This problem is basically the same as the LBS cut
problem except that there is no guarantee about the volume of the
outputted cut. Nanongkai and Saranurak \cite{NanongkaiS16} show that
one of the two algorithms by \cite{OrecchiaZ14} implies a $(\frac{3}{\sigma},\frac{3}{\sigma})$-approximate
LBS cut algorithm with running time $\tilde{O}((\frac{vol(A)}{\sigma})^{1.5})$.
While the approximation guarantees are better than the one in \ref{thm:LBS cut alg},
this algorithm is too slow for us. By the same techniques, one can
also show that the other algorithm by \cite{OrecchiaZ14} implies
a $(n,\frac{3}{\sigma})$-approximate LBS cut algorithm with running
time $\tilde{O}(\frac{vol(A)}{\alpha\sigma})$ similar to \ref{thm:LBS cut alg},
but the approximation guarantee on $c_{size}$ is too high for us.
Thus, the main challenge here is to get a good guarantee on both $c_{size}$
and running time. Fortunately, given the Extended Unit Flow algorithm
from \ref{sec:Almost-Flow}, it is not hard to obtain \ref{thm:LBS cut alg}.

\paragraph{Proof of \ref{thm:LBS cut alg}. }

The rest of this section is devoted to proving \ref{thm:LBS cut alg}.
Let ($G$, $A$, $\sigma,$ $\alpha$) be an input as in \ref{thm:LBS cut alg}.
Define $F=\left\lceil 1/\sigma\right\rceil $ and $h=\left\lceil 1/\alpha\right\rceil $.
Define the source and sink functions as 

\begin{align*}
\Delta(v) & =\begin{cases}
\left\lceil 1/\sigma\right\rceil \deg(v) & \forall v\in A\\
0 & \forall v\notin A
\end{cases}\\
T(v) & =\begin{cases}
0 & \forall v\in A\\
\deg(v) & \forall v\notin A
\end{cases}
\end{align*}

\begin{comment}
Consider a source function $\Delta:V\rightarrow\mathbb{Z}_{\ge0}$
where $\Delta(v)=\left\lceil 1/\sigma\right\rceil \deg(v)$ for all
$v\in A$ and $\Delta(v)=0$ for all $v\notin A$, and a sink function
$T:V\rightarrow\mathbb{Z}_{\ge0}$ where $T(v)=\deg(v)$ for all $v\notin A$,
and $T(v)=0$ for all $v\in A$. 
\end{comment}

Now we run the Extended Unit Flow algorithm from \ref{thm:Extended Unit Flow},
with input $(G,h,F,\Delta,T)$. Note that this input satisfies the
conditions in \ref{thm:Extended Unit Flow}: 
\begin{itemize}
\item $h\geq1$ and $F\geq1$ because $\sigma\leq1$ and $\alpha\leq1$, 
\item for all $v\in V,$ $\Delta(v)\leq\left\lceil 1/\sigma\right\rceil \deg(v)=F\deg(v)$
and $T(v)\leq\deg(v)$, and 
\item $\total{\Delta}=\left\lceil 1/\sigma\right\rceil vol(A)\le\left\lceil \frac{vol(V-A)}{2vol(A)}\right\rceil vol(A)\le vol(V-A)=\sum_{v\notin A}\deg(v)=\total T$
. 
\end{itemize}
The Extended Unit Flow algorithm $\cA$ finishes in time $O(hF(\total{\Delta}+\total{\overline{T}})\log m)=\tilde{O}(\frac{vol(A)}{\alpha\sigma^{2}})$,
where the equality is because of the inequalities above and the fact
that $\total{\bar{T}}=\sum_{v\in A}\deg(v)=vol(A)$. When $\cA$ finishes,
we obtain a source-feasible preflow $f$ with congestion $cong(f)\leq2hF=O(\frac{1}{\alpha\sigma})$,
total excess $\total{ex_{f}}$, and possibly a cut $S$ (when $\total{ex_{f}}>0$).
Let $c_{size}=\frac{2F}{\sigma}=O(1/\sigma^{2})$ and $c_{con}=\frac{2\alpha\cdot cong(f)}{\sigma}=O(1/\sigma^{2})$.
We output as follows.
\begin{itemize}
\item If $\total{ex_{f}}=0$, we report that there is no $(\alpha/c_{con})$-sparse
$(A,\sigma)$-overlapping cut (i.e. $\opt(G,\alpha/c_{con},A,\sigma)=0$)
as in Case 2 of \ref{def:most balanced sparse cut alg-1}. 
\item Otherwise, we output $S$ as an $(c_{size},c_{con},G,\alpha,A,\sigma)$-approximate
LBS cut as in Case 1 of \ref{def:most balanced sparse cut alg-1}.
\end{itemize}
\begin{comment}
Case 1: $\opt(G,\alpha/c_{con},A,\sigma)=0$; ie. there is no $(\alpha/c_{con})$-sparse
$(A,\sigma)$-overlapping cut

We obtain a number $s$, a preflow $f$, and a cut $S$ (when $s>0$).
Let $c_{cap}=UF\left\lceil \log\total{\Delta}\right\rceil =O(\frac{\log(vol(A)/\sigma)}{\alpha\sigma})$,
and $G(c_{cap})$ be a graph $G$ where each edge has capacity $c_{cap}$.
By \ref{thm:Extended Unit Flow}, $f$ is a feasible preflow for $(G(c_{cap}),\Delta,T)$
where $\total{ex_{f}}\le s$. Moreover, the cut $S$, if returned,
is such that $\phi(S)\le(\frac{20\log2m}{h}+\frac{4}{U})=\alpha$
and $vol(S)\ge\frac{c_{0}\sigma s}{\log n}$ for some constant $c_{0}$.
\end{comment}
{} 

To prove that the above algorithm is correct, observe that it suffices
to prove that $\opt(G,\frac{\sigma}{2cong(f)},A,\sigma)\le2\total{ex_{f}}/\sigma$.
Indeed, if $\total{ex_{f}}=0$, then this implies $\opt(G,\alpha/c_{con},A,\sigma)=0$
by the choice of $c_{con}$. Otherwise, we have $\total{ex_{f}}>0$.
Recall from \ref{thm:Extended Unit Flow} that the outputted set $S$
is such that $\phi_{G}(S)<\frac{1}{h}\leq\alpha$ and $vol(S)\ge\total{ex_{f}}/F$.
This means that $\opt(G,\alpha/c_{con},A,\sigma)\le2\total{ex_{f}}/\sigma\le c_{size}\cdot vol(S)$.
Therefore, we conclude with the following claim:
\begin{claim}
For any $(A,\sigma)$-overlapping cut $S'$ where $vol(S')>2\total{ex_{f}}/\sigma$,
we have $\phi(S')\ge\frac{\sigma}{2cong(f)}$. That is, $\opt(G,\frac{\sigma}{2cong(f)},A,\sigma)\le2\total{ex_{f}}/\sigma$.
\begin{proof}
Let $\Delta(S')=\sum_{v\in S'}\Delta(v)$ be the total source supply
from nodes in $S'$, $ex_{f}(S')=\sum_{v\in S'}ex_{f}(v)$ be the
total excess supply in nodes in $S'$, and $ab_{f}(S')=\sum_{v\in S'}ab_{f}(v)$
be the total supply absorbed by nodes in $S'$. Observe that $\Delta(S')=\left\lceil 1/\sigma\right\rceil vol(A\cap S')$
and $ab_{f}(S')\le vol(S'-A)$. So we have 
\begin{align*}
cong(f)\cdot\delta(S') & \ge\Delta(S')-ex_{f}(S')-ab_{f}(S')\\
 & \ge\frac{1}{\sigma}vol(A\cap S')-\total{ex_{f}}-vol(S'-A).
\end{align*}
This implies
\begin{eqnarray*}
cong(f)\cdot\phi(S') & \ge & \frac{\frac{1}{\sigma}vol(A\cap S')-\total{ex_{f}}-vol(S'-A)}{vol(S')}\\
 & = & \frac{\frac{1}{\sigma}vol(A\cap S')-\total{ex_{f}}-(vol(S')-vol(A\cap S'))}{vol(S')}\\
 & = & \frac{(\frac{1}{\sigma}+1)vol(A\cap S')-\total{ex_{f}}-vol(S')}{vol(S')}\\
 & \ge & (\frac{1}{\sigma}+1)\sigma-\frac{\sigma}{2}-1=\frac{\sigma}{2}.
\end{eqnarray*}

\end{proof}
\end{claim}
\begin{comment}
This means that $\opt(G,\frac{\sigma}{2U},A,\sigma)\le2s/\sigma$.
When $s=0$, we have $\opt(G,\alpha/c_{con},A,\sigma)=0$; thus our
algorithm outputs correctly in this case. When $s>0$, we have $\phi_{G}(S)\leq\alpha$
(from before) and

\begin{eqnarray*}
c_{size}\cdot vol(S) & \ge & \frac{2\log n}{c_{0}\sigma^{2}}\times\frac{c_{0}\sigma s}{\log n}=2s/\sigma\ge\opt(G,\frac{\sigma}{2U},A,\sigma)=\opt(G,\alpha/c_{con},A,\sigma).
\end{eqnarray*}

\danupon{We should say where the inequalities above come from.}

Therefore, $S$ is a $(c_{size},c_{con},G,\alpha,A,\sigma)$-LBS cut
as defined in \ref{def:LBS cut}; so our algorithm outputs correctly
in this case. 

\danupon{IMPORTANT: The last sentence is NOT true unless we change the definition of LBS cut to $\phi(S)\leq \alpha$ (and not just ``<'').}
\danupon{Things are too involved and equations/definitions are not referred to properly. This made it hard to verify.}
\end{comment}

%% file: pruning.tex
\section{Expander Pruning\label{sec:pruning}}

\danupon{Why this name? The description of the result doesn't justify the name.}

The main result of this section is the \emph{dynamic expander pruning}
algorithm. This algorithm was a key tool introduced by Wulff-Nilsen
\cite[Theorem 5]{Wulff-Nilsen16a} for obtaining his dynamic $\msf$
algorithm. We significantly improve his dynamic expander pruning algorithm
which is randomized and has $n^{0.5-\epsilon_{0}}$ update time for
some constant $\epsilon_{0}>0$. Our algorithm is \emph{deterministic}
and has $n^{o(1)}$ update time. Although the algorithm is deterministic,
our final dynamic $\msf$ algorithm is randomized because there are
other components that need randomization. 

First we state the precise statement (explanations follow below). 

\danupon{TO DO FOR THATCHAPHOL: Dynamic should be ``Decremental''.}

%
%\begin{thm}
%[Dynamic Expander Pruning]\label{thm:dynamic pruning subpoly}Let
%$\alpha_{0}=1/n^{\epsilon}$ for any $\epsilon=o(1)$. There is an
%algorithm $\cA$ that can do the following:
%\begin{itemize}
%\item $\cA$ is given an $n$-node $m$-edge graph $G_{0}=(V,E)$ with maximum
%degree $3$ undergoing the sequence of edge deletions of length $T=O(\alpha_{0}^{2}m)$.
%\item Given the $\tau$-th update, $\cA$ takes $n^{O(\log\log\frac{1}{\epsilon}/\log\frac{1}{\epsilon})}=n^{o(1)}$
%time. Then, $\cA$ either reports $\phi(G_{0})<\alpha_{0}$, or updates
%the pruning set $P$ to $P_{\tau}$ where $P_{\tau-1}\subseteq P_{\tau}\subseteq V$.
%\item If $\phi(G_{0})\ge\alpha_{0}$ then, for all $\tau$, there exists
%$W_{\tau}\subseteq P_{\tau}$ where $G_{\tau}[V-W_{\tau}]$ is connected.
%\end{itemize}
%\end{thm}

\danupon{Do we have to require $G$ to be connected.}
%
%\begin{thm}
%	[Dynamic Expander Pruning]\label{thm:dynamic pruning subpoly} Consider any $\epsilon(n)=o(1)$, and let $\alpha_{0}(n)=1/n^{\epsilon(n)}$. There is a dynamic
%	algorithm $\cA$ that can maintain a set of nodes $P$ for a graph $G$ undergoing $T=O(m\alpha_{0}^{2}(n))$ edge deletions as follows. Let $G_\tau$ and $P_\tau$ be the graph $G$ and set $P$ after the $\tau^{th}$ deletion, respectively.
%	\begin{itemize}[noitemsep]
%		\item Initially, in $O(1)$ time $\cA$ sets $P_0=\emptyset$ and takes as input an $n$-node $m$-edge graph $G_0=(V,E)$ with maximum
%		degree $3$. 
%		\item After the $\tau^{th}$ deletion, $\cA$ takes $n^{O(\log\log\frac{1}{\epsilon(n)}/\log\frac{1}{\epsilon(n)})}=n^{o(1)}$
%		time to either (i) report that $\phi(G_0)<\alpha_{0}(n)$ or (ii) report nodes to be added to $P_{\tau-1}$ to form $P_\tau$ where
%		\begin{align}
%		\exists W_{\tau}\subseteq P_{\tau}  \mbox{ s.t. $G_{\tau}[V-W_{\tau}]$ is connected.}\label{eq:P property}
%		\end{align}
%	\end{itemize}
%\end{thm}

\begin{thm}
	[Dynamic Expander Pruning]\label{thm:dynamic pruning subpoly} Consider any $\epsilon(n)=o(1)$, and let $\alpha_{0}(n)=1/n^{\epsilon(n)}$. There is a dynamic
	algorithm $\cA$ that can maintain a set of nodes $P$ for a graph $G$ undergoing $T=O(m\alpha_{0}^{2}(n))$ edge deletions as follows. Let $G_\tau$ and $P_\tau$ be the graph $G$ and set $P$ after the $\tau^{th}$ deletion, respectively.
	\begin{itemize}[noitemsep]
		\item Initially, in $O(1)$ time $\cA$ sets $P_0=\emptyset$ and takes as input an $n$-node $m$-edge graph $G_0=(V,E)$ with maximum
		degree $3$. 
		\item After the $\tau^{th}$ deletion, $\cA$ takes $n^{O(\log\log\frac{1}{\epsilon(n)}/\log\frac{1}{\epsilon(n)})}=n^{o(1)}$
		time to report nodes to be added to $P_{\tau-1}$ to form $P_\tau$ where, if $\phi(G_0)\geq \alpha_{0}(n)$, then
		\begin{align}
		\exists W_{\tau}\subseteq P_{\tau}  \mbox{ s.t. $G_{\tau}[V-W_{\tau}]$ is connected.}\label{eq:P property}
		\end{align}
	\end{itemize}
\end{thm}

\danupon{TO DO for Thatchaphol: The theorem has changed.}

%Recall that the algorithm in \Cref{thm:dynamic pruning subpoly} takes input graph $G$ in the local manner (see \Cref{rem:local input graph}).  

\danupon{Should we give name to \Cref{eq:P property}?}

\danupon{TO THINK: The name pruning still doesn't make much sense to me. It's more natural to me to imagine that the algorithm gradually marks nodes to make sure that unmarked nodes are in the same connected component. Pruning usually gives a sense of removing nodes.}

The goal of our algorithm is to gradually mark nodes in a graph $G=(V, E)$ so that at all time -- as edges in $G$ are deleted -- all nodes that are not yet marked are in the same connected component in $G$. 
In other words, the algorithm maintains a set $P$ of (marked) nodes, called \emph{pruning set},  such that there exists $W\subseteq P$ where $G[V-W]$ is connected (thus \Cref{eq:P property}). 
%
%sets $P_0\subseteq P_1\subseteq \ldots$ of (marked) nodes, called \emph{pruning sets}, so that after \Cref{eq:P property} is satisfied. 
% every node {\em not} in $P$ are in the same connected component\danupon{This is a simpler way to define $P$, right?}; i.e. \Cref{eq:P property} is satisfied. 
%
In
our application in \ref{sec:Dynamic MSF}, we will delete edges incident
to $P$ from the graph, hence the name pruning set.

Recall that the algorithm takes an input graph in the local manner, as noted in \Cref{rem:local input graph}, thus taking $n^{o(1)}$ time. Observe that if we can set $P=V$ from the beginning, the problem becomes trivial.
The challenge here is that we must set $P=\emptyset$ in the initial step, and thus must grow $P$ smartly and quickly (in $n^{o(1)}$ time) after each deletion so that \Cref{eq:P property} remains satisfied. 
%
%and does not have time to do anything -- especially to read the entire input graph -- in the initial step
%
% is thus to grow $P$ carefully and quickly (in $n^{o(1)})$ time) after each deletion so that \Cref{eq:P property} remains satisfied. 

Observe further that this task is not possible to achieve in general: if the first deletion cuts $G$ into two large connected components, then $P$ has to grow tremendously to contain one of these components, which is impossible to do in $n^{o(1)}$ time. Because of this, our algorithm is guaranteed to work only if the initial graph has {\em high enough expansion}; in particular, an initial expansion of $\alpha_{0}(n)$ as in \Cref{thm:dynamic pruning subpoly} suffices for us.

\paragraph{Organization.}
The rest of this section is for proving \ref{thm:dynamic pruning subpoly}.
The key tool is an algorithm called the \emph{one-shot expander pruning},
which was also the key tool in Nanongkai and Saranurak~\cite{NanongkaiS16} for obtaining
their Las Vegas dynamic $\sf$ algorithm. We show an improved version
of this algorithm in \ref{sec:static pruning} using the faster LBS
cut algorithm we developed in \ref{sec:LBS cut}. In \ref{sec:dynamic pruning},
we show how to use several instances of the one-shot expander pruning
algorithm to obtain the dynamic one and prove \ref{thm:dynamic pruning subpoly}.

\subsection{One-shot Expander Pruning\label{sec:static pruning}}

In the following, we show the \emph{one-shot expander pruning} algorithm
which is significantly improved from \cite{NanongkaiS16}. In words,
the one-shot expander pruning algorithm is different from the dynamic
one from \ref{thm:dynamic pruning subpoly} in two aspects: 1) it
only handles a single batch of edge deletions, instead of a sequence
of edge deletions, and so only outputs a pruning set $P$ once, and
2) the pruning set $P$ has a stronger guarantee than the pruning
set for dynamic one as follows: $P$ does not only contains all nodes
in the cuts that are completely separated from the graphs (i.e. the
separated connected components) but $P$ contains all nodes in the
cuts that have low conductance. Moreover, $P$ contains \emph{exactly
}those nodes and hence the complement $G[V-P]$ has high conductance.
For the dynamic expander pruning algorithm, we only have that there
is some $W\subseteq P$ where $G[V-W]$ is connected.

The theorem below shows the precise statement. Below, we think of
$G_{b}=(V,E\cup D)$ as the graph before the deletions, and $G=G_{b}-D$
as the graph after deleting $D$. In \cite{NanongkaiS16}, Nanongkai
and Saranurak show this algorithm where the dependency on $D$ is
$\sim D^{1.5+\delta}$, while in our algorithm the dependency of $D$
is $\sim D^{1+\delta}$.
\begin{thm}
[One-shot Expander Pruning]\label{thm:local pruning}There is an
algorithm $\cA$ that can do the following:
\begin{itemize}
\item $\cA$ is given $G$,$D,\alpha_{b},\delta$ as inputs: $G=(V,E)$
is an $n$-node $m$-edge graph with maximum degree $\Delta$, $\alpha_{b}$
is a conductance parameter, $\delta\in(0,1)$ is a parameter, and
$D$ is a set of edges where $D\cap E=\emptyset$ where $|D|=O(\alpha_{b}^{2}m/\Delta)$.
Let $G_{b}=(V,E\cup D)$. 
\item Then, in time $\overline{t}=\tilde{O}(\frac{\Delta|D|^{1+\delta}}{\delta\alpha_{b}^{6+\delta}})$,
$\cA$ either reports $\phi(G_{b})<\alpha_{b}$, or outputs a \emph{pruning
set }$P\subset V$. Moreover, if $\phi(G_{b})\ge\alpha_{b}$, then
we have

\begin{itemize}
\item $vol_{G}(P)\le2|D|/\alpha_{b}$, and
\item a pruned graph $H=G[V-P]$ has high conductance: $\phi(H)\ge\alpha=\Omega(\alpha_{b}{}^{2/\delta})$.
\end{itemize}
\end{itemize}
\end{thm}
We call $\overline{t}$ the \emph{time limit }and $\alpha$ the \emph{conductance
guarantee} of $\cA$. If we do not care about the time limit, then
there is the following algorithm gives a very good conductance guarantee:
just find the cut $C^{*}$ of conductance at most $\alpha_{b}/10$
that have maximum volume and output $P=C^{*}$. If $vol(P)>2|D|/\alpha_{b}$,
then report $\phi(G_{b})<\alpha_{b}$. Otherwise, we must have $\phi(G[V-P])=\Omega(\alpha_{b})$.
This can be shown using the result by Spielman and Teng \cite[Lemma 7.2]{SpielmanT11}.
However, computing the optimum cut $C^{*}$ is NP-hard. 

In \cite{NanongkaiS16}, they implicitly showed that using only the
LBS cut algorithm, which is basically an algorithm for finding a cut
similar to $C^{*}$ but the guarantee is only \emph{approximately}
and \emph{locally}, one can quickly obtain the one-shot expander pruning
algorithm whose conductance guarantee is not too bad. Below, we explicitly
state the reduction in \cite{NanongkaiS16}. See \ref{sec:LBS cut to Prune}
for the proof\footnote{Strictly speaking, in \cite{NanongkaiS16}, they use LBS cut algorithm
to obtain the \emph{local expansion decomposition }algorithm which
has slight stronger guarantee. Actually, even in \cite{NanongkaiS16},
they only need the one-shot expander pruning algorithm. So the proof
of \ref{lem:reduc to LBS} is simpler than the similar one in \cite{NanongkaiS16}.
Moreover, the reduction give a deterministic one-shot expander pruning
algorithm, but in \cite{NanongkaiS16}, they obtain a randomized local
expansion decomposition\emph{ }algorithm.}. 
\begin{lem}
[\cite{NanongkaiS16}]\label{lem:reduc to LBS}Suppose there is a
$(c_{size}(\sigma),c_{con}(\sigma))$-approximate LBS cut algorithm
with running time $t_{LSB}(n,vol(A),\alpha,\sigma)$ when given $(G,A,\sigma,\alpha)$
as inputs where $G=(V,E)$ is an $n$-node graph, $A\subset V$ is
a set of nodes, $\sigma$ is an overlapping parameter, and $\alpha$
is a conductance parameter. Then, there is a one-shot expander pruning
algorithm with input $(G,D,\alpha_{b},\delta)$ that has \emph{time
limit} 
\[
\overline{t}=O((\frac{|D|}{\alpha_{b}})^{\delta}\cdot\frac{c_{size}(\alpha_{b}/2)}{\delta}\cdot t_{LSB}(n,\frac{\Delta|D|}{\alpha_{b}},\alpha_{b},\alpha_{b}))
\]
and \emph{conductance guarantee} 
\[
\alpha=\frac{\alpha_{b}}{5c_{con}(\alpha_{b}/2)^{1/\delta-1}}.
\]

\end{lem}
Having the above lemma and our new LBS cut algorithm from \ref{sec:LBS cut},
we conclude:
\begin{proof}
[Proof of \ref{thm:local pruning}]From \ref{thm:LBS cut alg} we
have that $t_{LSB}(n,vol(A),\alpha,\sigma)=\tilde{O}(\frac{vol(A)}{\alpha\sigma^{2}})$
and $c_{size}(\sigma)=O(1/\sigma^{2})$ and $c_{con}(\sigma)=O(1/\sigma^{2})$.
So 
\begin{align*}
t_{LSB}(n,\frac{\Delta|D|}{\alpha_{b}},\alpha_{b},\alpha_{b}) & =\tilde{O}(\frac{\Delta|D|}{\alpha_{b}^{4}})
\end{align*}
and hence 
\[
\overline{t}=O((\frac{|D|}{\alpha_{b}})^{\delta}\cdot\frac{c_{size}(\alpha_{b}/2)}{\delta}\cdot\frac{\Delta|D|}{\alpha_{b}^{4}})=\tilde{O}(\frac{\Delta|D|^{1+\delta}}{\delta\alpha_{b}^{6+\delta}}).
\]
We also have $\alpha=\frac{\alpha_{b}}{5c_{con}(\alpha_{b}/2)^{1/\delta-1}}=\Omega(\alpha_{b}\cdot(\alpha_{b}^{2})^{1/\delta-1})=\Omega(\alpha_{b}{}^{2/\delta})$. 
\end{proof}

\subsection{Dynamic Expander Pruning\label{sec:dynamic pruning}}

In this section, we exploit the one-shot expander pruning algorithm
from \ref{sec:static pruning}. To prove \ref{thm:dynamic pruning subpoly},
it is more convenience to prove the more general statement as follows:
\begin{lem}
\label{thm:dynamic pruning}There is an algorithm $\cA$ that can
do the following:
\begin{itemize}
\item $\cA$ is given $G_{0},\alpha_{0},\ell$ as inputs: $G_{0}=(V,E)$
is an $n$-node $m$-edge graph with maximum degree $\Delta$, and
$\alpha_{0}=\frac{1}{n^{\epsilon}}$ and $\ell$ are parameters. Let
$P_{0}=\emptyset$.
\item Then $G_{0}$ undergoes the sequence of edge deletions of length $T=O(\alpha_{0}^{2}m/\Delta)$.
\item Given the $\tau$-th update, $\cA$ takes $\tilde{O}(\ell^{2}\Delta n^{O(1/\ell+\epsilon\ell^{\ell})})$
time. Then, $\cA$ either reports $\phi(G_{0})<\alpha_{0}$ and halt,
or $\cA$ updates the pruning set $P$ to $P_{\tau}$ where $P_{\tau-1}\subseteq P_{\tau}\subseteq V$.
\item If $\phi(G_{0})\ge\alpha_{0}$ then, for all $\tau$, there exists
$W_{\tau}\subseteq P_{\tau}$ where $G_{\tau}[V-W_{\tau}]$ is connected.
\end{itemize}
\end{lem}
From \ref{thm:dynamic pruning}, we immediately obtain \ref{thm:dynamic pruning subpoly}
by choosing the right parameters.
\begin{proof}
[Proof of \ref{thm:dynamic pruning subpoly}]We set $\ell=\frac{\log\frac{1}{\epsilon}}{2\log\log\frac{1}{\epsilon}}$,
so that $\ell^{\ell}=O(\frac{1}{\epsilon^{1/2}})$. Hence, 
\begin{align*}
n^{O(1/\ell+\epsilon\ell^{\ell})} & =n^{O(\log\log\frac{1}{\epsilon}/\log\frac{1}{\epsilon}+\epsilon^{1/2})}\\
 & =n^{O(\log\log\frac{1}{\epsilon}/\log\frac{1}{\epsilon})}=n^{o(1)}
\end{align*}
when $\epsilon=o(1)$. We apply \ref{thm:dynamic pruning} with this
parameters $\ell$ and $\alpha_{0}=\frac{1}{n^{\epsilon}}$ and we
are done.
\end{proof}
The rest of this section is for proving \ref{thm:dynamic pruning}.

\subsubsection{The Algorithm}

Let $G_{0},\alpha_{0},\ell$ be the inputs for the algorithm for \ref{thm:dynamic pruning}.
To roughly describe the algorithm, there will be $\ell+1$ \emph{levels}
and, in each level, this algorithm repeated calls an instance of the
one-shot expander pruning algorithm from \ref{thm:local pruning}.
In the deeper level $i$ (large $i$), we call it more frequently
but the size of the set $D$ of edges is smaller. We describe the
details and introduce some notations below. 

We fix $\delta=2/\ell$. For any $n,\delta,\alpha$, let $f_{n,\delta}(\alpha)=(c_{0}\alpha)^{2/\delta}$
be the conductance guarantee of the remaining graph from \ref{thm:local pruning}
where $c_{0}$ is some constant. We define $\alpha_{i}=f_{n,\delta}(\alpha_{i-1})$
for each \emph{level} $1\le i\le\ell+1$. Note that, we have a rough
bound $\alpha_{i}=\Omega((c_{0}\alpha_{0})^{\ell^{i}})$ for all $i$..
Given any input $G'=(V',E')$,$D',\alpha',\delta$, we denote $(X',P')=Prune_{\alpha'}(G',D')$
as the output of the one-shot expander pruning algorithm from \ref{sec:static pruning}
where $P'\subset V$ is the outputted pruning set and $X'=G'[V'-P']$
is the pruned graph. Note that we omit writing $\delta$ in $Prune_{\alpha'}(G',D')$
because $\delta=2/\ell$ is always fixed. 

For convenience, we say that a graph $X$ \emph{is an induced $\alpha$-expander
from time $\tau$ }if $X=G_{\tau}[U]$ for some $U\subset V$ and
$X$ has conductance at least $\alpha$, i.e. $\phi(X)\ge\alpha$.
For any time period $[\tau,\tau']$, we denote $D_{[\tau,\tau']}\subset E$
a set of edges to be deleted from time $\tau$ to $\tau'$. Observe
the following fact which follows by the definitions and \ref{thm:local pruning}:
\begin{fact}
\label{fact:next induced expander}Suppose that $X$ is an induced
$\alpha_{i}$-expander from time $\tau$ and $(X',P')=Prune_{\alpha_{i}}(X,D_{[\tau+1,\tau']})$.
Then $X'$ is an induced $\alpha_{i+1}$-expander from time $\tau'$. 
\end{fact}
To maintain the pruning set $P$, we will additionally maintain a
level-$i$ graphs $X^{i}$ and a level-$i$ pruning set $P^{i}$ for
each level $1\le i\le\ell+1$. Let $X_{\tau}^{i}$ and $P_{\tau}^{i}$
be $X^{i}$ and $P^{i}$ at time $\tau$ respectively. For each level
$i$, we initially set $X_{0}^{i}=G_{0}$ and $P_{0}^{i}=\emptyset$,
and $X^{i}$ and $P^{i}$ will be updated periodically for every $d_{i}$
time steps where $d_{i}=n^{1-i/\ell}$ for $i\le\ell$ and $d_{\ell+1}=1$.
Note that $d_{\ell}=1$. In particular, this means:
\begin{fact}
\label{fact:periodical update}For any number $k\ge0$ and time $\tau\in[kd_{i},(k+1)d_{i})$,
we have $X_{\tau}^{i}=X_{kd_{i}}^{i}$ and $P_{\tau}^{i}=P_{kd_{i}}^{i}$.
\end{fact}
In each time step, we spend time in each of the $\ell+1$ levels.
See the precise description on each level in \ref{alg:dynamic pruing}.
At any time, for any $i$, whenever we call $Prune_{\alpha_{i}}(X^{i},\cdot)$
and it report that $\phi(X^{i})<\alpha_{i}$, our algorithm will report
that $\phi(G_{0})<\alpha_{0}$ and halt.

\begin{algorithm}[H]
\caption{\label{alg:dynamic pruing}Dynamic expander pruning algorithm}

\begin{description}
\item [{Initialization:}] $X^{0}=G_{0}$ and $P_{0}=\emptyset$. For each
$1\le i\le\ell$, $d_{i}=n^{1-i/\ell}$, $X_{0}^{i}=G_{0}$, $P_{0}^{i}=\emptyset$.\end{description}
\begin{enumerate}
\item For each level $1\le i\le\ell$, for each number $k_{i}\ge0$, in
time period $[k_{i}d_{i}+1,(k_{i}+1)d_{i}]$:

\begin{enumerate}
\item Let $k_{i-1}$ be such that $k_{i-1}d_{i-1}<k_{i}d_{i}+1\le(k_{i-1}+1)d_{i-1}$.
\item During the period, distribute evenly the work to:

\begin{enumerate}
\item Update $(X_{(k_{i}+1)d_{i}}^{i},P_{(k_{i}+1)d_{i}}^{i})=Prune_{\alpha_{i-1}}(X_{k_{i-1}d_{i-1}}^{i-1},D_{[\min\{1,(k_{i-1}-1)d_{i-1}+1\},k_{i}d_{i}]}).$
\item Include $P_{(k_{i}+1)d_{i}}^{i}$ into the pruning set $P$.
\end{enumerate}
\end{enumerate}
\item For level $\ell+1$, at time $\tau$:

\begin{enumerate}
\item Update $(X_{\tau}^{\ell+1},P_{\tau}^{\ell+1})=Prune_{\alpha_{\ell}}(X_{\tau}^{\ell},D_{[\tau,\tau]})$.
\item Include $P_{\tau}^{\ell+1}$ into the pruning set $P$.\end{enumerate}
\end{enumerate}
\end{algorithm}

\subsubsection{Analysis}
\begin{lem}
\label{prop:induced expander}Suppose $\phi(G_{0})\ge\alpha_{0}$.
For any $1\le i\le\ell$ and $k_{i}$, $X_{k_{i}d_{i}}^{i}$ is an
induced $\alpha_{i}$-expander from time $\max\{0,(k_{i}-1)d_{i}\}$.\end{lem}
\begin{proof}
When $k_{i}=0$, this is trivial. We now prove the claim for $k_{i}>0$
by induction on $i$. Let $k_{i-1}$ be a number from Step 1.a such
that $k_{i-1}d_{i-1}<k_{i}d_{i}+1\le(k_{i-1}+1)d_{i-1}$. $X_{k_{i-1}d_{i-1}}^{i-1}$
is an induced $\alpha_{i-1}$-expander from time $(k_{i-1}-1)d_{i-1}$
by induction hypothesis. By Step 1.b.ii and \ref{fact:next induced expander},
we have $X_{(k_{i}+1)d_{i}}^{i}$ is an induced $\alpha_{i}$-expander
from time $k_{i}d_{i}$. By translating back the time by $d_{i}$
steps, we can conclude the claim.\end{proof}
\begin{lem}
\label{lem:induced expander l+1}Suppose $\phi(G_{0})\ge\alpha_{0}$.
For any time $\tau$, $X_{\tau}^{\ell+1}$ is an induced $\alpha_{\ell+1}$-expander
from time $\tau$.\end{lem}
\begin{proof}
Note that $d_{\ell}=1$. By \ref{prop:induced expander} when $i=\ell$,
we have that after the $\tau$-th update, $X_{\tau}^{\ell}$ is an
induced $\alpha_{\ell}$-expander from time $\tau-1$. By Step 2.a
and \ref{fact:next induced expander}, $X_{\ell+1}$ is an induced
$\alpha_{\ell+1}$-expander from time $\tau$.\end{proof}
\begin{lem}
\label{lem:dyn pruning never fail}Suppose $\phi(G_{0})\ge\alpha_{0}$.
Then the algorithm never reports that $\phi(G_{0})<\alpha_{0}$.\end{lem}
\begin{proof}
Recall that the algorithm will report $\phi(G_{0})<\alpha_{0}$ only
when, for some $i$ and $j$, the call of $Prune_{\alpha_{i}}(X_{j}^{i},\cdot)$
reports that $\phi(X_{j}^{i})<\alpha_{i}$. By \ref{fact:periodical update},
\ref{prop:induced expander} and \ref{lem:induced expander l+1},
if $\phi(G_{0})\ge\alpha_{0}$, then $X_{j}^{i}$ is an induced $\alpha_{j}$-expander
for all $j$ and in particular $\phi(X_{j}^{i})\ge\alpha_{i}$. So
the algorithm never reports that $\phi(G_{0})<\alpha_{0}$.
\end{proof}
The following proposition is easy to see by Step 1.b.ii and Step 2.b
in \ref{alg:dynamic pruing}:
\begin{prop}
For any $1\le i\le\ell+1$ and $k_{i}$, $P_{k_{i}d_{i}}^{i}\subseteq P_{k_{i}d_{i}}$.\label{prop:include pruning set}\end{prop}
\begin{lem}
For any $\tau$, $V-V(X_{\tau}^{\ell+1})\subseteq P_{\tau}$.\label{lem:outside expander in pruning set}\end{lem}
\begin{proof}
It suffices to prove that $V(X_{\tau}^{\ell+1})\supseteq V-P_{\tau}$.
Let $k_{\ell}=\tau$. We have $V(X_{\tau}^{\ell+1})=V(X_{k_{\ell}d_{\ell}}^{\ell})-P_{\tau}^{\ell+1}$.
By Step 1.a and 1.b.i, we can write

\begin{align*}
V(X_{k_{\ell}d_{\ell}}^{\ell}) & =V(X_{(k_{\ell-1}-1)d_{\ell-1}}^{\ell-1})-P_{k_{\ell}d_{\ell}}^{\ell}\\
 & =V(X_{(k_{\ell-2}-2)d_{i-2}}^{\ell-2})-P_{(k_{\ell-1}-1)d_{\ell-1}}^{\ell-1}-P_{k_{\ell}d_{\ell}}^{\ell}\\
 & \vdots\\
 & =V(X^{0})-\bigcup_{0\le i<\ell}P_{(k_{\ell-i}-i)d_{\ell-i}}^{\ell-i}\\
 & =V-\bigcup_{0\le i<\ell}P_{(k_{\ell-i}-i)d_{\ell-i}}^{\ell-i}
\end{align*}
where $k_{i-i}$ is the largest number where $(k_{\ell-i}-1)d_{\ell-i}\le k_{\ell-i+1}d_{\ell-i+1}$.
(For convenience, let $X_{j}^{i}=X_{0}^{i}=G_{0}$ and $P_{j}^{i}=\emptyset$
for any negative $j<0$.) Observe that $(k_{\ell-i}-i)d_{\ell-i}\le k_{\ell}d_{\ell}$
for all $i\le\ell$. Therefore, by \ref{prop:include pruning set},
we have $P_{(k_{\ell-i}-i)d_{\ell-i}}^{\ell-i}\subseteq P_{(k_{\ell-i}-i)d_{\ell-i}}\subseteq P_{k_{\ell}d_{\ell}}$.
This implies that $V(X_{k_{\ell}d_{\ell}}^{\ell})\supseteq V-P_{k_{\ell}d_{\ell}}=V-P_{\tau}.$
We conclude $V(X_{\tau}^{\ell+1})=V(X_{k_{\ell}d_{\ell}}^{\ell})-P_{\tau}^{\ell+1}\supseteq V-P_{\tau}$,
by \ref{prop:include pruning set} again.
\end{proof}
Now, we can conclude the correctness of the algorithm:
\begin{cor}
After given the $\tau$-th update, the algorithm either correctly
reports $\phi(G_{0})<\alpha_{0}$ and halt, or updates the pruning
set\emph{ $P$ }to\emph{ $P_{\tau}$ where }$P_{\tau-1}\subseteq P_{\tau}\subseteq V$.
If $\phi(G_{0})\ge\alpha_{0}$ then, for all $\tau$, there exists
$W_{\tau}\subseteq P_{\tau}$ where $G_{\tau}[V-W_{\tau}]$ is connected.\label{cor:dyn pruning correct}\end{cor}
\begin{proof}
By \ref{lem:dyn pruning never fail}, we have that the algorithm either
\emph{correctly} reports that $\phi(G_{0})<\alpha_{0}$ and halt,
or updates the pruning set $P$ to $P_{\tau}$ and $P_{\tau-1}\subseteq P_{\tau}\subseteq V$
because we only grow $P$ through time. Let $W_{\tau}=V-V(X_{\tau}^{\ell+1})$.
By \ref{lem:outside expander in pruning set}, we have $W_{\tau}\subseteq P_{\tau}$
and also $G_{\tau}[V-W_{\tau}]=X_{\tau}^{\ell+1}$ which is an $\alpha_{\ell+1}$-expander
by \ref{lem:induced expander l+1}, when $\phi(G_{0})\ge\alpha_{0}$.
In particular $G_{\tau}[V-W_{\tau}]$ is connected.
\end{proof}
Finally, we analyze the running time.
\begin{lem}
For each update, the algorithm takes $\tilde{O}(\ell^{2}\Delta n^{O(1/\ell+\epsilon\ell^{\ell})})$
time.\label{lem:dyn pruning run time}\end{lem}
\begin{proof}
We separately analyze the running time for each level. At level $i$,
in time period $[k_{i}d_{i}+1,(k_{i}+1)d_{i}]$, the bottleneck is
clearly for calling $Prune_{\alpha_{i-1}}(X_{k_{i-1}d_{i-1}}^{i-1},D_{[\min\{1,(k_{i-1}-1)d_{i-1}+1\},k_{i}d_{i}]})$.
Note that $|D_{[\min\{1,(k_{i-1}-1)d_{i-1}+1\},k_{i}d_{i}]}|\le d_{i-1}$.
By \ref{thm:local pruning}, this takes time $\overline{t}_{i}=\tilde{O}(\frac{\Delta\cdot d_{i-1}^{^{1+\delta}}}{\delta\alpha_{i-1}^{6+\delta}})$.
Since we distribute the work evenly in the period, this takes $\overline{t}_{i}/d_{i}=\tilde{O}(\frac{\Delta\cdot d_{i-1}^{^{1+\delta}}}{\delta\alpha_{i-1}^{6+\delta}d_{i}})=\tilde{O}(\frac{\Delta n^{3/\ell}}{\delta\alpha_{i-1}^{6+\delta}})$
per step. This is because $\delta=2/\ell$, $d_{i}=n^{1-i/\ell}$,
and so 
\[
d_{i-1}^{^{1+\delta}}/d_{i}\le n^{\delta}\cdot d_{i-1}/d_{i}=n^{\delta+1/\ell}=n^{3/\ell}.
\]
Since there are $\ell+1$ level, this takes in total per time step
\begin{align*}
\sum_{i\le\ell+1}\tilde{O}(\frac{\Delta n^{3/\ell}}{\delta\alpha_{i-1}^{6+2/\ell}}) & =\tilde{O}(\frac{\ell^{2}\Delta n^{3/\ell}}{\alpha_{\ell}^{8}}) & \mbox{by }\ell\ge1\mbox{ and }\mbox{\ensuremath{\delta}=2/\ensuremath{\ell}}\\
 & =\tilde{O}(\frac{\ell^{2}\Delta n^{3/\ell}}{((c_{0}\alpha_{0})^{\ell^{\ell}})^{8}}) & \mbox{by }\alpha_{i}=\Omega((c_{0}\alpha_{0})^{\ell^{i}})\\
 & =\tilde{O}(\ell^{2}\Delta n^{O(1/\ell+\epsilon\ell^{\ell})}).
\end{align*}
 
\end{proof}
By \ref{cor:dyn pruning correct} and \ref{lem:dyn pruning run time},
this concludes the proof of \ref{thm:dynamic pruning}, and hence
\ref{thm:dynamic pruning subpoly}.

%% file: pruning_lasvegas.tex
\section{Pruning on Arbitrary Graphs}
\label{sec:pruning_lasvegas}

In \ref{thm:dynamic pruning subpoly}, we show a fast deterministic
algorithm that guarantees connectivity of the pruned graph $G[V-W]$ only when
an initial graph is an expander. If the initial graph is not an expander,
then there is no guarantee at all. With a simple modification, in this section, 
we will show a fast randomized algorithm for an arbitrary initial graph that either
outputs the desired pruning set or reports failure. Moreover, if the
the initial graph is an expander, then it never fails with high probability. 

This section is needed in order to make our final algorithm Las Vegas.
If we only want a Monte Carlo algorithm, then it is enough to use
\ref{thm:dynamic pruning subpoly} when we combine every component
together in \ref{sec:Dynamic MSF}.

\begin{thm}
	\label{thm:pruning detect failure}Consider any $\epsilon(n)=o(1)$,
	and let $\alpha_{0}(n)=1/n^{\epsilon(n)}$. There is a dynamic algorithm
	$\cA$ that can maintain a set of nodes $P$ for a graph $G$ undergoing
	$T=O(m\alpha_{0}^{2}(n))$ edge deletions as follows. Let $G_{\tau}$
	and $P_{\tau}$ be the graph $G$ and set $P$ after the $\tau^{th}$
	deletion, respectively. 
	\begin{itemize}
		\item Initially, in $\tilde{O}(n \log\frac{1}{p})$ time $\cA$ sets $P_{0}=\emptyset$
		and takes as input an $n$-node $m$-edge graph $G_{0}=(V,E)$ with
		maximum degree $3$. 
		\item After the $\tau^{th}$ deletion, $\cA$ takes $O(n^{O(\log\log\frac{1}{\epsilon(n)}/\log\frac{1}{\epsilon(n)})}\log\frac{1}{p})=O(n^{o(1)}\log\frac{1}{p})$
		time to either 1) report nodes to be added to $P_{\tau-1}$ to form
		$P_{\tau}$ where 
		\[
		\exists W_{\tau}\subseteq P_{\tau}\mbox{ s.t. \ensuremath{G_{\tau}[V-W_{\tau}]} is connected}
		\]
		or 2) reports failure. If $\phi(G_{0})\geq\alpha_{0}(n)$, then $\cA$
		never fails with probability $1-p$. 
	\end{itemize}
\end{thm}
\begin{proof}
	Let $\cA_{\sf}$ be an instance of the Monte Carlo dynamic spanning forest by Kapron
	et al. \cite{KapronKM13} that guarantees to maintain a correct spanning forest with probability $1-p$ 
	when the update sequence has length $\poly(n)$. 
	Let $\cA_{prune}$ be an instance of the dynamic expander
	pruning from \ref{thm:dynamic pruning subpoly}. Given $\cA_{\sf}$
	and $\cA_{prune}$, the algorithm is very simple.
	
	The preprocessing algorithm is just to initialize $\cA_{\sf}$ and
	$\cA_{prune}$ on the graph $G$ in time $\tilde{O}(n\log\frac{1}{p})$ and $O(1)$ respectively. Then, given
	a sequence of edge deletions, $\cA_{prune}$ maintains the pruning
	set $P$ and $\cA_{\sf}$ maintains a spanning forest $F$ of the
	current graph $G$. We say that $F$ \emph{spans} $V-P$ iff all the
	nodes in $V-P$ are in the same connected component of $F$. The update
	algorithm is just to check if $F$ spans $V-P$. If yes, then we report
	$P$ as the desired pruning set. If no, then we report failure. 
	
	Now, we analyze the update algorithm. The update time of $\cA_{\sf}$
	is $O(\log^{O(1)}n\cdot\log\frac{1}{p})$. Moreover, we can check
	if $F$ spans $V-P$ easily by implementing ET-tree on $F$. This
	takes $O(\log n)$ update time. 
	The time used by $\cA_{prune}$ is $O(n^{O(\log\log\frac{1}{\epsilon(n)}/\log\frac{1}{\epsilon(n)})})$.
	So the total update time is at most $O(n^{O(\log\log\frac{1}{\epsilon(n)}/\log\frac{1}{\epsilon(n)})} \log\frac{1}{p})$

	It remains to show the correctness. 
	If the algorithm does not fail, then $F$ spans $V-P$. Hence, there
	is $W\subset P$ where $G[V-W]$ is connected. Finally, if $\phi(G_{0})\geq\alpha_{0}(n)$,
	then by \ref{thm:dynamic pruning subpoly} we have that there is $W\subset P$
	where $G[V-W]$ is connected. Then, $F$ must span $V-P\subseteq V-W$
	with high probability, because $F$ is a spanning forest of $G$ with
	high probability by the guarantee in \cite{KapronKM13}.
\end{proof}

%% file: contraction.tex
\section{Reduction from Graphs with Few Non-tree Edges Undergoing Batch Insertions\label{sec:contraction}}

In this section, we show the following crucial reduction:
\begin{thm}
\label{thm:reduc restricted dec}Suppose there is a decremental $\msf$
algorithm $\cA$ for any $m'$-edge graph with max degree 3 undergoing
a sequence of edge deletions of length $T(m')$, and $\cA$ has $t_{pre}(m',p)$
preprocessing time and $t_{u}(m',p)$ worst-case update time with
probability $1-p$. 

Then, for any numbers $B$ and $k$ where $15k\le m'$, there is a
fully dynamic $\msf$ algorithm $\cB$ for any $m$-edge graph with
at most $k$ non-tree edges such that $\cB$ can:
\begin{itemize}
\item preprocess the input graph in time 
\[
t'{}_{pre}(m,k,B,p)=t_{pre}(15k,p')+O(m\log^{2}m),
\]

\item handle a batch of $B$ edge insertions or a single edge deletion in
time: 
\[
t'_{u}(m,k,B,p)=O(\frac{B\log k}{k}\cdot t{}_{pre}(15k,p')+B\log^{2}m+\frac{k\log k}{T(k)}+\log k\cdot t{}_{u}(15k,p')),
\]
 
\end{itemize}
where $p'=\Theta(p/\log k)$ and the time guarantee for each operation
holds with probability $1-p$.
\end{thm}
The proof of \ref{thm:reduc restricted dec} is by extending the reduction
by Wulff-Nilsen \cite{Wulff-Nilsen16a} in two ways. First, the resulting
algorithm is more efficient when there are few non-tree edges. Second,
the resulting algorithm can also quickly handle a batch of edge insertions. 

Although, the extension of the reduction is straightforward and also
uses the same ``contraction'' technique by Henzinger and King \cite{HenzingerK97b}
and Holm et al. \cite{HolmLT01}, we emphasize that our purpose for
using the ``contraction'' technique is conceptually very different
from all previous applications of the (similar) technique \cite{HenzingerK97,HolmLT01,HolmRW15,Wulff-Nilsen16a}.
The purpose of all previous applications is for reducing decremental
algorithms to fully dynamic algorithms. However, this goal is not
crucial for us. Indeed, in our application, by slightly changing the
algorithm, the input dynamic $\msf$ algorithm for \ref{thm:reduc restricted dec}
can also be fully-dynamic and not decremental. But it is very important
that the reduction must give an algorithm that is faster when there
are few non-tree edges and can handle batch insertions. Therefore,
this work illustrates a new application of the ``contraction'' technique. 

There are previous attempts for speeding up the algorithm when there
are few non-tree edges. In the dynamic $\sf$ algorithm of Nanongkai
and Saranurak \cite{NanongkaiS16} and the dynamic $\msf$ algorithm
of Wulff-Nilsen \cite{Wulff-Nilsen16a}, they both also devised the
algorithms that run on a graph with $k$ non-tree edges by extending
the 2-dimensional topology tree of Frederickson \cite{Frederickson85}.
The algorithms have $O(\sqrt{k})$ update time. In the
context of \cite{NanongkaiS16,Wulff-Nilsen16a}, they have $k=n^{1-\epsilon_{0}}$
for some small constant $\epsilon_{0}>0$ where $n$ is the number
of nodes, and hence $O(\sqrt{k})=O(n^{0.5-\epsilon_{0}/2})$. This
eventually leads to their dynamic $\sf$ and $\msf$ algorithms with
update time $n^{0.5-\Omega(1)}$. 

In our application paper, we will have $k=n^{1-o(1)}$ and the update
time of $O(\sqrt{k})$ is too slow. Fortunately, using the reduction
from this section, we can reduce to the problem where the algorithm
runs on graphs with only $O(k)$ edges, and then recursively run our
algorithm on that graph. Together with other components, this finally
leads to the algorithm with subpolynomial update time. 

The rest of this section is for proving \ref{thm:reduc restricted dec}.
Although the proof is by straightforwardly extending the reduction
of Wulff-Nilsen \cite{Wulff-Nilsen16a} which is in turn based on
the reduction by Holm et al. \cite{HolmLT01}, the reduction itself
is still quite involved. Moreover, in \cite{Wulff-Nilsen16a}, it
is only outlined how to extend from \cite{HolmLT01}. Therefore, below,
we give a more detailed proof for completeness.

\subsection{Reduction to Decremental Algorithms for Few Non-tree Edges}

In this section, we reduce from fully dynamic $\msf$ algorithms running
on a graph with $k$ non-tree edges and can handle a batch insertion
to decremental $\msf$ algorithms running on a graph with $k$ non-tree
edges as well. We will reduce further to decremental algorithms running
on a graph with $O(k)$ edges in later sections. This can be done
by straightforwardly adjusting the reduction from \cite{HolmLT01,Wulff-Nilsen16a},
we extend the reduction so that the resulting algorithm can handle
batch insertions, and the input algorithm also runs on graph with
few non-tree edges. 
\begin{lem}
Suppose there is a decremental $\msf$ algorithm $\cA$ for any $m$-edge
graph with at most $k$ non-tree edges and has preprocessing time
$t_{pre}(m,k,p)$ and update time $t_{u}(m,k,p)$. Then, for any $B\ge5\left\lceil \log k\right\rceil $,
there is a fully dynamic $\msf$ algorithm $\cB$ for any $m$-edge
graph with at most $k$ non-tree edges such that $\cB$ can:
\begin{itemize}
\item preprocess the input graph in time 
\[
t'_{pre}(m,k,B,p)=t_{pre}(m,k,p')+O(m\log m),
\]

\item handle a batch of $B$ edge insertions or an edge deletion in time:
\[
t'_{u}(m,k,B,p)=O(\sum_{i=0}^{\left\lceil \log k\right\rceil }t_{pre}(m,\min\{2^{i+1}B,k\},p')/2^{i}+B\log m+\log k\cdot t_{u}(m,k,p')),
\]

\end{itemize}
where $p'=O(p/\log k)$ and the time guarantee for each operation
holds with probability $1-p$.\label{lem:reduc to dec few edge}
\end{lem}

\subsubsection{Preprocessing}

We are given an input graph $G=(V,E)$ for $\cB$ and parameters $B$
and $p$. Let $F=\msf(G)$ denote the $\msf$ of $G$ through out
the update sequence. Let $N=E-F$ denote the set of non-tree edges.
We have that $|E|\le m$ and $|N|\le k$ at each step.

Let $L=\left\lceil \log k\right\rceil $ and $p'=p/c_{0}L$ for some
large enough constant $c_{0}$. In the algorithm, we will maintain
subgraphs $G_{i,j}$ of $G$ for each $0\le i\le L$ and $1\le j\le4$.
Additionally, there is a a subgraph $G_{L,0}$ of $G$. Let $N_{i,j}=E(G_{i,j})-\msf(G_{i,j})$.
We maintain an invariant that $N=\bigcup_{i,j}N_{i,j}$ and $|N_{i,j}|\le\min\{2^{i}B,k\}$. 

Let $\cD_{i,j}$ be an instance of the decremental $\msf$ algorithm
$\cA$ from the assumption of \ref{lem:reduc to dec few edge} that
maintains $\msf(G{}_{i,j})$. Initially, we set $G_{L,1}=G$ and other
$G_{i,j}=\emptyset$. The preprocessing algorithm is simply to initialize
$\cD_{L,1}$ and the top tree $\cT(F)$ on $F$,

\subsubsection{Update}

Given the $\tau$-th update, we say that we are\emph{ }at time step
$\tau$. There are two cases: either inserting a batch of edges or
deleting an edge. In either cases, after handling the update, we then
apply the clean-up procedure. We now describe the procedure.

\paragraph{Inserting a Batch of Edges:}

Let $I$ be a set of edges to be inserted where $|I|\le B$. We define
the set $R$ as follows. For each edge $e=(u,v)\in I$, if $u$ and
$v$ are not connected in $F$, then add $e$ to $F$. Otherwise,
$u$ and $v$ are connected in $F$ by the unique path $P_{u,v}$.
Let $f$ be the edge with maximum weight in $P_{u,v}$. Note that
$f$ can be found using the top tree $\cT(F)$. If $w(e)>w(f)$, then
we include $e$ to $R$. If $w(e)<w(f)$, then we set $F\gets F\cup e-f$
and include $f$ to $R$. Observe that $|R|\le B$. Then we run the
clean-up procedure with input $R$.

\paragraph{Deleting an Edge:}

Let $e$ be the edge to be deleted. For all $i$ and $j$, we set
$G_{i,j}\gets G_{i,j}-e$ and update $\cD_{i,j}$ accordingly. Let
$R_{0}$ be the set of reconnecting edges of $\msf(G{}_{i,j})$ returned
from $\cD_{i,j}$ over all $i$ and $j$. Among edges in $R_{0}$,
let $f$ be the lightest edge that can reconnect $F$. If $f$ exists,
set $F\gets F\cup f-e$ and $R\gets R_{0}-f$. Otherwise, $F\gets F-e$
and $R\gets R_{0}$. Then  we run the clean-up procedure with input
$R$. Observe that $|R|\le4L+1$.

\paragraph{Clean-up: }

In the following, suppose that $\cD'$ is an instance of decremental
$\msf$ algorithm running on some graph $G'$. Let $N'=E(G')-\msf(G')$.
For any $i$ and $j$, when we write $\cD_{i,j}\gets\cD'$, this means
that we set the instance $\cD_{i,j}$ to be $\cD'$. So $G_{i,j}=G'$.
If we write $\cD_{i,j}\gets\emptyset$, this means that we destroy
the instance $\cD_{i,j}$. So $G_{i,j}=\emptyset$. The time needed
for setting $\cD_{i,j}\gets\cD'$ or $\cD_{i,j}\gets\emptyset$ is
constant because it can be done by swapping pointers.

Let $R$ be the set of input edges for the clean-up procedure. Let
$R'$ be another set of edges that we will define below. Now, we describe
the clean-up procedure. For each $i$ starting from $0$ to $L+1$,
we execute the \emph{clean-up procedure for level $i$}. For any fixed
$i$, the procedure for level $i$ is as follows. 

For $i=0$, we initialize an instance of dynamic contracted $\msf$
algorithm $\cD'_{0}$ on $G'_{0}=(V,F\cup R\cup R')$. We set $\cD_{0,j}\gets\cD'_{0}$
for some $j\in\{1,2\}$ where $\cD_{0,j}=\emptyset$.

For $i>0$, all the steps $\tau$ not divisible by $2^{i}$ are for
initializing an instance of decremental $\msf$ algorithm $\cD'_{i}$
on a graph $G'_{i}$ that will be specified below. If $2^{i}$ divides
$\tau$, we claim that $\cD'{}_{i}$ is finished initializing on some
graph $G'{}_{i}$. For $0\le i\le L$, we set $\cD_{i,j}\gets\cD'_{i}$
for some $j\in\{1,2\}$ where $\cD_{i,j}=\emptyset$. If $i=L+1$,
we just set $\cD_{L,0}\gets\cD'_{i}$. Then, for any $1\le i\le L+1$,
we set $(\cD_{i-1,3},\cD_{i-1,4})\gets(\cD_{i-1,1},\cD_{i-1,2})$
and $(\cD_{i-1,1},\cD_{i-1,2})\gets(\emptyset,\emptyset)$. 

We do the following during the time period $[\tau,\tau+2^{i})$ for
initializing $\cD'_{i}$ on $G'_{i}$. Let $I_{1}=[\tau,\tau+2^{i-1})$
and $I_{2}=[\tau+2^{i-1},\tau+2^{i})$ be the first and second halves
of the period. During $I_{1}$, we evenly distribute the work for
initializing $\cD'_{i}$ on $G'_{i}$ where $G'_{i}=(V,F\cup N'_{i})$,
$N'_{i}=N_{i-1,3}\cup N_{i-1,4}$ for $0<i\le L$ and $N'_{L+1}=N_{L,0}\cup N_{L,3}\cup N_{L,4}$
for $i=L+1$. We note that and $N_{i-1,3},N_{i-1,4}$ and $F$ are
the sets of edges at time $\tau$ and hence $G'_{i}$ is a subgraph
of $G$. After $I_{1}$, we have finished the initialization $\cD'_{i}$
on the $2^{i-1}$-step-old version of $G'_{i}$. Therefore, during
$I_{2}$, we update $\cD'_{i}$ at ``double speed'' so that after
these $2^{i-1}$ steps, $\cD'_{i}$ is running on the up-to-date $G'_{i}$
as desired. More precisely, at time $\tau+2^{i-1}+k$, we feed to
$\cD'_{i}$ the updates of $G'_{i}$ from time $\tau+2^{i-1}+2k$
and $\tau+2^{i-1}+2k+1$, for each $0\le k<2^{i-1}$. 

Now, we can define the set $R'$. At any time $\tau$, before we run
the clean-up procedure for level $0$. We set $R'$ to be the set
of reconnecting edges returned by all $\cD'_{i}$ that are given edge
deletions at double speed.

\subsubsection{Correctness}

To see that the description for clean-up procedure is valid, observe
the following:
\begin{prop}
For any $i$ and $j$, $G_{i,j}$ is a subgraph of $G$ at any time.
\end{prop}

\begin{prop}
For all $0\le i\le L$, before setting $\cD_{i,j}\gets\cD'_{i}$ for
some $j\in\{1,2\}$, we have $\cD_{i,j}=\emptyset$ for some $j=\{1,2\}$.\end{prop}
\begin{proof}
We set $\cD_{i,j}\gets\cD'_{i}$ only when $2^{i}$ divides $\tau$.
But once $2^{i+1}$ divides $\tau$, we run the procedure for level
$i+1$ and set $(\cD_{i,1},\cD_{i,2})\gets(\emptyset,\emptyset)$. \end{proof}
\begin{lem}
Suppose that $N\subseteq\bigcup_{i,j}N_{i,j}$. When $e\in F$ is
deleted, let $f^{*}$ be the lightest reconnecting edge for $F$ in
$G$. Then $f^{*}\in R_{0}$.\label{lem:light edge return}\end{lem}
\begin{proof}
Before deleting $e\in F$, we have $f^{*}\in N$ and hence $f^{*}\in N_{i,j}$
for some $i,j$. As $G_{i,j}$ is a subgraph of $G$, $f^{*}$ is
also the lightest reconnecting edge in $G_{i,j}$ after deleting $e$
in $G_{i,j}$. So $f^{*}$ returned by $\cD_{i,j}$ and is included
into $R_{0}$. 
\end{proof}
The following lemma concludes the correctness of the algorithm for
\ref{lem:reduc to small dec}.
\begin{lem}
Throughout the updates, we have $N=\bigcup_{i,j}N_{i,j}$ and $F=\msf(G)$.\label{lem:reduc correct}
\end{lem}
After preprocessing, $N=\bigcup_{i,j}N_{i,j}$ and $F=\msf(G)$ by
construction. We will prove that both statements are maintained after
each update using the claims below.

First, we prove $N\subseteq\bigcup_{i,j}N_{i,j}$. When the update
is a batch $I$ of edge insertions, $R$ is exactly the set of new
non-tree edges (i.e. the new edges in $N$) because, before inserting
$I$, $F=\msf(G)$. When the update is an edge deletion, $R$ is exactly
the set of non-tree edges in all $G_{i,j}$ that become tree edges
(i.e. the edges removed from $\bigcup_{i,j}N_{i,j}$ but potentially
still in $N$). Now, after the clean-up procedure for level $0$,
we have that either $N_{0,1}$ or $N_{0,2}$ is set from $\emptyset$
to $R\cup R'$ (i.e. $R$ is included into $\bigcup_{i,j}N_{i,j}$).
Therefore, $N\subseteq\bigcup_{i,j}N_{i,j}$ after the procedure for
level $0$. 
\begin{claim}
For $1\le i\le L+1$, $N\subseteq\bigcup_{i,j}N_{i,j}$ after applying
the clean-up procedures of level $i$. \label{claim:no new non-tree}\end{claim}
\begin{proof}
For $i\le L$, suppose we are at time $\tau$ divisible by $2^{i}$
and we set $\cD_{i,j}\gets\cD'_{i}$ where $\cD_{i,j}=\emptyset$,
$(\cD_{i-1,3},\cD_{i-1,4})\gets(\cD_{i-1,1},\cD_{i-1,2})$ and $(\cD_{i-1,1},\cD_{i-1,2})\gets(\emptyset,\emptyset)$.
That is, the edges of $N_{i-1,3}\cup N_{i-1,4}$ contributing to $\bigcup_{i,j}N_{i,j}$
are replaced by the edges in $N'_{i}$. We argue that after this we
still have $N\subseteq\bigcup_{i,j}N_{i,j}$. This is because, at
time $\tau-2^{i}$, when we start the initialization of $\cD'_{i}$
on $G'_{i}$, we set $N'_{i}=N_{i-1,3}\cup N_{i-1,4}$ exactly. Then,
from time $\tau-2^{i}$ to $\tau$, all reconnecting edges returned
by $\cD'_{i}$ are included in to $R'$ in every step. That is, $R'$
contains the edges that are removed from $N'_{i}$ but potentially
still in $N$. As $R'$ is included into $\bigcup_{i,j}N_{i,j}$ at
every step, we are done.

For $i=L+1$, we set $\cD_{L,0}\gets\cD'_{i}$, $(\cD_{L,3},\cD_{L,4})\gets(\cD_{L,1},\cD_{L,2})$
and $(\cD_{L,1},\cD_{L,2})\gets(\emptyset,\emptyset)$. Although,
$\cD_{L,0}\neq\emptyset$ before we set $\cD_{L,0}\gets\cD'_{i}$,
we have that $N'_{i}=N_{L,0}\cup N_{L,3}\cup N_{L,4}$ exactly at
$2^{i}$ steps ago. Using the same argument, we have $N\subseteq\bigcup_{i,j}N_{i,j}$
after the procedure.
\end{proof}
Second, we prove that $\bigcup_{i,j}N_{i,j}\subseteq N$ which follows
from the two claims below.
\begin{claim}
\label{claim:remove Nij }Whenever an edge $f$ is removed from $N$,
then $f$ is removed from $\bigcup_{i,j}N_{i,j}$.\end{claim}
\begin{proof}
Observe that an edge $f$ can be removed from the set $N$ by one
of two reasons: 1) $f$ is deleted from $G$, or 2) some edge $e$
is deleted and $f$ is a reconnecting edge. If $f$ is deleted $G$,
then $f$ is deleted from all $G_{i,j}$. If $f$ is a reconnecting
edge, by \ref{lem:light edge return}, we know $f$ is the lightest
reconnecting edge in $G$. But $G_{i,j}$ is a subgraph of $G$ for
all $i$ and $j$. If $f\in N_{i,j}$ for any $i,j$, then $f$ must
be the lightest reconnecting edge in $G_{i,j}$ and hence $f$ is
removed from $N_{i,j}$. In either case, when $f$ is removed from
$N$, then $f$ is removed from $\bigcup_{i,j}N_{i,j}$ as well. \end{proof}
\begin{claim}
Whenever an edge $f$ is added into $\bigcup_{i,j}N_{i,j}$, then
$f$ is added into $N$.\end{claim}
\begin{proof}
For any fixed $i$, $f$ can be added into $\bigcup_{j}N_{i,j}$ only
when $2^{i}$ divides $\tau$ and we set $\cD_{i,j}\gets\cD'_{i}$.
Now, when we start the initialization of $\cD'_{i}$ on $G'_{i}$
at time $\tau-2^{i}$, we set $G'_{i}=(V,F\cup N'_{i})$ and so $N'_{i}\subseteq N$
at that time. From time $\tau-2^{i}$ to $\tau$, there is no edges
added to $N'_{i}$ because no $\msf$-edge can become non-tree edge
in a graph undergoing only edge deletions. Moreover whenever an edge
$f$ is removed $N$, it is removed from $N'_{i}$ for the same reason
as in \ref{claim:remove Nij }. So at time $\tau$, we have $N'_{i}\subseteq N$.
\end{proof}
Now, we have that $\bigcup_{i,j}N_{i,j}=N$ is maintained. To show
that $F=\msf(G)$ after the update, whenever an edge $e$ is deleted,
by $N\subseteq\bigcup_{i,j}N_{i,j}$ and \ref{lem:light edge return},
the lightest reconnecting edge $f^{*}$is included in $R_{0}$. So
$F\gets F\cup f^{*}-e$ and $F$ becomes $\msf(G)$. For each inserted
edge $e=(u,v)$, we may only remove the heaviest edge $f$ in the
path $u$ to $v$ in $F$. So $F=\msf(G)$ as well. This concludes
the proof of \ref{lem:reduc correct}.

\subsubsection{Running Time}

First, we need this lemma.
\begin{lem}
For any $i$ and $j$, $|N_{i,j}|\le\min\{2^{i+1}B,k\}$.\label{lem:bound non-tree level i}\end{lem}
\begin{proof}
First, $N_{i,j}\subseteq N$ by \ref{lem:reduc correct} and so $|N_{i,j}|\le|N|\le k$.
Next, we will prove that $|N_{i,j}|\le2^{i}B$, for any $i$ and $j$
by induction. As we know $|N_{L,0}|\le k$ and we always set $(\cD_{i,3},\cD_{i,4})\gets(\cD_{i,1},\cD_{i,2})$
for all $i\le L$, it remains to bound only $|N_{i,1}|$ and $|N_{i,2}|$
for all $i\le L$. 

For $i=0$, in the clean-up procedure we only set $N_{0,1}$ and $N_{0,2}$
as the input set $R$. When there is a batch insertion, we have that
$|R|\le B$. When there is an edge deletion, we have $|R|\le4L+1$
because, for each $i,j$, the instance $\cD_{i,j}$ can return one
reconnecting edge per time step. Also, observe that $|R'|\le2(L+1)$
because, for $1\le\cD'_{i}\le L+1$, $\cD'_{i}$ can return two reconnecting
edge per time step (as they are possibly updated at double speed).
Therefore, $|N_{0,1}|,|N_{0,2}|\le|R\cup R'|\le\max\{B,4L+1\}+2(L+1)\le2B$
as $B\ge5L$. Next, for $0<i\le L$, in the clean-up procedure we
only set $N_{i,1}$ and $N_{i,2}$ as the set $N'_{i}$ where $N'_{i}=N_{i-1,3}\cup N_{i-1,4}$
So $|N_{i,1}|,|N_{i,2}|\le|N'_{i}|\le2\cdot2^{i}B=2^{i+1}B$.
\end{proof}
Now, we can conclude the preprocessing and update time of the algorithm
for \ref{lem:reduc to small dec}.
\begin{lem}
The preprocessing algorithm takes $t_{pre}(m,k,p')+O(m\log m)$ time.\label{lem:reduc prep time}\end{lem}
\begin{proof}
We can compute $F=\msf(G)$ in time $O(m\log m)$. The algorithm
is to just initialize the top tree $\cT(F)$ on $F$ and the decremental
$\msf$ $\cD_{i,j}$ for all $i,j$. By \ref{thm:top tree}, we initialize
$\cT$ in time $O(m)$. Then, as $N_{L,1}=N$ and $|N|\le k$, we
can initialize all $\cD_{L,1}$ in time $t_{pre}(m,k,p')$.
\end{proof}
Next, we bound the time spent on the clean-up procedure at each step.
\begin{lem}
For each update, the time spent on the clean-up procedure is 
\[
O(\sum_{i=0}^{\left\lceil \log k\right\rceil }t_{pre}(m,\min\{2^{i+1}B,k\},p')/2^{i}+\log k\cdot t_{u}(m,k,p'))
\]
 with probability $1-p/2$.\label{lem:clean-up time}\end{lem}
\begin{proof}
For each $i$, we analyze the time spent on the clean-up procedure
of level $i$. When the time step $\tau$ is divisible by $2^{i}$,
this takes $O(1)$ time because all we do is only setting $\cD_{i,j}\gets\cD'$
for some $\cD'$ for several $j$. After that, for the period $I_{1}=[\tau,\tau+2^{i-1})$,
we distribute evenly the work for initializing $\cD'_{i}$. This takes
$t_{pre}(m,|N'_{i}|,p')/2^{i-1}$ time per step. For the period $I_{2}=[\tau+2^{i-1},\tau+2^{i})$,
we feed the updates to $\cD'_{i}$ at ``double speed''. This takes
$2\cdot t_{u}(m,|N'_{i}|,p')$ time per step. By \ref{lem:bound non-tree level i},
the total time for step for the clean-up procedure of level $i$ is
\[
\frac{t_{pre}(m,\min\{2^{i+1}B,k\},p')}{2^{i-1}}+2\cdot t_{u}(m,\min\{2^{i+1}B,k\},p')).
\]
Finally, as $L=\left\lceil \log k\right\rceil $, summing the time
for all levels give the bound in lemma. Moreover, the bound holds
with probability $1-p'\times O(L)\le1-p/2$.\end{proof}
\begin{lem}
The time for inserting a batch of edges of size at most $B$ and the
time for deleting an edge is at most 
\[
O(\sum_{i=0}^{\left\lceil \log k\right\rceil }t_{pre}(m,\min\{2^{i+1}B,k\},p')/2^{i}+\log k\cdot t_{u}(m,k,p')+B\log m)
\]
 with probability $1-p$. \label{lem:reduc update time}\end{lem}
\begin{proof}
We show the update time outside the clean-up procedure. For insertion,
we need $O(B\log m)$ time be the property of top tree. For deletion,
we need $\sum_{i,j}t_{u}(m,|N_{i,j}|,p')=O(t_{u}(m,k,p')\log k)$
by \ref{lem:bound non-tree level i}, and the bound holds with probability
$1-p'O(L)\le1-p/2$. By \ref{lem:clean-up time}, we are done and
the bound holds with probability $1-2\cdot p/2=1-p$.
\end{proof}
\ref{lem:reduc correct}, \ref{lem:reduc prep time}, and \ref{lem:reduc update time}
concludes the proof of \ref{lem:reduc to dec few edge}.

\subsection{Contraction}

From \ref{lem:reduc to dec few edge}, we have reduced the problem
to decremental algorithms on graphs with few non-tree edges. We want
to further reduce the problem to decremental algorithms on graph with
few edges (using some additional data structures). Informally, we
would like to prove the following:
\begin{lem}
[Informal statement of \ref{lem:reduc contract}] Suppose there is
a decremental $\msf$ algorithm $\cA'$ for any $m'$-edge graph with
preprocessing time $t_{pre}(m',p)$ and update time $t_{u}(m',p)$.
Then, for any $m$, $k$ where $5k\le m'$, and $B$, let $G=(V,E)$
be a graph with $m$-edge graph and at most $k$ non-tree edges. Then,
with some additional data structures, there is a decremental dynamic
$\msf$ algorithm $\cB$ for $G$ with preprocessing time $t'_{pre}(m,k,p)=t_{pre}(5k,p)+O(k\log m)$
and edge-deletion time $t'_{u}(m,k,p)=t_{u}(5k,p)+O(\log m)$ with
probability $1-p$.
\end{lem}
The rest of this section is devoted to the proof of \ref{lem:reduc contract}.
We use the contraction technique by Holm et al. \cite{HolmLT01}.
Below, we define some related notions and analyze their properties.
Then, we show the proof of the reduction.

\subsubsection{Property of Contracted Graphs/Forests}

For any tree $T=(V,E)$ and the set $S\subseteq V$ of terminals,
the connecting paths of $T$ with respect to $S$ are the path that
the minimal collections of disjoint paths in $T$ that ``connect''
the terminals in $S$. Below is the formal definition:
\begin{defn}
[Connecting Paths]Given a tree $T=(V,E)$ and a set of \emph{terminals}
$S\subseteq V$, the set $\cP_{S}(T)$ of \emph{connecting paths of
$T$ with respect to $S$} is defined as follows:
\begin{enumerate}
\item $\cP_{S}(T)$ is a collection of edge-disjoint paths.
\item The graph $\Disjunion_{P\in\cP_{S}(T)}P$ is a (connected) subtree
of $T$.
\item For any terminal $u\in S$, $u$ is an endpoint of some path $P\in\cP_{S}(T)$. 
\item For any endpoint $u$ of $P\in\cP_{S}(T)$, either $u\in S$ or there
are other two paths $P',P''\in\cP_{S}(T)$ whose endpoint is $u$.
\end{enumerate}

The definition can be extended as follows: for any forest $F=(V,E)$
and a set of terminal $S$ (where possibly $S\not\subset V$), $\cP_{S}(F)=\Disjunion_{T\in F}\cP_{S\cap V(T)}(T)$
is the disjoint-union of $\cP_{S\cap V(T)}(T)$ over all (connected)
tree $T$ in $F$.\label{def:super edges}

\end{defn}
Condition 4 in \ref{def:super edges} implies the minimality of $\cP_{S}(T)$.
Indeed, suppose otherwise that $u$ is an endpoint of $P\in\cP_{S}(T)$
but $u\notin S$ and there is only one other path $P'$ whose endpoint
is $u$. Then, we can replace $P$ and $P'$ with a path $P''=P\cup P'$
while other conditions are still satisfied. The following lemma formally
shows that the definition of $\cP_{S}(F)$ is uniquely defined.
\begin{lem}
For any tree $T=(V,E)$ and a set $S\subseteq V$ of terminals, $\cP_{S}(T)$
is uniquely defined.\end{lem}
\begin{proof}
Suppose there are two different sets $\cP$ and $\cP'$ of connecting
paths of $T$ with respect to $S$. Let $end(\cP)$ and $end(\cP')$
be the set of endpoints of paths in $\cE$ and $\cE'$. Condition
3 states that $S\subseteq end(\cP)\cap end(\cP')$. Let $H=\Disjunion_{P\in\cP}P$
and $H'=\Disjunion_{P\in\cP'}P$. Condition 2 states that both $H$
and $H'$ are (connected) subtrees of $T$. 

Observe that all the sets of leaves of $H$ and $H'$ must be the
same. Otherwise, we can assume w.l.o.g. that there is a leaf $u$
of $H$ where $u\in V(H)\setminus V(H')$. By Condition 4, $u\in S$
as $u$ is a leaf in $H$. But $u\notin end(\cP')$ which contradicts
that fact that $S\subseteq end(\cP')$. Now, as $H$ and $H'$ share
the same at of leaves, it follows that $H$ and $H'$ are the same
subtree in $T$.

As $\cP$ and $\cP'$ are different, w.l.o.g. there is a node $u\in end(\cP)\setminus end(\cP')$.
As $S\subseteq end(\cP')$, we have that $u\notin S$. By Condition
4, $u$ is an endpoint of three paths $P,P',P''\in\cP$ which are
edge-disjoint by Condition 1. In particular, $\deg_{H}(u)\ge3$. Next,
as $u\in V(H)=V(H')$ but $u\notin end(\cP')$, we have that $u$
is an internal node of some path $P'\in\cP$. So $\deg_{H'}(u)=2$.
This is a contradiction because $H$ and $H'$ are the same subtree
but $2=\deg_{H'}(u)=\deg_{H}(u)\ge3$.
\end{proof}
Below, for any set of edges $E'$, let $end(E')$ denote the set of
endpoints of edges in $E'$. Observe the following:
\begin{prop}
For any graph $G=(V,E)$, forest $F\subseteq E$, and any $S\supseteq end(E-F)$,
each connecting path in $\cP_{S}(F)$ is an induced path in $G$. \end{prop}
\begin{defn}
[Contracted Graphs/Forests and Super Edges]For any weighted graph
$G=(V,E,w)$, forest $F\subseteq E$ and a set $S\supseteq end(E-F)$
of terminals, the \emph{contracted (multi)-graph} $G'$ and \emph{contracted
forest} $F'$ with respect to $S$ is obtained from $G$ and $F$
respectively by 1) removing all edges in $F\setminus\Disjunion_{P\in\cP_{S}(F)}P$,
and 2) replacing each connecting path $P_{uv}=(u,\dots,v)\in\cP_{S}(F)$
in $F$ with an edge $(u,v)$, called \emph{super edge}, with weight
$w(u,v)=\max_{e\in P_{uv}}w(e)$. We denote $(G',F')=\contract_{S}(G,F)$.
For each $e\in P_{uv}\subseteq F$, we say $(u,v)\in F'$ is the \emph{super
edge covering} $e$ with respect to $S$.
\end{defn}
In the following propositions, we show some properties of the contracted
graph and forest. Let $G=(V,E)$ be a graph and $F\subseteq E$ be
a forest. Denote by $N=E-F$ the set of non-tree edges and $S\supseteq end(N)$
a set of terminal contains points non-tree edges. Let $(G',F')=\contract_{S}(G,F)$
be the contracted graph and forest with respect to $S$. First, we
show that the number of edges in the contracted graph is linear in
the number of original non-tree edges plus the number of terminals.
\begin{prop}
$E(G')=N\disjunion E(F')$ and $E(G')\le|N|+2|S|$.\label{prop:contract few edge}\end{prop}
\begin{proof}
We have that $E(G')=N\disjunion E(F')$ by construction. It is enough
to prove that $|E(F')|\le2|S|$. For each tree $T$ in $F$, let $T'$
be the corresponding contracted tree in $F'$. Again, let $end(\cP_{G}(T))$
is the set of endpoints of paths in $\cP_{G}(T)$ and let $S_{T}=S\cap V(T)$
be the set of terminal in $T$. We have $V(T')=end(\cP_{G}(T))\supseteq S_{T}$.
Also, for all $u\in end(\cP_{G}(T))\setminus S_{T}$, $\deg_{T'}(u)\ge3$
by Condition 4 of \ref{def:super edges}. It is well known that, in
any tree, the number of nodes with degree at least three is at most
the number of nodes with degree at most two. So $|S_{T}|\ge|end(\cP_{G}(T))\setminus S_{T}|$
and hence $|end(\cP_{G}(T))|\le2|S_{T}|$. Note that $|E(T')|=|V(T')|-1=|end(\cP_{G}(T))|-1$.
So 
\[
|E(F')|=\sum_{T'\in F'}|E(T')|\le\sum_{T\in F}|end(\cP_{G}(T))|\le2\sum_{T\in F}|S_{T}|\le2|S|.
\]
\end{proof}
\begin{prop}
If $F=\msf(G)$, then $F'=\msf(G')$.\label{prop:msf preserve}\end{prop}
\begin{proof}
As $E(G')=N\disjunion E(F')$ by \ref{prop:contract few edge}, we
only need to prove that $N\cap\msf(G')=\emptyset$. For each $e=(u,v)\in N$,
then $u$ and $v$ must be connected in $F$, otherwise $F$ is not
spanning. Let $P_{u,v}\subseteq F$ be a path in $F$ connecting $u$
and $v$. We have that $w(e)>\max_{f\in P_{u,v}}w(f)$ otherwise $F$
is not an $\msf$. Let $P'_{u,v}\subseteq F'$ be a path in $F'$
connecting $u$ and $v$. Observe that each edge $e'\in P'_{u,v}$
corresponds to a connecting path $P_{e'}\subseteq P_{u,v}$. So $w(e')\le\max_{f\in P_{u,v}}w(f)$.
So we have $\max_{e'\in P'_{u,v}}w(e')=\max_{f\in P_{u,v}}w(f)<w(e)$.
That is, $e\notin\msf(G')$.
\end{proof}
Here, we show the change of 1) the contracted graph/forest $(G',F')$,
2) the $\msf$ of the graph $\msf(G)$, and 3) the $\msf$ of the
contracted graph $\msf(G')$ when we delete an edge $e$ from $G$.
Note that the set $S$ of terminal does \emph{not} change.
\begin{prop}
\label{prop:contracted change}Suppose that $F=\msf(G)$. For any
edge $e$ of $G$, let $G_{1}=G-e$, $F_{1}=F-e$ (so $F_{1}=F$ if
$e\notin F$) and $(G'_{1},F'_{1})=\contract_{S}(G_{1},F_{1})$ (note
that $S=end(E-F))$.
\begin{enumerate}
\item If $e\in N$, then $(G'_{1},F'_{1})=(G'-e,F')$, $\msf(G_{1})=F=F_{1}$
and $\msf(G'_{1})=F'=F'_{1}$.
\item Else, if $e\in F\setminus\Disjunion_{P\in\cP_{S}(F)}P$, then $(G'_{1},F'_{1})=(G',F')$,
$\msf(G_{1})=F-e=F_{1}$ and $\msf(G'_{1})=F'=F'_{1}$.
\item Else, $e\in\Disjunion_{P\in\cP_{S}(F)}P$ and there is a super edge
$e'=(u',v')\in F'$ covering $e$ with respect to $S$. Then, $(G'_{1},F'_{1})=(G'-e',F'-e')$.
Then, one of two cases holds:

\begin{enumerate}
\item $\msf(G{}_{1})=F-e=F_{1}$ and $\msf(G'_{1})=F'-e'=F'_{1}$, or
\item There is $f\in N$ where $\msf(G{}_{1})=F\cup f-e=F_{1}\cup f$ and
$\msf(G'_{1})=F'\cup f-e'=F'_{1}\cup f$.
\end{enumerate}
\end{enumerate}
\end{prop}

\subsubsection{The Proof}

Holm et al. \cite{HolmLT01} show the following crucial structure
based on top tree. 
\begin{lem}
[Lemma 15 of \cite{HolmLT01}]\label{lem:contraction}There is an
algorithm $\cC$ that runs in two phases:
\begin{enumerate}
\item In the first phase: $\cC$ maintains an at-most-$m$-edge forest $F$
undergoing a sequence of edge insertions and deletions. $\cC$ can
handle each edge update of $F$ in $O(\log m)$. 
\item Then, in the second phase: given any set $N$ of edges, $\cC$ can
return $(G',F)=\contract_{S}(G,F)$ in time $O(|N|\log m)$ where
$G=(V,F\cup N)$ and $S=end(N)$. Moreover, for any edge $e\in F$,
$\cC$ can return a super edge $e'=(u',v')\in F'$ covering $e$ with
respect to $S$, if $e'$ exists, in time $O(\log m)$.
\end{enumerate}
\end{lem}
Having \ref{lem:contraction} together with \ref{prop:contracted change},
we can prove the main reduction of this section.
\begin{lem}
\label{lem:reduc contract}Suppose there is a decremental $\msf$
algorithm $\cA'$ for any $m'$-edge graph with preprocessing time
$t_{pre}(m',p)$ and update time $t_{u}(m',p)$. Then, for any $m$,
$k$ where $5k\le m'$, and $B$, let $G=(V,E)$ be a graph with $m$-edge
graph and at most $k$ non-tree edges. Suppose that $F=\msf(G)$ and
$N=E-F$ are given. Moreover, there is a given instance $\cC$ of
the algorithm \ref{lem:contraction} that is running its first phase
on $F$. Then, there is a decremental dynamic $\msf$ algorithm $\cB$
for $G$ with preprocessing time $t'_{pre}(m,k,p)=t_{pre}(5k,p)+O(k\log m)$
and edge-deletion time $t'_{u}(m,k,p)=t_{u}(5k,p)+O(\log m)$ with
probability $1-p$.\end{lem}
\begin{proof}
The preprocessing algorithm for $\cB$ is the following. First, we
switch $\cC$ to the second phase and give the non-tree edges $N$
to $\cC$, and obtain $(G',F)=\contract_{S}(G,F)$ where $S=end(N)$.
By \ref{prop:contract few edge}, $|E(G')|\le|N|+2|S|\le5k$. After
that we initialize $\cA'$ on $G'$ with probability parameter $p$
in time $t_{pre}(|E(G')|,p)=t_{pre}(5k,p)$. In total this takes $t'_{pre}(m,k,p)=t_{pre}(5k,p)+O(k\log m)$.

Throughout that the edge-deletion sequence for $\cB$, the set $S$
is fixed. We will maintain the following invariant 1) $F=\msf(G)$,
2) $(G',F')=\contract(G,F)$, and 3) for each $e\in F$, we can find
a super edge $e'=(u',v')\in F'$ covering $e$ with respect to $S$
in time $O(\log m)$. 

Given an edge $e$ to be deleted, let $G_{1}=G-e$, let $F_{1}=F-e$
(so $F_{1}=F$ if $e\notin F$) and $(G'_{1},F'_{1})=\contract(G_{1},F_{1})$.
We note that we can obtain $(G'_{1},F'_{1})$ in time $O(\log m)$.
Indeed, by \ref{prop:contracted change} either $(G'_{1},F_{1}')=(G',F')$,
$(G'_{1},F_{1}')=(G'-e,F')$ or $(G'_{1},F_{1}')=(G'-e',F'-e')$ where
$e'\in F'$ is a super edge covering $e$ w.r.t. $S$. By the invariants
3, we can check which case of \ref{prop:contracted change} holds,
and compute $(G'_{1},F'_{1})$ in time $O(\log m)$. Then, we feed
the change from $G'$ to $G'_{1}$ to $\cA'$. This takes $t_{pre}(5k,p)$
because there is at most 1 edge update from $G'$ to $G'_{1}$ according
the \ref{prop:contracted change}.

Finally, there are two cases. For the first case, suppose that $\cA'$
returns a reconnecting edge $f$, i.e. $\msf(G'_{1})=\msf(G')\cup f-e'$.
We are in Case 3.b of \ref{prop:contracted change}. We update $F$
to be $F_{1}\cup f$ in $O(1)$ time to satisfy invariant 1 as $F_{1}\cup f=\msf(G_{1})$.
We update $F'$ to be $F'_{1}\cup f$ in time $O(1)$. As $F'_{1}\cup f=\msf(G'_{1})$
and $(G'_{1},\msf(G'_{1}))=\contract(G_{1},\msf(G_{1}))$ by \ref{prop:msf preserve},
so invariant 2 is satisfied. Observe that $f$ is the only new edge
in $F$ and $f\in F'$ is also the super edge covering $f$ w.r.t.
to $S$. So invariant 3 is easily maintained.

For the second case, $\cA$ does not return a reconnecting edge. We
simply update $F$ and $F'$ to $F_{1}$ and $F'_{1}$ in $O(1)$
time. By other cases of \ref{prop:contracted change}, all the three
invariants are maintained.
\end{proof}

\subsection{Reduction to Decremental Algorithm for Few Edges}

In this section, we show how to speed up the resulting algorithm from
\ref{lem:reduc to dec few edge} using \ref{lem:reduc contract}.
Recall that in the reduction in \ref{lem:reduc to dec few edge} reduces
to decremental $\msf$ algorithms which runs on graphs with few non-tree
edges. While the reduction in \ref{lem:reduc contract} can quickly
reduce further to decremental $\msf$ algorithms running on graphs
with few edges, \ref{lem:reduc contract} needs some additional data
structures to be prepared. So we will show how to augment the reduction
in \ref{lem:reduc to dec few edge} so that at any time the needed
additional data structures for \ref{lem:reduc contract} is prepared.
This gives the following lemma:
\begin{lem}
\label{lem:reduc to small dec}Suppose there is a decremental $\msf$
algorithm $\cA$ for any $m'$-edge graph with preprocessing time
$t_{pre}(m',p)$ and update time $t_{u}(m',p)$. Then, for any $m$,
$k$ where $5k\le m'$, and $B$, there is a fully dynamic $\msf$
algorithm $\cB$ for any $m$-edge graph with at most $k$ non-tree
edges such that $\cB$ can:
\begin{itemize}
\item preprocess the input graph in time $t'_{pre}(m,k,B,p)=t_{pre}(5k,p')+O(m\log^{2}m)$,
and
\item handle a batch of $B$ edge insertions or an edge deletion in time:
$t'_{u}(m,k,B,p)=O(\frac{B\log k}{k}\cdot t_{pre}(5k,p')+B\log^{2}m+\log k\cdot t_{u}(5k,p'))$, 
\end{itemize}
where $p'=O(p/\log k)$ and the time guarantee for each operation
holds with probability $1-p$.
\end{lem}
To prove \ref{lem:reduc to small dec}, we show how to augment the
reduction in \ref{lem:reduc to dec few edge} as follows. For each
$0\le i\le L+1$ and $1\le j\le6$, we let $\cC_{i,j}$ be the instance
of the algorithm from \ref{lem:contraction}. We say that $\cC_{i,j}$
is \emph{ready}, if it is in its first phase and $\cC_{i,j}$ is running
on $F=\msf(G)$ undergoing edges updates. When we want to start the
initialization of $\cD'_{i}$, we claim that, for some $j$, $\cC_{i,j}$
is ready. Then, the ready $\cC_{i,j}$ will be shift to its second
phase by plugging into the reduction \ref{lem:reduc to small dec}.
More precisely, $\cC_{i,j}$ will be given a set $N'_{i}$ where $G'_{i}=(V,F\cup N'_{i})$
and $\cC_{i,j}$ needs to return the contracted graph and forest $\contract(G'_{i},F)$.
We say that $\cC_{i,j}$ is \emph{occupied by $\cD'_{i}$}. For any
$\cD',\cD''$, suppose that $\cC_{i,j}$ was occupied by $\cD'$ and
we set $\cD''\gets\cD'$, then we say $\cC_{i,j}$ is occupied by
$\cD''$. If $\cC_{i,j}$ was occupied by $\cD'$ and we set $\cD'\gets\cD''$,
then we say $\cC_{i,j}$ is not occupied by anyone and we say \emph{$\cC_{i,j}$
}is\emph{ free}. But note that $\cC_{i,j}$ can be free but not ready.
Now, we show how we make sure that for some $j$, $\cC_{i,j}$ is
ready when we want to start the initialization of $\cD'_{i}$.

\paragraph{Preprocessing.}

In the beginning all $\cC_{i,j}$ are free. After computing $F=\msf(G)$,
we inserting edges in $F$ into all $\cC_{i,j}$ so that all $\cC_{i,j}$
are ready. In total, this takes additional $O(L\times m\log m)=O(m\log^{2}m)$
time to the preprocessing algorithm in \ref{lem:reduc to dec few edge}.

\paragraph{Updates.}

Fix any $i$ and $j$. Suppose that at time $\tau$, $\cC_{i,j}$
was ready and then is occupied by $\cD'_{i}$. Observe the following.
At time $\tau+2^{i}$, $\cC_{i,j}$ will be occupied by $\cD_{i,j'}$
for some $j'\in\{1,2\}$ as we set $\cD_{i,j}\gets\cD'_{i}$. Then,
at time $\tau+3\cdot2^{i}$, $\cC_{i,j}$ will be occupied by $\cD_{i,j''}$
for some $j''\in\{3,4\}$ as we set $\cD_{i,j''}\gets\cD_{i,j'}$.
Then, at time $\tau+5\cdot2^{i}$, $\cC_{i,j}$ will be free. Then,
during the next $2^{i}$ steps, we will make sure that $\cC_{i,j}$
is ready by spending time $O(B\log m)$ per step. Let $F_{\tau}$
be the $\msf$ that $\cC_{i,j}$ was running on at time $\tau$. When
$\cC_{i,j}$ become occupied, we will maintain the difference of edges
in $F_{\tau}$ and the current $\msf$ $F$. At time $\tau+5\cdot2^{i}$,
the difference between $F_{\tau}$ and $F$ is at most $O(B\cdot2^{i})$
edges. So we can update $F_{\tau}$ to become $F$ in $2^{i}$ steps,
using $O(\frac{B2^{i}\log m}{2^{i}})=O(B\log m)$ time per step. Summing
over all $i$ and $j$, this takes additional $O(L\times B\log m)=O(B\log^{2}m)$
time per step to the update algorithm in \ref{lem:reduc to dec few edge}.

Therefore, for any fixed $i$, only every $2^{i}$ steps one of $\cC_{i,j}$
can changed from being ready to being occupied. But in the next $6\cdot2^{i}$
steps such $\cC_{i,j}$ will become ready. As we have $6$ instances
$\cC_{i,1},\dots,\cC_{i,6}$. At any time, when we want to start the
initialization of $\cD'_{i}$, $\cC_{i,j}$ is ready for some $j$.

\paragraph{Proof of \ref{lem:reduc to small dec}.}

Recall that $\cA$ is the algorithm from the assumption of \ref{lem:reduc to small dec}
with preprocessing time $t_{pre}(m',p)$ and update time $t_{u}(m',p)$,
and $\cB$ is the resulting algorithm of \ref{lem:reduc to small dec}
with preprocessing time $t'_{pre}(m,k,B,p)$ and update time $t'_{u}(m,k,B,p)$. 

Let $\cA'$ be a decremental $\msf$ algorithm runs on graphs with
$m$-edge and $k$-non-tree-edge with parameter $p$. Denote the preprocessing
time of $\cA'$ by $t{}_{pre}^{\cA'}(m,k,p)$ and the update time
of $\cA'$ by $t_{u}^{\cA'}(m,k,p)$. To prove \ref{lem:reduc to small dec},
we will use \ref{lem:reduc to dec few edge} to first reduce $\cB$
to $\cA'$. Then, with additional preprocessing time of $O(m\log^{2}m)$
and update time of $O(B\log^{2}m)$, the argument above shows that
we can further reduce $\cA'$ to $\cA$ using \ref{lem:reduc contract}.
That is, we have $t{}_{pre}^{\cA'}(m,k,p)=t_{pre}(5k,p)+O(k\log m)$
and $t_{u}^{\cA'}(m,k,p)=t_{u}(5k,p)+O(\log m)$. Hence, the preprocessing
of $\cB$ is 

\begin{align*}
t'_{pre}(m,k,B,p)+O(m\log^{2}m) & =t_{pre}^{\cA'}(m,k,p')+O(m\log m)+O(m\log^{2}m) & \mbox{by \ref{lem:reduc to dec few edge}}\\
 & =t_{pre}(5k,p')+O(m\log^{2}m),
\end{align*}
and the update time of $\cB$ is 

\begin{align*}
 & t'_{u}(m,k,B,p)+O(B\log^{2}m)\\
 & =O(\sum_{i=0}^{\left\lceil \log k\right\rceil }t_{pre}^{\cA'}(m,\min\{2^{i+1}B,k\},p')/2^{i}+B\log m+\log k\cdot t_{u}^{\cA'}(m,k,p'))+O(B\log^{2}m) & \mbox{by \ref{lem:reduc to dec few edge}}\\
 & =O(\sum_{i=0}^{\left\lceil \log k\right\rceil }t_{pre}(5\cdot\min\{2^{i+1}B,k\},p')/2^{i}+B\log m+\log k\cdot t_{u}(5k,p'))+O(B\log^{2}m)\\
 & =O(\frac{B\log k}{k}\cdot t{}_{pre}(5k,p')+B\log^{2}m+\log k\cdot t{}_{u}(5k,p')).
\end{align*}
The last equality follows because we claim that $t_{pre}(5\cdot\min\{2^{i+1}B,k\},p')/2^{i}=O(\frac{B}{k})\cdot t{}_{pre}(5k,p')$.
To see this, there are two cases. If $2^{i+1}B\ge k$, then $\frac{t_{pre}(5\cdot\min\{2^{i+1}B,k\},p')}{2^{i-1}}=\frac{t_{pre}(5k,p')}{2^{i-1}}=O(\frac{B}{k})\cdot t_{pre}(5k,p')$.
If $2^{i+1}B<k$, then $\frac{t_{pre}(5\cdot\min\{2^{i+1}B,k\},p')}{2^{i-1}}=\frac{t_{pre}(5\cdot2^{i+1}B,p')}{2^{i-1}}\le\frac{t_{pre}(5\cdot2^{i+1}B\cdot\frac{k}{2^{i+1}B},p')}{2^{i-1}\cdot\frac{k}{2^{i+1}B}}=O(\frac{B}{k})\cdot t_{pre}(5k,p')$
where the inequality is because $t_{pre}(k,p')$ is at least linear
in $k$. This concludes the proof.

\subsection{Reduction to Restricted Decremental Algorithm for Few Edges}

The final step is apply the following standard reduction: 
\begin{prop}
Suppose there is an algorithm $\cA$ as in \ref{thm:reduc restricted dec}.
Then, there is a decremental $\msf$ algorithm $\cA'$ for any $m'$-edge
graph with preprocessing time $O(m')+t_{pre}(3m',p/3)$ and update
time $O(m'/T(m'))+3t_{u}(3m',p/3)$ with probability $1-p$.\label{lem:reduc standard}
\end{prop}
From \ref{lem:reduc to small dec}, we immediately obtain \ref{thm:reduc restricted dec}.
\begin{proof}
[Proof of \ref{thm:reduc restricted dec}]Given the algorithm $\cA$,
by \ref{lem:reduc standard} there is an algorithm $\cA'$ for any
$m'$-edge graph with preprocessing time $t_{pre}^{\cA'}(m',p)=O(m')+t_{pre}(3m',p/3)$
and update time $t_{u}^{\cA'}(m',p)=O(m'/T(m'))+3t_{u}(3m',p/3)$
with probability $1-p$. Plugging $\cA'$ to \ref{lem:reduc to small dec},
we obtain the resulting algorithm $\cB$ of \ref{thm:reduc restricted dec}
that can run on any $m$-edge graph with at most $k$ non-tree edges.
Let $p'=O(p/\log k)$ be from \ref{thm:reduc restricted dec} and
$p''=p'/3$. $\cB$ has the preprocessing time 

\begin{eqnarray*}
t'_{pre}(m,k,B,p) & = & t_{pre}^{\cA'}(5k,p')+O(m\log^{2}m)\\
 & = & t_{pre}(15k,p'')+O(m\log^{2}m),
\end{eqnarray*}
and $\cB$ has the update time
\begin{eqnarray*}
t'_{u}(m,k,B,p) & = & O(\frac{B\log k}{k}\cdot t_{pre}^{\cA'}(5k,p')+B\log^{2}m+\log k\cdot t_{u}^{\cA'}(5k,p'))\\
 & = & O(\frac{B\log k}{k}\cdot(k+t{}_{pre}(15k,p''))+B\log^{2}m+\log k\cdot(\frac{k}{T(k)}+3t{}_{u}(15k,p'')))\\
 & = & O(\frac{B\log k}{k}\cdot t{}_{pre}(15k,p'')+B\log^{2}m+\frac{k\log k}{T(k)}+\log k\cdot t{}_{u}(15k,p'')).
\end{eqnarray*}
\end{proof}

%% file: decomposition.tex
\global\long\def\msfbuild{\textsf{Build}}

\section{$\protect\msf$ Decomposition\label{sec:MST Decomposition}}

In this section, we show an improved algorithm for computing a hierarchical
decomposition of a graph called \emph{$\msf$ decomposition}. This
decomposition is introduced by Wulff-Nilsen \cite[Section 3.1]{Wulff-Nilsen16a}
and it is the main subroutine in the preprocessing algorithm of his
dynamic $\msf$ algorithm and also of ours. Our improved algorithm
has a better trade-off between the running time and the ``quality''
of the decomposition as will be made precise later. The improved version
is obtained simply by using the flow-based \emph{expansion decomposition}
algorithm\footnote{The expansion decomposition algorithm was used as a main preprocessing
algorithm for their dynamic $\sf$ algorithm.} by Nanongkai and Saranurak \cite{NanongkaiS16} as the main subroutine,
instead of using diffusion/spectral-based algorithms as in \cite{Wulff-Nilsen16a}.
Moreover, as the expansion decomposition algorithm is defined based
on \emph{expansion }(which is defined in \ref{sub:decomp given partition})
and not conductance, this is easier to work with and it simplifies
some steps of the algorithm in \ref{sub:decomp given partition}.
Before stating the main result in \ref{thm:MSF decomposition}, we
need the following definition:
\begin{defn}
[Hierarchical Decomposition]For any graph $G=(V,E)$, a \emph{hierarchical
decomposition} $\cH$ of $G$ is a rooted tree. Each node $C\in\cH$
corresponds to some subgraph of $G$ which is called a \emph{cluster}.
There are two conditions that $\cH$ needs to satisfy: 1) the root
cluster of $\cH$ corresponds to the graph $G$ itself, 2) for each
non-leaf cluster $C\in\cH$, let $\{C'_{i}\}_{i}$ be the children
of $C$. Then vertices of $\{C'_{i}\}_{i}$ form a partition of vertices
in $C$, i.e. $V(C)=\Disjunion_{i}V(C'_{i})$. The root cluster is
a \emph{level-$1$} cluster. A child of \emph{level-$i$ cluster}
is a \emph{level-$(i+1)$ cluster}. The \emph{depth} of $\cH$ is
the depth of the tree. Let $E^{C}=E(C)-\Disjunion_{i}E(C'_{i})$ be
the set of edges in $C$ which are not edges in any of $C'_{i}$'s.
We call an edge $e\in E^{C}$ a \emph{$C$-own }edge, and an edge
$f\in E(C)-E^{C}=\Disjunion_{i}E(C'_{i})$ a \emph{$C$-child} edge.
\label{def:hie decomp}
\end{defn}
We note that, for any cluster $C$ with a child $C'$, it is possible
that $E(C')\subsetneq E(C[V(C')])$. That is, there might be some
edge $e=(u,v)\in E(C)$ where $u,v\in V(C')$ but $e\notin E(C')$.
In other words, there can be a $C$-own edge $(u,v)$ where both $u,v\in V(C')$.
Observe the following:
\begin{fact}
Let $\cH$ be a hierarchical decomposition of a graph $G=(V,E)$.
Then $\Disjunion_{C\in\cH}E^{C}=E$.\label{fact:C-edge partition} 
\end{fact}
Throughout this section, we assume that, in an input graph with $m$-edge,
the edges have distinct weights ranging from number $1$ to $m$.
Throughout this section, let $\gamma=n^{O(\sqrt{\log\log n/\log n})}=n^{o(1)}$
where $n$ is the number of nodes in a graph. The main result of this
section is the below theorem:
\begin{thm}
\label{thm:MSF decomposition}There is a randomized algorithm called
\emph{$\msf$ decomposition}, $\msfdecomp$, which takes the following
as input:
\begin{itemize}
\item a connected graph $G=(V,E,w)$ with $n$ nodes, $m$ edges and max
degree 3, where $w:E\rightarrow\{1,\dots,m\}$ is the weight function
of edges in $G$,
\item a failure probability parameter $p\in(0,1]$, a conductance parameter
$\alpha\in[0,1]$, and parameters $d\ge3$, $s_{low}$ and $s_{high}$
where $s_{high}\ge s_{low}$.
\end{itemize}

In time $\tilde{O}(nd\gamma\log\frac{1}{p})$ where $\gamma=n^{O(\sqrt{\log\log n/\log n})}$,
the algorithm returns (i) a graph $G'=(V,E,w')$ with a new weight
function $w':E\rightarrow\mathbb{R}$ and (ii) a hierarchical decomposition
$\cH$ of the re-weighted graph $G'$ with following properties:
\begin{enumerate}
\item For all $e\in E$, $w'(e)\ge w(e)$.
\item $|\left\{ e\in E\mid w(e)\neq w'(e)\right\} |\le\alpha d\gamma n$.\footnote{We can actually prove that $|\left\{ e\in E\mid w(e)\neq w'(e)\right\} |\le\tilde{O}(\alpha\gamma n)$
but this does not improve the running time significantly.}
\item For any cluster $C\in\cH$ and any set of edges $D$, $\msf(C-D)=\Disjunion_{C':\textnormal{child of }C}\msf(C'-D)\disjunion(\msf(C-D)\cap(E^{C}-D))$.
\item $\cH$ has depth at most $d$.
\item A cluster $C$ is a leaf cluster iff $E(C)\le s_{high}$. 
\item Each leaf cluster contains at least $s_{low}/3$ nodes.
\item For level $i$, $|\Disjunion_{C:\textnormal{non-leaf, level-}i}E^{C}|\le n/(d-2)+\alpha\gamma n$.
\item With probability $1-p$, all non-root clusters $C\in\cH$ are such
that $\phi(C)=\Omega(\alpha/s_{low})$.
\end{enumerate}
\end{thm}
We call the lower bound of conductance for all non-root clusters is
the\emph{ conductance guarantee} of the hierarchical decomposition
$\cH$, which is $\Omega(\alpha/s_{low})$ in our algorithm. Compared
with the $\msf$ decomposition algorithm in \cite[Section 3.1]{Wulff-Nilsen16a},
our algorithm runs significantly faster and has a better trade-off
guarantee between conductance of the cluster and the number of edges
re-weighted. In particular, the running time of our algorithm does
not depends on the conductance parameter~$\alpha$. 

Now, we give some intuition why this decomposition can be useful in
our application. Given an input $n$-node graph $G$, we set $\alpha=1/\gamma^{3}$,
$d=\gamma$, $s_{low}=\gamma$, and $s_{high}=n/\gamma$. The algorithm
increases the weight of only $(1/\gamma)$-fraction of edges resulting
in the re-weighted graph $G'$, and then it outputs the hierarchy decomposition
$\cH$ of $G'$. Property 3 of $\cH$ is crucial and it implies that
$\msf(G')=\Disjunion_{C\in\cH}(\msf(C)\cap E^{C})$, and this holds
even after deleting any set of edges. This suggests that, to find
$\msf(G)$, we just need separately find $\msf(C)\cap E^{C}$, i.e.,
the $C$-own edges that are in $\msf(C)$, for every cluster $C\in\cH$.
That is, the task of maintaining the $\msf$ is also ``decomposed''
according the decomposition. Other properties are about bounding
the size of some sets of edges and the conductance of clusters. These
properties will allow our dynamic $\msf$ algorithm to have fast update
time. 

The rest of this section is for proving \ref{thm:MSF decomposition}.

\subsection{Expansion Decomposition that Respects a Given Partition\label{sub:decomp given partition}}

The \emph{expansion decomposition} algorithm \cite{NanongkaiS16}
is an algorithm that, roughly, given a graph $G=(V,E)$, it outputs
a partition $\cQ=\{V_{1},\dots,V_{k}\}$ of $V$ such that, for each
$i$, $G[V_{i}]$ has no sparse cuts and there are not many edges
crossing different parts of $\cQ$. The goal of this section is to
extend the algorithm to ensure that each part $V_{i}$ is not too
small. See \ref{thm:expdecomp respect} for the precise statement.
This requirement is needed for the construction of the $\msf$ decomposition,
as shown in \cite{Wulff-Nilsen16a}, because each leaf cluster must
not be too small (Property 6 in \ref{thm:MSF decomposition}). The
algorithm in this subsection speeds up and simplifies the algorithm
in \cite[Section 6]{Wulff-Nilsen16a} which does the same task. Before
stating \ref{thm:expdecomp respect}, we need the following definition.
\begin{defn}
Let $\cP$ be a partition of set $V$. We say that a set $S\subset V$
\emph{respects} $\cP$ if for each set $U\in\cP$, either $U\subseteq S$
or $U\cap S=\emptyset$. Let $\cQ$ be another partition of $V$.
We say that $\cQ$ respects $\cP$ if, for each set $S\in\cQ$, $S$
respects $\cP$.
\end{defn}
We prove the following extended expansion decomposition algorithm:
\begin{lem}
\label{thm:expdecomp respect}There is a randomized algorithm $\cA$
that takes as inputs a connected graph $G=(V,E)$ with $n$-node
and max degree 3, a partition $\cP$ of $V$ where, for each set of
nodes $U\in\cP$, $G[U]$ is connected and $c_{0}s\le|U|\le s$ for
some constant $c_{0}$, a conductance parameter $\alpha\in[0,1]$,
and a failure probability parameter $p$. Then, in time $O(n\gamma\log\frac{1}{p})$,
$\cA$ outputs a partition $\cQ=\{V_{1},\dots,V_{k}\}$ of $V$ with
the following properties:
\begin{enumerate}
\item $\cQ$ respects $\cP$.
\item $| \{(u,v)\in E\mid u\in V_{i}$ and $v\in V_{j}$ where $i\neq j\} | \le\alpha\gamma n$.
\item For all $V_{i}\in\cQ$, $G[V_{i}]$ is connected. With probability
$1-p$, for all $V_{i}\in\cQ$, $\phi(G[V_{i}])=\Omega(\alpha/s)$. 
\end{enumerate}
\end{lem}
Note that by giving a partition $\cP$ where each $U\in\cP$ has size
around $s$, the algorithm in \ref{thm:expdecomp respect} will output
the partition $\cQ$ where each part $V_{i}\in\cQ$ has size at least
$\Omega(s)$.

To prove this, we use the expansion decomposition algorithm by Nanongkai
and Saranurak \cite{NanongkaiS16} in a black-box manner. Before stating
the algorithm in \ref{thm:global decomp}, we recall that, for any graph
$G=(V,E)$, the \emph{expansion} of $G$ is $h(G)=\min_{S\subset V}\frac{\delta(S)}{\min\{|S|,|V-S|\}}$.
Note the following connection to conductance:
\begin{fact}
In any connected graph $G=(V,E)$ with max degree $\Delta=O(1)$,
$\phi(G)=\Theta(h(G))$.\label{fact:exp vs con}\end{fact}
\begin{proof}
For any set $S\subset V$, we have $vol(S)\ge|S|$ as $G$ is connected,
and $vol(S)\le\Delta|S|$ as $G$ has max degree $\Delta$. So $vol(S)=\Theta(|S|)$.
Hence, $\frac{\delta(S)}{\min\{vol(S),vol(V-S)\}}=\Theta(\frac{\delta(S)}{\min\{|S|,|V-S|\}})$
for all $S\subset V$, and so $\phi(G)=\Theta(h(G))$.\end{proof}
\begin{lem}
[Expansion Decomposition \cite{NanongkaiS16} (Paraphrased)]\label{thm:global decomp}There
is a randomized algorithm $\cA$ that takes as inputs a (multi-)graph
$G=(V,E)$ with $n\ge2$ vertices and $m$ edges and an expansion
parameter $\alpha>0$, and a failure probability parameter $p$. Then,
in $O(m\gamma\log\frac{1}{p})$ time, $\cA$ outputs a partition $\cQ=\{V_{1},\dots,V_{k}\}$
of $V$, 
\begin{enumerate}
\item $| \{(u,v)\in E\mid u\in V_{i}$ and $v\in V_{j}$ where $i\neq j\} | \le\alpha\gamma n$.
\item With probability $1-p$, for all $V_{i}\in\cQ$, $h(G[V_{i}])\ge\alpha$.
\end{enumerate}
\end{lem}
The algorithm for \ref{thm:expdecomp respect} is very simple. Given
a graph $G=(V,E)$ and a partition $\cP$, we contract each set of
nodes $U\in\cP$ into a single node resulting in a contracted graph
$G_{\cP}$. Then run the expansion decomposition in \ref{thm:global decomp}
on $G_{\cP}$ and obtain the partition $\cQ'$ of nodes in $G_{\cP}$.
We just output $\cQ$ which is obtained from $\cQ$ by ``un-contracting''
each set of $\cP$. Now, we show the correctness of this simple approach.
\begin{lem}
Let $G=(V,E)$ be a graph with max degree $\Delta$ and $\cP$ be
a partition of $V$ where, for each set of nodes $U\in\cP$, $G[U]$
is connected and $|U|\le s$. Let $G_{\cP}$ be a graph obtained from
$G$ by contracting each set $U\in\cP$ into a single node. If $h(G)\le\frac{1}{2s}$,
then $h(G_{\cP})\le4s^{2}\Delta h(G)$.\label{lem:preserve exp}\end{lem}
\begin{proof}
Consider a cut $S\subset V$ in $G$ where $\frac{\delta_{G}(S)}{\min\{|S|,|V-S|\}}=h(G)$.
Next, consider another cut $S'\subset V$ in $G$ where $S'$ is the
union of all sets $U\in\cP$ such that $U\subseteq S$. Clearly, $S'\subseteq S$
and so $|V-S'|\ge|V-S|$. We claim that $|S'|\ge|S|/2$. Let $\cP_{S}$
be the collection of set $U$ which contain nodes both in $S$ and
$V-S$. Note that $\delta(S)\ge|\cP_{S}|$ because, for each $U\in\cP_{S}$,
$G[U]$ is connected. We have

\begin{align*}
|S'| & =|S|-\sum_{U\in\cP_{S}}|U\cap S|\\
 & \ge|S|-s|\cP_{s}| & \mbox{as }|U|\le s\\
 & \ge|S|-s\delta_{G}(S)\\
 & \ge|S|-sh(G)|S|\\
 & \ge|S|/2 & \mbox{as }h(G)\le\frac{1}{2s}.
\end{align*}
Next, we bound $\delta_{G}(S')$. Note that 
\begin{align*}
\delta_{G}(S') & \le\delta_{G}(S)+\sum_{U\in\cP_{S}}E(U,S')\\
 & \le\delta_{G}(S)+\Delta s|P_{S}|\\
 & \le(1+\Delta s)\delta_{G}(S)\le2\Delta s\delta_{G}(S).
\end{align*}
Therefore, we have that $h_{G}(S')=\frac{\delta_{G}(S')}{\min\{|S'|,|V-S'|\}}\le\frac{2\Delta s\delta_{G}(S)}{\frac{1}{2}\min\{|S|,|V-S|\}}=4\Delta s\cdot h(G)$.
Next, let $S'_{\cP}$ be a set of nodes in $G_{\cP}$ obtained from
$S'$ by contracting each set $U\in\cP$ into a node. Observe that
$\delta_{G_{\cP}}(S'_{\cP})=\delta_{G}(S')$ because $S'$ respects
$\cP$. Also, $s|S'_{\cP}|\ge|S'|$ and $s|V(G_{\cP})-S'_{\cP}|\ge|V-S|$
because $|U|\le s$ for all $U\in\cP$. So, we can conclude 
\[
h(G_{\cP})\le\frac{\delta_{G_{\cP}}(S'_{\cP})}{\min\{|S'_{\cP}|,|V(G_{\cP})-S'_{\cP}|\}}\le\frac{\delta_{G}(S')}{\frac{1}{s}\min\{|S'|,|V-S'|\}}=s\cdot h_{G}(S')\le4s^{2}\Delta h(G).
\]

\end{proof}

\paragraph{Proof of \ref{thm:expdecomp respect}.}

Now, we are ready the proof the main lemma.
\begin{proof}
[Proof of \ref{thm:expdecomp respect}]The precise algorithm is the
following. Given the input $(G,\cP,\alpha,p)$ where $G$ has max
degree $\Delta=3$, we first construct a multi-graph $G_{\cP}=(V',E')$
obtained from $G$ by contracting each set $U\in\cP$ into a node.
Note that $G_{\cP}$ has at most $\frac{n}{c_{0}s}$ nodes, as $|U|\ge c_{0}s$
for all $U\in\cP$, and $G_{\cP}$ has $O(n)$ edges. Then we run
the expansion decomposition algorithm from \ref{thm:global decomp}
with $(G_{\cP},c_{0}s\alpha,p)$ as inputs, and outputs a partition
$\cQ'=\{V'_{1},\dots,V'_{k}\}$ of $V'$. For each $V'_{i}\in\cQ'$,
let $V_{i}\subseteq V$ be the set obtained from $V'_{i}$ by ``un-contracting''
each set in $U\in\cP$. The algorithm just returns $\cQ=\{V_{1},\dots,V_{k}\}$
as its output. The total running time is $O(n\gamma\log\frac{1}{p})$
by \ref{thm:global decomp} and because other operations take linear
time. Now, we prove the correctness. 

Clearly, $\cQ$ respects $\cP$ by constriction. Also, we have
\begin{align*}
\{(u,v)\in E\mid u\in V_{i},v\in V_{j},i\neq j\} & =\{(u,v)\in E'\mid u\in V'_{i},v\in V'_{j},i\neq j\}\\
 & \le(c_{0}s\alpha)\gamma\frac{n}{c_{0}s}=\alpha\gamma n,
\end{align*}
by \ref{thm:global decomp}. Next, note that $\phi(G[V_{i}])=\Theta(h(G[V_{i}]))$
by \ref{fact:exp vs con} and the fact that $G$ has max degree 3.
So it is enough to show that $h(G[V_{i}])=\Omega(\alpha/s)$ for all
$i$ with probability $1-p$. By \ref{lem:preserve exp}, for all
$1\le i\le k$, we have that if $h(G[V_{i}])<h(G_{\cP}[V'_{i}])/4s^{2}\Delta$,
then $h(G[V_{i}])>1/2s=\Omega(\alpha/s)$ and we are done. So we assume
otherwise, which means that

\[
h(G[V_{i}])\ge h(G_{\cP}[V'_{i}])/4s^{2}\Delta\ge c_{0}s\alpha/4s^{2}\Delta=\Omega(\alpha/s)
\]
for all $1\le i\le k$ with probability $1-p$. We note that we can
additionally make sure that $G[V_{i}]$ is connected with certainty
in linear time. This concludes the proof.
\end{proof}

\subsection{$\protect\msf$ Decomposition Algorithm}

In this section, we just plug the extended version of the expansion
decomposition algorithm from \ref{sub:decomp given partition} to
the approach by Wulff-Nilsen \cite{Wulff-Nilsen16a} for constructing
the $\msf$ decomposition. One minor contribution of this section
is that we present the $\msf$ decomposition in a more modular way
than how it is presented in \cite{Wulff-Nilsen16a}. In particular,
we list and prove all the needed properties of the $\msf$ decomposition
here and hide all the implementation details from the other sections.
In particular, the notion of $M$-clusters (as defined below) is hidden
from other sections. We hope that this facilitate the future applications
of this decomposition.

First, we need the following algorithm by Frederickson:
\begin{lem}
[Frederickson \cite{Frederickson85}]\label{lem:group}There is an
algorithm which takes as input a tree $T=(V,E)$ with $n$ nodes and
max degree 3 and a parameter $s$. Then, in $O(n)$ time, the algorithm
outputs a partition $\cC=\{V_{i}\}_{i}$ of $V$ where $s/3\le|V_{i}|\le s$
and $T[V_{i}]$ is connected for all $i$.
\end{lem}
The algorithm $\msfdecomp(G,p,\alpha,d,s_{low},s_{high})$ for \ref{thm:MSF decomposition}
is as follows. First, we compute the $\msf$ $M$ of $G$. Then, given
$(M,s_{low})$ to \ref{lem:group}, we compute the outputted partition
$\cP_{M}$ called \emph{$M$-partition}. For each $V'\in\cP_{M}$,
we call $M[V']$ an \emph{$M$-cluster}. Denote by $\cC_{M}$ and
$E(\cC_{M})$ the set of $M$-clusters and the union of edges of $M$-clusters
respectively. Note that $\cC_{M}$ is just a forest where each tree
has size between $s_{low}/3$ and $s_{low}$. Next, let $d'=d-2$.
For $1\le i\le d'$, we denote by $E_{i}$ the set of edges of weights
in the range $(m-i\frac{m}{d'},m-(i-1)\frac{m}{d'}]$ which are \emph{not
}in $E(\cC_{M})$. For any $i>d'$, let $E_{i}=\emptyset$. Finally,
we call $\msfbuild(G,\cP_{M},1)$ from \ref{alg:MSF decomp build}.
Let $\textsf{ExpDecomp}$ denote the extended expansion decomposition
algorithm from \ref{thm:expdecomp respect}.

\begin{algorithm}
\caption{\label{alg:MSF decomp build}$\protect\msfbuild(C,\protect\cP,i)$}

The numbers $n,m,p,\alpha,d,s_{low},s_{high}$ are fixed by the input
of \ref{thm:MSF decomposition}.
\begin{enumerate}
\item If $|E(C)|\le s_{high}$, return. // i.e. $C$ is a leaf cluster.
\item Else, // i.e. $C$ is a non-leaf cluster.

\begin{enumerate}
\item Compute $\cQ=\textsf{ExpDecomp}(C,\cP,\alpha,p/2n)$. Write $\cQ=\{V^{1},\dots,V^{k}\}$.
\item For all $j\le k$, set $C^{j}=(V^{j},E^{j})$ where $E^{j}=E(C[V^{j}])-E_{i}$
as a child cluster of $C$.
\item Set $E^{C}=E(C)-\Disjunion_{j}E^{j}$ as a set of $C$-own edges. 
\item For all $e\in E^{C}$, set $w'(e)\gets\max\{w(e),m-i\frac{m}{d'}+0.5\}$
\item For all $j\le k$, run $\msfbuild(C^{j},\cP^{j},i+1)$ where $\cP^{j}=\{V'\in\cP\mid V'\subseteq V^{j}\}$.\end{enumerate}
\end{enumerate}
\end{algorithm}

\subsubsection{Analysis}

First, we note that, in Step 2.a, a valid input is given to \ref{thm:expdecomp respect}:
\begin{prop}
In the recursion by invoking $\msfbuild(G,\cP_{M},1)$, if $\msfbuild(C,\cP,i)$
is called, then $\cP$ is a partition of $V(C)$. Moreover, $\cP\subseteq\cP_{M}$.\label{prop:respect P_M}\end{prop}
\begin{proof}
We prove by induction. The base case is trivial because $\cP_{M}$
is a partition of $V(G)$. Next, by induction hypothesis, suppose
that $\cP$ is a partition of $V(C)$ and $\cP\subseteq\cP_{M}$.
So $(C,\cP,\alpha,p/n)$ is a valid input for the algorithm from \ref{thm:expdecomp respect}
in Step 2.a. Let $\cQ=\{V^{1},\dots,V^{k}\}$ be the outputted partition
of $V(C)$. By \ref{thm:expdecomp respect}, $\cQ$ respects $\cP$.
Therefore, for all $j$, $\cP^{j}=\{V'\in\cP\mid V'\subseteq V^{j}\}$
in Step 2.e is actually a partition of $V^{j}.$ That is, for all
child clusters $C^{j}$ of $C$, when $\msfbuild(C^{j},\cP^{j},i)$
is called, $\cP^{j}$ is a partition of $V(C^{j})=V^{j}$ and $P^{j}\subseteq\cP\subseteq\cP_{M}$.
\end{proof}
As all the steps are valid, we obtain a hierarchical decomposition
denoted by $\cH$ with the following basic properties.
\begin{prop}
We have the following:\label{prop:basic cH}
\begin{enumerate}
\item $\msfbuild(G,\cP_{M},1)$ returns a hierarchical decomposition $\cH$.
\item For any cluster $C\in\cH$, $V(C)$ respects $\cP_{M}$. 
\item A cluster $C\in\cH$ is a leaf cluster iff $C$ has at most $s_{high}$
edges. Moreover, each leaf cluster contains at least $s_{low}/3$
nodes.
\item With probability $1-p$, all non-root clusters $C\in\cH$ are such
that $\phi(C)=\Omega(\alpha/s_{low})$. 
\end{enumerate}
\end{prop}
\begin{proof}
(1): There are two conditions we need to show about $\cH$. First,
the root cluster clearly corresponds to the graph $G$ itself. Next,
for any non-leaf cluster $C$, let $C^{1},\dots,C^{k}$ be the children
of $C$. We have that $\{V(C^{1}),\dots,V(C^{k})\}$ is a partition
of $V(C)$ by \ref{thm:expdecomp respect} used in Step 2.a.

(2): This follows from \ref{prop:respect P_M}.

(3): The first statement is by Step 1. For the second statement, for
all clusters $C\in\cH$, $V(C)$ respects $\cP_{M}$. So $|V(C)|\ge s_{low}/3$. 

(4): By Step 2.a, all non-root clusters $C$ is outputted from \ref{thm:expdecomp respect}.
Since \ref{thm:expdecomp respect} is called at most $2n$ times,
the claim holds with probability $1-2n\cdot\frac{p}{2n}=1-p$.
\end{proof}
In the outputted hierarchical decomposition $\cH$, observe that $C\in\cH$
is a level-$i$ cluster iff $\msfbuild(C,\cP,i)$ is called in the
level-$i$ recursion when we call $\msfbuild(G,\cP_{M},1)$. For any
$i$, let $E_{\ge i}=\Disjunion_{j\ge i}E_{j}$. Note that, for $i>d'$,
$E_{\ge i}=\emptyset$. By Step 2.b of \ref{alg:MSF decomp build},
observe the following:
\begin{prop}
For any level-$i$ cluster $C$, $E(C)\subseteq E_{\ge i}\disjunion E(\cC_{M})$.\label{thm:edge of level-i}\end{prop}
\begin{lem}
$\cH$ has depth at most $d$.\label{lem:msfdecomp depth}\end{lem}
\begin{proof}
Recall that $d=d'+2$. Let $C'$ be a level-$(d'+2)$ cluster and
$C$ be the parent cluster of $C'$. Since $C$ is a level-$(d'+1)$
cluster, by \ref{thm:edge of level-i}, $E(C)\subseteq E(\cC_{M})$.
This means that, $C$ is a forest. As $C'$ is connected by \ref{thm:expdecomp respect},
$C'$ cannot intersect more than two $M$-clusters. So $|V(C')|\le s_{low}$
and, hence, $|E(C')|\le s_{low}$. As $s_{low}\le s_{high}$, $C'$
must be returned as a leaf cluster by Step 1. 
\end{proof}
For any level-$i$ non-leaf cluster $C$, let $E_{i}(C)=E_{i}\cap E(C)$
and let $\partial^{C}$ be the set of edges whose endpoints are in
different child clusters. By Step 2.b and Step 2.c, observe the following:
\begin{prop}
For any level-$i$ non-leaf cluster $C$, the set of $C$-own edges
is $E^{C}=E_{i}(C)\cup\partial^{C}$.\label{prop:C-own edges}
\end{prop}
We note that \ref{prop:C-own edges} is \emph{not} true for a leaf
cluster $C$ because $E^{C}$ may contains some $M$-cluster edges
(i.e. $E(\cC_{M})$).
\begin{lem}
$|\left\{ e\in E\mid w(e)\neq w'(e)\right\} |\le\alpha d\gamma\cdot n$.\label{lem:bound reweight}\end{lem}
\begin{proof}
For any level-$i$ cluster $C\in\cH$, we claim that $|\left\{ e\in E^{C}\mid w(e)\neq w'(e)\right\} |\le\alpha\gamma\cdot|V(C)|$.
Having this claim, the lemma follows because the recursion depth is
at most $d$ by \ref{lem:msfdecomp depth}, and, for any depth $i$,
any two level-$i$ clusters $C$ and $C'$ are node-disjoint.

Suppose that $C$ is a level-$i$ cluster. If $C$ is a leaf cluster,
then $\left\{ e\in E^{C}\mid w(e)\neq w'(e)\right\} =\emptyset$.
So we assume $C$ is a non-leaf cluster. By \ref{prop:C-own edges}
$E^{C}=E_{i}(C)\cup\partial^{C}$. For any edge $e\in E_{i}(C)$,
we have $w(e)\ge m-i\frac{m}{d'}+1$, so $w'(e)=w(e)$ by Step 2.d.
So $\left\{ e\in E(C)\mid w(e)\neq w'(e)\right\} \subseteq\partial(C)$.
By \ref{thm:expdecomp respect}, $|\partial^{C}|\le\alpha\gamma\cdot|V(C)|$.
So this concludes the claim.\end{proof}
\begin{lem}
For level $i$, $|\Disjunion_{C:\textnormal{non-leaf, level-}i}E^{C}|\le n/(d-2)+\alpha\gamma n$.\end{lem}
\begin{proof}
Let $\cC_{i}$ be the set of level-$i$ non-leaf clusters. By \ref{prop:C-own edges},
$\Disjunion_{C\in\cC_{i}}E^{C}=\Disjunion_{C\in\cC_{i}}E_{i}(C)\cup\partial^{C}\subseteq E_{i}\cup\Disjunion_{C\in\cC_{i}}\partial^{C}$.
We have $|E_{i}|\le n/d'=n/(d-2)$. Also, by \ref{thm:expdecomp respect},
$|\partial^{C}|\le\alpha\gamma\cdot|V(C)|$ and hence $|\Disjunion_{C\in\cC_{i}}\partial^{C}|\le\alpha\gamma n$
because any two level-$i$ clusters $C$ and $C'$ are node-disjoint.\end{proof}
\begin{lem}
\label{lem:respect MSF}For any cluster $C\in\cH$ and any set of
edges $D$, 
\[
\msf(C-D)=\Disjunion_{C':\textnormal{child of }C}\msf(C'-D)\disjunion(\msf(C-D)\cap(E^{C}-D)).
\]
\end{lem}
\begin{proof}
We only prove that $\Disjunion_{C':\textnormal{child of }C}\msf(C')\disjunion(\msf(C)\cap E^{C})=\msf(C)$.
The lemma follows by observing that the argument holds true even when
the set of edges $D$ are removed from all clusters. We write
\begin{eqnarray*}
E(C) & = & E^{C}\disjunion\Disjunion_{C':\textnormal{child of }C}E(C')\\
 & = & E^{C}\disjunion(\Disjunion_{C':\textnormal{child of }C}E(C')\cap E(\cC_{M}))\disjunion(\Disjunion_{C':\textnormal{child of }C}E(C')-E(\cC_{M})).
\end{eqnarray*}

First, we know that edges in $\Disjunion_{C':\textnormal{child of }C}E(C')\cap E(\cC_{M})$
are tree-edges in $\msf(C)$ because $E(\cC_{M})\subseteq\msf(G)$
and $C$ is a subgraph of $G$ where some edges not in $E(\cC_{M})$
have their weight increased. This means that we can construct $\msf(C)$
using an instance $I$ of Kruskal's algorithm where the initial forest
is the edges in $\Disjunion_{C':\textnormal{child of }C}E(C')\cap E(\cC_{M})$. 

Next, suppose that $C$ is a level-$i$ cluster. So $\Disjunion_{C':\textnormal{child of }C}E(C')-E(\cC_{M})\subseteq E_{\ge i+1}$
by \ref{thm:edge of level-i}. By Step 2.d, any $C$-own edge $e\in E^{C}$
is heavier than any $E_{\ge i+1}$. So the instance $I$ will scan
all edges in $\Disjunion_{C':\textnormal{child of }C}E(C')-E(\cC_{M})$
before any edges in $E^{C}$. After finishing scanning $\Disjunion_{C':\textnormal{child of }C}E(C')-E(\cC_{M})$,
$I$ has constructed, as a part of $\msf(C)$, the following: 
\begin{align*}
 & \msf\left((\Disjunion_{C':\textnormal{child of }C}E(C')\cap E(\cC_{M}))\disjunion(\Disjunion_{C':\textnormal{child of }C}E(C')-E(\cC_{M}))\right)\\
 & =\msf(\Disjunion_{C':\textnormal{child of }C}E(C'))\\
 & =\Disjunion_{C':\textnormal{child of }C}\msf(C') & \mbox{as }C'\mbox{'s are node disjoint}.
\end{align*}
Since there are only the edges in $E^{C}$ that $I$ have not scanned
yet, this means that 
\[
\Disjunion_{C':\textnormal{child of }C}\msf(C')\disjunion(\msf(C)\cap E^{C})=\msf(C).
\]
\end{proof}
\begin{lem}
The $\msf$ decomposition algorithm $\msfdecomp(G,p,\alpha,d,s_{low},s_{high})$
runs in time $\tilde{O}(nd\gamma\log\frac{1}{p})$.\label{lem:MSF decomp runtime}\end{lem}
\begin{proof}
The bottleneck is the time for calling $\msfbuild(G,\cP_{M},1)$.
For any level-$i$ cluster $C\in\cH$, when $\msfbuild(C,\cP,i)$
is called, this takes time $\tilde{O}(\alpha\gamma|V(C)|\log\frac{n}{p})$
excluding the time in the further recursion. Again, the lemma follows
because the recursion depth is at most $d$ by \ref{lem:msfdecomp depth},
and, for any depth $i$, any two level-$i$ clusters $C$ and $C'$
are node-disjoint.
\end{proof}
\ref{prop:basic cH}, \ref{lem:msfdecomp depth}, \ref{lem:bound reweight}
and \ref{lem:respect MSF} concludes the correctness of \ref{thm:MSF decomposition}.
\ref{lem:MSF decomp runtime} bounds the running time.

%% file: MSF_algorithm.tex
\section{Dynamic $\protect\msf$ Algorithm\label{sec:Dynamic MSF}}

In this section, we prove the main theorem:

\begin{thm}
\label{thm:dyn MST final}There is a fully dynamic $\msf$ algorithm
on an $n$-node $m$-edge graph that has preprocessing time $O(m^{1+O(\sqrt{\log\log m/\log m})}\log\frac{1}{p})=O(m^{1+o(1)}\log\frac{1}{p})$
and worst-case update time $O(n^{O(\log\log\log n/\log\log n)}\log\frac{1}{p})=O(n^{o(1)}\log\frac{1}{p})$
with probability $1-p$.
\end{thm}

By using a standard reduction or a more powerful reduction from \ref{thm:reduc restricted dec},
it is enough to show the following:

\begin{lem}
\label{lem:dec MST final}There is a decremental $\msf$ algorithm
$\cA$ on an $n$-node $m$-edge graph $G$ with max degree $3$ undergoing
a sequence of edge deletions of length $T=\Theta(n^{1-O(\log\log\log n/\log\log n)})$.
$\cA$ has preprocessing time $O(n^{1+O(\sqrt{\log\log n/\log n})}\log\frac{1}{p})$
and worst-case update time $O(n^{O(\log\log\log n/\log\log n)}\log\frac{1}{p})$
with probability $1-p$. 
\end{lem}

We note that essentially all the ideas in this section, in particular
the crucial definition of \emph{compressed clusters}, already appeared
in Wulff-Nilsen \cite{Wulff-Nilsen16a}. In this section, we only
make sure that, with our improved tools from previous sections, we
can integrate all of them using the same approach as in \cite{Wulff-Nilsen16a}.
Obviously, the run time analysis must change because our algorithm
is faster and need somewhat more careful analysis. Although the correctness
follows as in \cite{Wulff-Nilsen16a}, the terminology changes a bit
because $\msf$ decomposition from \ref{thm:MSF decomposition} is
presented in a more modular way. 

The high-level idea in \cite{Wulff-Nilsen16a} of the algorithm $\cA$
is simple. To maintain $\msf(G)$, we maintain a graph $H$, called
the \emph{sketch graph}, where at any time $\msf(G)=\msf(H)$ and $H$
contains only few non-tree edges with high probability. Then we just
maintain $\msf(H)$ using another algorithm for graphs with few non-tree
edges. 

\paragraph{Organization.}
The rest of this section is for proving \ref{lem:dec MST final}.
In \ref{sec:MSF_algorithm}, we describe the whole algorithm which combines all the tools from previous sections.
The preprocessing algorithm is in \ref{sec:MSF_preprocess} and the update algorithm is in \ref{sec:MSF_update}. 
We summarize all the main notations in \ref{table:def}.
Next, in \ref{sec:MSF_correct}, we show that the sketch graph $H$ is indeed maintained 
such that $\msf(G)=\msf(H)$ (shown in \ref{sec:sketch preserve}) 
and $H$ has few non-tree edges (shown in \ref{sec:sketch sparse}). 
Note that this implies that $\msf(G)$ is correctly maintained.
Lastly, we analyze the running time in \ref{sec:MSF_time}. 
We first bound the preprocessing time in \ref{sec:time_pre}, the time
needed for maintaining the sketch graph $H$ itself in \ref{sec:time_sketch}, 
and the time needed for maintaining $\msf(H)$ in $H$ in \ref{sec:time_msf(H)}.
We put everything together and conclude the proof in \ref{sec:wrap up}.

\subsection{The Algorithm}\label{sec:MSF_algorithm}

For any number $m$ and $p\in(0,1)$, in this section, the goal is
to describe the decremental $\msf$ algorithm $\cA(m,p)$ for any
$m$-edge graph $G=(V,E,w)$ with max degree 3 such that $\cA(m,p)$
can handle $T(m)$ edge deletions and, with probability $1-p$, has
preprocessing and update time $t_{pre}(m,p)$ and $t_{u}(m,p)$. We
will show that $t_{pre}(m,p)=O(m^{1+O(\sqrt{\log\log m/\log m})}\log\frac{1}{p})$,
$t_{u}(m,p)=O(m^{O(\log\log\log m/\log\log m)}\log\frac{1}{p})$ and
$T(m)=\Theta(m^{1-O(\log\log\log m/\log\log m)})$. This will imply
\ref{lem:dec MST final}.

By induction on $m$, we assume that, for any $m_{0}\le m-1$ and
$p_{0}\in(0,1)$, we have obtained the decremental $\msf$ algorithm
$\cA(m_{0},p_{0})$ that can run on any $m_{0}$-edge graph $G_{0}$
with max degree $3$ undergoing a sequence edge deletions of length
$T(m_{0})$. Let $t_{pre}(m_{0},p_{0})$ and $t_{u}(m_{0},p_{0})$
denote the preprocessing and update time of $\cA$ on $G_{0}$, respectively,
that hold with probability at least $1-p_{0}$.

By this assumption, \ref{thm:reduc restricted dec} implies the following.
For any number $m_{1}$, $k_{1}$, $B_{1}$ and $p_{1}\in(0,1)$ where
$k_{1}\le(m-1)/15$, there is a fully dynamic algorithm $\cA_{few}(m_{1},k_{1},B_{1},p_{1})$
that can run on any (multi-)graph $G_{1}$ with at most $m_{1}$ edges
and at most $k_{1}$ non-tree edges. Moreover $\cA_{few}(m_{1},k_{1},B_{1},p_{1})$
can handle inserting a batch of edges of size $B_{1}$. Let $t{}_{pre}^{few}(m_{1},k_{1},B_{1},p_{1})$,
$t{}_{ins}^{few}(m_{1},k_{1},B_{1},p_{1})$, $t{}_{del}^{few}(m_{1},k_{1},B_{1},p_{1})$
denote the preprocessing time, the batch insertion time, and the deletion
time of $\cA_{few}(m_{1},k_{1},B_{1},p_{1})$ respectively, that hold
with probability at least $1-p_{1}$. Below, we will slightly abuse
notation. For any graph $G_{2}$, and parameters $B_{2}$ and $p_{2}$,
we denote $\cA_{few}(G_{2},B_{2},p_{2})$ as an instance of $\cA_{few}(m_{2},k_{2},B_{2},p_{2})$
running on $G_{2}$ with at most $m_{2}$ edges and $k_{2}$ non-tree
edges.

In the following subsections, we will first describe how we preprocess
the input graph $G$ for $\cA(m,p)$ in \ref{sec:MSF_preprocess}. In the process, we introduce
several definitions related to \emph{compressed clusters }which were
defined in \cite{Wulff-Nilsen16a} and will be the central definitions
of our algorithm. Then, we describe how we update in \ref{sec:MSF_update}.

\subsubsection{Preprocessing and Definitions Related to Compressed Clusters}
\label{sec:MSF_preprocess}

Let $n$ be the number of nodes in $G$. We can assume that $G$ is
initially connected otherwise we run the algorithm on each connected
component of $G$. So $n=\Theta(m)$ initially. Since we will handle
only $T(m)=o(m)$ edge deletions, we have $n=\Theta(m)$ at all time. 

Let $\gamma=n^{O(\sqrt{\log\log n/\log n})}$ be the factor from \ref{thm:MSF decomposition}.
We run the $\msf$ decomposition algorithm $\msfdecomp(G,\alpha,p,d,s_{low},s_{high})$
where $\alpha=1/\gamma^{3}$, $d=\gamma$, $s_{low}=\gamma$, and
$s_{high}=n/\gamma$. So we obtain a re-weighted graph $G'=(V,E,w')$
together with its hierarchical decomposition $\cH$ with conductance guarantee
$\alpha_{0}=\Omega(\alpha/s_{low})=\Omega(1/\gamma^{4})$. We denote
by $E^{\neq}=\{e\in E\mid w(e)\neq w'(e)\}$ the set of re-weighted
edges. Let $E^{\neq}(w)$ and $E^{\neq}(w')$ be the set of weighted
edges from $E^{\neq}$ where the weight of $e$ is $w(e)$ and $w'(e)$
respectively. By \ref{thm:MSF decomposition}, we know that a cluster
$C\in\cH$ is a leaf cluster iff $E(C)\le s_{high}$ (before any edge
deletion). For convenience, we call each leaf cluster a \emph{small
cluster} and non-leaf cluster a \emph{large cluster}. For any child
cluster $C'$ of $C$, we say $C'$ is a small (large) child of $C$
iff $C'$ is a small (large) cluster. 
\begin{prop}
Any cluster $C\in\cH$ has at most $O(n/s_{high})$ large child clusters.\label{fact:bound large children}
\begin{proof}
All the large children of $C$ are edge-disjoint and each of them
contains at least $s_{high}$ edges.
\end{proof}
\end{prop}
Let $G_{\smalltext}=\Disjunion_{C:\smalltext}C$. be the union of
all small clusters. We maintain the $\msf$ $M_{\smalltext}$ of $G_{\smalltext}$
by separately initializing $\cA(C,p)$ on each small cluster $C$. 
(Note that we $E(C) \le s_{high} \le m-1$ and so we can initialize $\cA$ on $C$ by our assumption.) 
For any large cluster $C$, let $M_{\smalltext}(C)=\Disjunion_{C':\textnormal{small child of }C}\msf(C')$.
As every small cluster has a unique parent and small clusters are
node-disjoint, we have the following: 
\begin{prop}
$M_{\smalltext}=\Disjunion_{C:\largetext}M_{\smalltext}(C)$. \label{fact:disjoint M_small(C) }
\end{prop}
For each large cluster $C\in\cH$ excluding the root cluster, we initialize
the dynamic pruning algorithm $\pruning$ on $C$ using \ref{thm:pruning detect failure}
with a conductance parameter $\alpha_{0}$. Whenever some edge in
$C$ is deleted, $\pruning$ will update a set $P_{0}^{C}\subseteq V(C)$
of nodes in $C$. Initially, $P_{0}^{C}=\emptyset$. Let $\pi$ denote
the update time of $\pruning$ on each $C$. Given that the sequence
of edge deletions in $C$ is at most $O(\alpha_{0}^{2}|E(C)|)$ (as
we will show later), as $\alpha_{0}=1/n^{O(\sqrt{\frac{\log\log n}{\log n}})}$,
we have 
\[
\pi=n^{O(\log\log(\sqrt{\frac{\log n}{\log\log n}})/\log(\sqrt{\frac{\log n}{\log\log n}}))} \log\frac{1}{p}
=n^{O(\log\log\log n/\log\log n)} \log\frac{1}{p}
\]
by \ref{thm:pruning detect failure}.

For each large cluster (including the root cluster) $C\in\cH$, it is
more convenient to define $P^{C}$ as a union of $P_{0}^{C'}$ over
all large child cluster $C'$ of $C$. That is, $P^{C}=\Disjunion_{C':\textnormal{large child of }C}P_{0}^{C'}$.
We call $P^{C}$ the \emph{total pruning set} of $C$. Recall the definition
of $C$-own edges $E^{C}$ from \ref{def:hie decomp}. For any set
$U\subseteq V(C)$, denote by $\incident^{C}(U)=\{(u,v)\in E^{C}\mid u\in U$
or $v\in U\}$ the set of $C$-own edges incident to $U$.

Now, we define an important definition called \emph{compressed clusters}.
\begin{defn}\label{def:compressed cluster}
For any large cluster $C$ and any set of nodes $U\subseteq\Disjunion_{C':\text{large child of }C}V(C')$,
the \emph{compressed cluster of $C$ with respect to $U$}, denoted
by $\overline{C}(U)$ is obtained from $C$ by 1) replacing  each
small child $C'$ by $\msf(C')$, and 2) contracting nodes in each
large child cluster $C'$ into a single node (called \emph{super node}),
and 3) removing edges (used to) incident to $U$, i.e. removing $\incident^{C}(U)$.
\end{defn}
For convenience, we define the contraction in the step 2 above such
that all $C$-own edges are preserved. That is, all the self loops
are removed except the ones which are $C$-own edges. Let $E^{\overline{C}(U)}=E^{C}-\incident^{C}(U)$
be the set of $C$-own edges which is not incident to $U$. The following
observation shows some basic structure of $\overline{C}(U)$:
\begin{prop}
\label{fact:edge of compressed}For any large cluster $C\in\cH$ and
$U\subseteq V(C)$, we have 
\begin{itemize}
\item $E(\overline{C}(U))=M_{\smalltext}(C)\disjunion E^{\overline{C}(U)}$,
and
\item $M_{\smalltext}(C)\subseteq\msf(\overline{C}(U))$.
\end{itemize}
\end{prop}
\begin{proof}
For the first statement, we partition edges in the cluster $C$ into
$C$-child edges in small children of $C$, $C$-child edges in large
children of $C$, and $C$-own edges. Recall the definitions from \ref{def:hie decomp}. That is, 
\[
E(C)=\Disjunion_{C':\text{small child of }C}E(C')\disjunion\Disjunion_{C':\text{large child of }C}E(C')\cup E^{C}.
\]
We show how $E(C)$ is changed during the process of constructing $\overline{C}(U)$.
First, replacing each small child $C'$ of $C$ by $\msf(C')$ is
to replace $\Disjunion_{C':\text{small child of }C}E(C')$ by $M_{\smalltext}(C)$.
Second, contracting the large children of $C$ is to remove $\Disjunion_{C':\text{large child of }C}E(C')$.
Third, as $U\subseteq\Disjunion_{C':\text{large child of }C}V(C')$,
removing edges incident to $U$ is to replace $E^{C}$ by $E^{\overline{C}(U)}$.
So we have that $E(\overline{C}(U))=M_{\smalltext}(C)\disjunion E^{\overline{C}(U)}$.

For the second statement, by \ref{thm:MSF decomposition}, we have
\[
\msf(C)\supseteq\Disjunion_{C':\text{child of }C}\msf(C')=M_{\smalltext}(C)\disjunion\Disjunion_{C':\text{large child of }C}\msf(C').
\]
Let $\widehat{C}$ be obtained from $C$ after replacing each small
child $C'$ by $\msf(C')$ and contracting large children. Note that
$\overline{C}(U)$ can be obtained from $\widehat{C}$ by removing
all edges incident to $U$. We claim that $M_{\smalltext}(C)\subseteq\msf(\widehat{C})$.
Indeed, contracting large children of $C$ is the same as contracting
edges in $\Disjunion_{C':\text{large child of }C}\msf(C')$. But $\Disjunion_{C':\text{large child of }C}\msf(C')\subseteq\msf(C)$,
so the remaining $\msf$-edges do not change. So $M_{\smalltext}(C)\subseteq\msf(\widehat{C})$.
As $U\subseteq\Disjunion_{C':\text{large child of }C}V(C')$, the
set of edges incident to $U$ is disjoint from $M_{\smalltext}(C)$,
so $M_{\smalltext}(C)\subseteq\msf(\overline{C}(U))$.
\end{proof}
Let $S_{\overline{C}(U)}\subseteq V(\overline{C}(U))$ be the set
of {\em super nodes} in $\overline{C}(U)$. By \ref{fact:bound large children}
$C$ has at most $O(n/s_{high})=O(\gamma)$ large child clusters,
so we have the following:
\begin{prop}
$|S_{\overline{C}(U)}|=O(\gamma)$.\label{prop:super node few}
\end{prop}
Next, we partition $E^{\overline{C}(U)}=E_{1}^{\overline{C}(U)}\disjunion E_{2}^{\overline{C}(U)}\disjunion E_{3}^{\overline{C}(U)}$
where $E_{i}^{\overline{C}(U)}$ is the set of edges $e\in E^{\overline{C}(U)}$
where $(i-1)$ endpoints of $e$ are incident to $S_{\overline{C}(U)}$.
Let $\overline{C}_{i}(U)=(V(\overline{C}),M_{\smalltext}(C)\disjunion E_{i}^{\overline{C}(U)})$
for all $i=1,2,3$. The reason that it is useful to partition $E^{\overline{C}(U)}$
into three parts is because, for $i\in\{2,3\}$, there is a small
set of nodes that ``cover'' all non-tree edges in $\overline{C}_{i}(U)$:
\begin{prop}
For $i\in\{2,3\}$, all non-tree edges in $\overline{C}_{i}(U)$ are
incident to $S_{\overline{C}}$. \label{prop:super node cover}\end{prop}
\begin{proof}
By \ref{fact:edge of compressed}, we have that $M_{\smalltext}(C)\subseteq\msf(\overline{C}_{i}(U))$.
So all non-tree edges in $\overline{C}_{i}(U)$ can only be edges
in $E_{i}^{\overline{C}(U)}$. By definition of $E_{i}^{\overline{C}(U)}$
for $i\in\{2,3\}$, each edge $e\in E_{i}^{\overline{C}(U)}$ is incident
to $S_{\overline{C}}$.
\end{proof}
In the algorithm, for each large cluster $C$, what we really maintain
are always the compressed clusters with respect to $P^{C}$. We define
them with respect to any set $U$ just for the analysis. So we denote
$\overline{C}=\overline{C}(P^{C})$ and call it simply the \emph{compressed
cluster of $C$}. Also, $E^{\overline{C}}$, $E_{i}^{\overline{C}}$
and $\overline{C}_{i}$ are similarly defined, for $i=1,2,3$. Although
two arbitrary clusters $C$ and $D$ in $\cH$ may be not edge-disjoint,
we have that this is the case for compressed clusters.
\begin{prop}
\label{prop:compressed edge disj}Compressed clusters are edge-disjoint.
That is, for any large clusters $C,D\in\cH$, $E(\overline{C})\cap E(\overline{D})=\emptyset$.\end{prop}
\begin{proof}
$M_{\smalltext}(C)$ and $M_{\smalltext}(D)$ are disjoint by \ref{fact:disjoint M_small(C) }.
For \emph{any} set $U\subseteq V(C)$ and $U'\subseteq V(D)$, $E^{\overline{C}(U)}$
and $E^{\overline{D}(U')}$ are disjoint. This follows because $E^{\overline{C}(U)}\subseteq E^{C}$,
$E^{\overline{D}(U')}\subseteq E^{D}$, and $E^{C}\cap E^{D}=\emptyset$
by \ref{fact:C-edge partition}. So, by \ref{fact:edge of compressed},
$E(\overline{C}(U))\cap E(\overline{D}(U'))=\emptyset$.
\end{proof}
For each large cluster $C$, we maintain $\msf(\overline{C}_{1})$
using $\cA_{few}(\overline{C}_{1},1,p)$. Next, we maintain $\msf(\overline{C}_{2})$
using an instance of the algorithm $\cA_{2}$ from \ref{lem:non-tree cover}.
Note the constraint in \ref{lem:non-tree cover} is satisfied.
Indeed, \ref{prop:super node cover} implies that every non-tree edge
in $\overline{C}_{2}$ has exactly one endpoint in the set of super
nodes $S_{\overline{C}}$. Moreover, every node $u\in V(\overline{C}_{2})\setminus S_{\overline{C}}$
has degree at most $3$ just because $G$ has max degree $3$. Next,
we maintain $\msf(\overline{C}_{3})$ using instances of the algorithm
$\cA_{3}$ from \ref{lem:MSF in multigraph}. 
\begin{rem}
\label{rem:nodes in C_3}Note that edges in $\overline{C}_{3}$ consist
of $M_{\smalltext}(C)$ and $E_{3}^{\overline{C}}$, and they do not
share endpoints. So $\msf(\overline{C}_{3})=M_{\smalltext}(C)\disjunion\msf(E_{3}^{\overline{C}})$.
As $M_{\smalltext}(C)$ is already maintained, it is enough to maintain
$\msf(E_{3}^{\overline{C}})$. So, actually, we run $\cA_{3}$ on
the graph consisting of edges from $E_{3}^{\overline{C}}$. This graph
has only $O(\gamma)$ nodes by \ref{prop:super node few}.
\end{rem}
Now, we describe the main object of our algorithm. The sketch graph is 
$H = (V,E(H))$ where 
\begin{align}
E(H)=E^{\neq}(w)\cup M_{\smalltext}\cup\bigcup_{C:\largetext}(\incident^{C}(P^{C})\cup\bigcup_{i=1,2,3}\msf(\overline{C}_{i})\cup J^{C}).
\label{eq:sketch graph}
\end{align}
where $J^{C}\subseteq E^{\overline{C}}$ is called a set of \emph{junk
edges} of a large cluster $C$. Initially, $J^{C}=\emptyset$ for
all large cluster $C$. We will describe how $J^{C}$ is updated later.
Note that $H$ can be a multigraph because of $E^{\neq}(w)$.
\begin{rem}
For each edge $\msf(\overline{C}_{i})$, we include its original endpoints
into $E(H)$ and not the endpoint in the compressed cluster $\overline{C}$
where some nodes are already contracted as one node. This can be done easily by associating
the original endpoints of each edge whenever we contract some nodes.
\end{rem}
The last step of our preprocessing algorithm is to initialize
$\cA_{few}(H,B,p)$ on $H$ and obtain $\msf(H)$. We summarize the
preprocessing algorithm in \ref{alg:MSF prep}.

\begin{algorithm}
\caption{\label{alg:MSF prep}Preprocessing algorithm. 
%	$\msfdecomp$ is from \ref{thm:MSF decomposition}. 
%	$\cA_{few}$ is obtained using \ref{thm:reduc restricted dec}.
%	$\pruning$ is from \ref{thm:pruning detect failure}.
%	$\cA_{2}$ is from \ref{lem:non-tree cover} and $\cA_{3}$ is from \ref{lem:MSF in multigraph}
}

\begin{enumerate}
\item $(G',\cH)=\msfdecomp(G,\alpha,p,d,s_{low},s_{high})$ where $\alpha=1/\gamma^{3}$,
$d=\gamma$, $s_{low}=\gamma$, and $s_{high}=n/\gamma$.
\item Initialize $\cA(C,p)$ for each small cluster $C$ and obtain
$M_{\smalltext}=\msf(G_{\smalltext})$ where $G_{\smalltext}=\Disjunion_{C:\smalltext}C$.
\item For each large cluster $C$ (excluding root cluster), initialize $\pruning$
on $C$ with a conductance parameter $\alpha_{0}=\Omega(1/\gamma^{4})=1/n^{O(\sqrt{\frac{\log\log n}{\log n}})}$.
\item For each large cluster $C$, 

\begin{enumerate}
\item construct $\overline{C}_{1}$, $\overline{C}_{2}$, and $\overline{C}_{3}$.
\item Initialize $\cA_{few}(\overline{C}_{1},1,p)$ and obtain $\msf(\overline{C}_{1})$.
\item Initialize $\cA_{i}(\overline{C}_{i})$ and obtain $\msf(\overline{C}_{i})$,
for $i=2,3$.
\end{enumerate}
\item Construct the sketch graph $H$ where $V(H)=V$ and 
\[
E(H)=E^{\neq}(w)\cup M_{\smalltext}\cup\bigcup_{C:\largetext}(\incident^{C}(P^{C})\cup\bigcup_{i=1,2,3}\msf(\overline{C}_{i})\cup J^{C})
\]
where, for each large cluster $C$, $P^{C}=\emptyset$ and $J^{C}=\emptyset$
initially.
\item Initialize $\cA_{few}(H,B,p)$ on $H$ where $B=O(\pi d)$ and obtain
$\msf(H)$.\end{enumerate}
\end{algorithm}

\begin{table}
	\begin{tabular}{|>{\raggedright}p{0.2\textwidth}|>{\raggedright}p{0.75\textwidth}|}
		\hline 
		\textbf{Notation} & \textbf{Description}\tabularnewline
		\hline 
		\multicolumn{2}{c}{}\tabularnewline
		\multicolumn{2}{c}{\textbf{Algorithms}}\tabularnewline
		\hline 
		$\msfdecomp$ & $\msf$ decomposition algorithm from \ref{thm:MSF decomposition}\tabularnewline
		\hline 
		$\pruning$ & Dynamic expander pruning algorithm from \ref{thm:pruning detect failure}\tabularnewline
		\hline 
		$\cA_{few}$ & Dynamic $\msf$ for graphs with few non-tree edges obtained by induction
		and using \ref{thm:reduc restricted dec}\tabularnewline
		\hline 
		$\cA_{2}$ & from \ref{lem:non-tree cover}\tabularnewline
		\hline 
		$\cA_{3}$ & from \ref{lem:MSF in multigraph}\tabularnewline
		\hline 
		\multicolumn{2}{c}{}\tabularnewline
		\multicolumn{2}{c}{\textbf{Graphs, Edges and Parameters}}\tabularnewline
		\hline 
		$G=(V,E,w)$ & The input graph\tabularnewline
		\hline 
		$\gamma$ & The factor $\gamma=n^{O(\sqrt{\log\log n/\log n})}$ from \ref{thm:MSF decomposition}\tabularnewline
		\hline 
		$G'=(V,E,w')$ & The re-weighted graph where $(G',\cH)=\msfdecomp(G,\alpha,p,d,s_{low},s_{high})$
		where $\alpha=1/\gamma^{3}$, $d=\gamma$, $s_{low}=\gamma$, and
		$s_{high}=n/\gamma$\tabularnewline
		\hline 
		$\cH$ & The hierarchical decomposition of $G'$\tabularnewline
		\hline 
		$E(C)$ & The set of edges in a cluster $C\in\cH$. Note that, possibly, $(u,v)\notin E(C)$
		but $u,v\in V(C)$\tabularnewline
		\hline 
		$E^{C}$ & The set of \emph{$C$-own} edges. $E^{C}=E(C)-\Disjunion_{C':\textnormal{child of }C}E(C')$\tabularnewline
		\hline 
		$E(C)-E^{C}$ & The set of \emph{$C$-child} edges\tabularnewline
		\hline 
		$E^{\neq}$  & $E^{\neq}=\{e\in E\mid w(e)\neq w'(e)\}$\tabularnewline
		\hline 
		$E^{\neq}(w),$ $E^{\neq}(w')$ & The set of edges in $E^{\neq}$ with weight assigned by $w$ and $w'$,
		respectively\tabularnewline
		\hline 
		$G_{\smalltext}$ & $\Disjunion_{C:\smalltext}C$\tabularnewline
		\hline 
		$M_{\smalltext}$ & $\msf(G_{\smalltext})=\Disjunion_{C:\smalltext}\msf(C)$\tabularnewline
		\hline 
		$H$ & The sketch graph. See \ref{eq:sketch graph}.\tabularnewline
		\hline 
		$\pi$  & The update time of $\pruning$. $\pi=n^{O(\log\log\log n/\log\log n})\log\frac{1}{p}$.\tabularnewline
		\hline 
		\multicolumn{2}{c}{}\tabularnewline
		\multicolumn{2}{c}{\textbf{Inside a large cluster $C$}}\tabularnewline
		\hline 
		$M_{\smalltext}(C)$ & $\Disjunion_{C':\textnormal{small child of }C}\msf(C)$\tabularnewline
		\hline 
		$P_{0}^{C}$ & The \emph{pruning set }$P_{0}^{C}\subseteq V(C)$ of $C$ maintained
		by an instance of $\pruning$ that was initialized in $C$\tabularnewline
		\hline 
		$P^{C}$ & The \emph{total pruning set }of $C$. $P^{C}=\Disjunion_{C':\textnormal{large child of }C}P_{0}^{C'}$.\tabularnewline
		\hline 
		$\incident^{C}(U)$ & The set of $C$-own edges incident to $U$. $\incident^{C}(U)=\{(u,v)\in E^{C}\mid u\in U$
		or $v\in U\}$.\tabularnewline
		\hline 
		$\ensuremath{\overline{C}(U)}$ & The compressed cluster of $C$ with respect to $U$. See \ref{def:compressed cluster}.\tabularnewline
		\hline 
		$\ensuremath{\overline{C}}$ & The compressed cluster of $C$. $\ensuremath{\overline{C}}=\overline{C}(P^{C})=(V(\overline{C}),M_{\smalltext}(C)\disjunion E^{\overline{C}})$. \tabularnewline
		\hline 
		$E^{\ensuremath{\overline{C}}}$ & The set of $\overline{C}$-own edges. $E^{\ensuremath{\overline{C}}}=E^{C}-\incident^{C}(P^{C})$
		is the set of $C$-own edges \emph{not }incident to $P^{C}$.\tabularnewline
		\hline 
		$E_{i}^{\overline{C}}$ & The set of edges in $e\in E^{\overline{C}}$ where $(i-1)$ endpoints
		of $e$ are incident to super nodes in $\overline{C}$. Note that
		$E^{\ensuremath{\overline{C}}}=E_{1}^{\ensuremath{\overline{C}}}\disjunion E_{2}^{\ensuremath{\overline{C}}}\disjunion E_{3}^{\ensuremath{\overline{C}}}$.\tabularnewline
		\hline 
		$\ensuremath{\overline{C}}_{i}$ & $\ensuremath{\overline{C}}_{i}=(V(\overline{C}),M_{\smalltext}(C)\disjunion E_{i}^{\overline{C}})$\tabularnewline
		\hline 
		$J^{C}$ & The set of junk edges in $C$\tabularnewline
		\hline 
	\end{tabular}
	
	\caption{Definitions in \ref{sec:Dynamic MSF} \label{table:def}}
\end{table}

\subsubsection{Update}
\label{sec:MSF_update}

Now, we describe how to update the sketch graph $H$ given an edge
deletion of $G$, so that at every step we have $\msf(G) = \msf(H)$ and $H$ is extremely sparse.
We will handle at most $T(m)=m/(3\pi d\gamma)\le n/\pi d\gamma$
edge deletions. From now, we just write $T=T(m)$.

We describe how $H$ changes by showing, in the following order, how
we update 1) $E^{\neq}(w)$, 2) $M_{\smalltext}$ and, for each large
cluster $C$, 3) $\incident^{C}(P^{C})$, 4) $\msf(\overline{C}_{i})$
for $i=1,2,3$, and lastly 5) $J^{C}$.

Let $e$ be a given edge of $G$ to be deleted. We set $G\gets G-e$
and $G'\gets G'-e$. 
By \ref{fact:C-edge partition}, there is the unique cluster $C$ where $e$ is
a $C$-own edge (i.e. $e\in E^{C}$). We set $E^{C}\gets E^{C}-e$.
In particular, all ancestor clusters $C'$ of $C$ are changed: $E(C')\gets E(C')-e$ accordingly.
For each large cluster $C$ where $e\in E(C)$, the total pruning set 
$P^{C}=\Disjunion_{C':\textnormal{large child of }C}P_{0}^{C'}$
is updated by the instances of $\pruning$ that was initialized
in each large child $C'$ of $C$. Recall that $P^{C}$ only grows. If $e\in E^{\neq}(w)$,
then we set $E^{\neq}(w)\gets E^{\neq}(w)-e$. This determines the
changes of $E^{\neq}(w)$, $M_{\smalltext}$, and $\incident^{C}(P^{C})$ 
for each large cluster $C$. 

For any large $C$ and $i=1,2,3$, recall that $\overline{C}_{i}=\overline{C}_{i}(P^{C})=(V(\overline{C}),M_{\smalltext}(C)\disjunion E_{i}^{\overline{C}(P^{C})})$
is determined by $P^{C}$ and $M_{\smalltext}(C)$. The description
above already determines how $\overline{C}_{i}$ changes.
Hence, $\msf(\overline{C}_{i})$ is determined as well.
Finally, for $J^C$, whenever some edge $f$ is removed from $\msf(\overline{C}_{i})$,
for some $\overline{C}$ and $i\in\{1,2,3\}$ but $f$ is actually
not deleted from $G$ yet, then we include $f$ into $J^{C}$ as a
junk edge. 
\begin{rem}
We call these edges junk edges because of the following reason. Even
if $H$ did not include junk edges, then we can show that the algorithm
is still correct, i.e. $\msf(H)=\msf(G)$. However, junk edges are
needed for the performance reason: Given an edge deletion in $G$,
there can be $O(\pi d)$ many edges removed from $\bigcup_{C:\largetext,i=2,3}\msf(\overline{C}_{i})$.
But removing that many edges in $H$ will take too much time when
we recursively maintain $\msf(H)$ in $H$. So we just mark these
removed edges as junk edges, but do not actually remove them from
$H$.
\end{rem}
By the way we maintain junk edges, we have:
\begin{prop}
Given an edge $e$ to be deleted from $G$, only $e$ can be removed
from 
$$\bigcup_{C:\largetext}(\incident^{C}(P^{C})\cup\bigcup_{i=1,2,3}\msf(\overline{C}_{i})\cup J^{C}).$$\label{prop:few deletion in H}\end{prop}

\paragraph{Reporting Failure.}

During the sequence of updates, our algorithm might report ``failure''.
Once there is a failure, we terminate the whole algorithm and then
restart from the preprocessing. Here, we list the events such that
if they happen, the algorithm will report failure. First, we report
failure if any instance of $\cA$ or $\cA_{few}$ takes time more
than the time bound which is guaranteed to hold with high probability.
More formally, this is when an instance $\cA(m_{0},p_{0})$, for some
$m_{0},p_{0}$ takes time more than $t_{u}(m_{0},p_{0})$ for some
update, or when an instance $\cA_{few}(m_{1},k_{1},B_{1},p_{1})$
takes time more than $t_{del}^{few}(m_{1},k_{1},B_{1},p_{1})$ for
some edge deletion or more than $t_{ins}^{few}(m_{1},k_{1},B_{1},p_{1})$
for some batched insertion of size $B$. Second, we also report failure
whenever some instance of $\pruning$ from \ref{thm:pruning detect failure}
reports failure. It will be shown later in \ref{lem:fail low prob}
that failure happens with very low probability.

%% file: MSF_correct.tex
\subsection{Correctness}
\label{sec:MSF_correct}

In this section, we suppose that the algorithm does not fails. 
(Actually, we only need that no instance of $\pruning$ fails.) 
Then the sketch graph $H$ is maintained with the two desired properties.
First, we show in \ref{sec:sketch preserve} that $H$ ``preserves'' the $\msf$ i.e. $\msf(H) = \msf(G)$.
Second, we show in \ref{sec:sketch sparse} that $H$ is extremely sparse i.e. $|E(H)-\msf(H)| = O(n/\gamma)$.

Before proving the main goals, we prove a small technical lemma which 
ensures that the sequence of updates in each large cluster
is not too long for the dynamic expander pruning algorithm from \ref{thm:pruning detect failure}.
\begin{lem}
For each large cluster $C$, there is at most $O(\alpha_{0}^{2}|E(C)|)$
edge deletions in $C$.\end{lem}
\begin{proof}
As $|E(C)|\ge s_{high}=n/\gamma$ and $\alpha_{0}=\Omega(1/\gamma^{4})$,
so $\alpha_{0}^{2}|E(C)|=\Omega(n/\gamma^{9})$. But the total length
of update sequence is $T\le n/(\pi d\gamma)\le n/\gamma^{9}=O(\alpha_{0}^{2}|E(C)|)$
for large enough $n$.
\end{proof}

From now, in this section we assume that no instance of $\pruning$ fails.

\subsubsection{Sketch Graph Preserves $\msf$}\label{sec:sketch preserve}

Now, the goal is to prove the following:
\begin{lem}
$\msf(H)=\msf(G)$.\label{lem:can work in H}
\end{lem}
As the algorithm maintains $\msf(H)$, we can conclude from this lemma that $\msf(G)$ is correctly maintained.

Let $H'=M_{\smalltext}\cup\bigcup_{C:\largetext}(\incident^{C}(P^{C})\cup\bigcup_{i=1,2,3}\msf(\overline{C}_{i})\cup J^{C})$,
i.e. $H=H'\disjunion E^{\neq}(w)$. We first show that it suffices
to show that $\msf(G')\subseteq H'$.
\begin{lem}
	If $\msf(G')\subseteq H'$, then $\msf(G)=\msf(H)$.\label{lem:can work in H'}\end{lem}
\begin{proof}
	Suppose that $\msf(G')\subseteq H'$. Observe that $H'$ is a subgraph
	of $G'$, so $\msf(G')=\msf(H')$. Let $G''=G'\disjunion E^{\neq}(w)$
	be a multi-graph obtained from $G'$ by inserting $E^{\neq}(w)$ into
	$G'$. Note that $G''=G\disjunion E^{\neq}(w')$. Since $G''$ can
	obtained from $G$ by inserting a parallel edge heavier than edges
	in $G$, we have $\msf(G)=\msf(G'')$ So 
	\begin{align*}
	\msf(G) & =\msf(G'')\\
	& =\msf(G'\disjunion E^{\neq}(w))\\
	& =\msf(H'\disjunion E^{\neq}(w)) & \mbox{as }\msf(G')=\msf(H')\\
	& =\msf(H) & \mbox{because }H=H\disjunion E^{\neq}(w).
	\end{align*}
\end{proof}

The following lemma implies that $\msf(G')\subseteq H'$ because the
root cluster of $\cH$ corresponds to $G'$.
\begin{lem}
For any cluster $C\in\cH$, $\msf(C)\subseteq H'$.\label{lem:msf in H'}\end{lem}
\begin{proof}
We prove by induction on the hierarchy $\cH$ in a bottom-up manner.
For the base case, for each leaf cluster $C$, $\msf(C)\subseteq M_{\smalltext}\subseteq H'$
by definition. Next, we will prove that, for any large cluster $C$,
$\msf(C)\subseteq H'$, given that $\msf(C')\subseteq H'$ for all
child clusters $C'$ of $C$. By \ref{thm:MSF decomposition}, we
have that $\msf(C)=\Disjunion_{C':\textnormal{child of }C}\msf(C')\disjunion(\msf(C)\cap E^{C})$.
So it suffices to show that $\msf(C)\cap E^{C}\subseteq H'$. 

Recall that the total pruning set of $C$ is $P^{C}=\Disjunion_{C':\textnormal{large child of }C}P_{0}^{C'}$.
By \ref{thm:pruning detect failure}, for each large child $C'$
of $C$, there exists a set $W_{0}^{C'}\subseteq P_{0}^{C'}$ where
$C'[V(C')-W_{0}^{C'}]$ is connected, because we assume that no instance of $\pruning$ fails.
Let $W^{C}=\Disjunion_{C':\textnormal{large child of }C}W_{0}^{C'}$.
We need the following two claims:
\begin{claim}
$\msf(C)\cap E^{C}\subseteq\bigcup_{i=1,2,3}\msf(\overline{C}_{i}(W^{C}))\cup\incident^{C}(W^{C})$.\label{claim:pruning work}\end{claim}
\begin{proof}
Since $\msf(C)=\Disjunion_{C':\textnormal{child of }C}\msf(C')\disjunion(\msf(C)\cap E^{C})$
after any edge deletions by \ref{thm:MSF decomposition}, one can
identify $\msf(C)\cap E^{C}$ by running Kruskal's algorithm on $C$
where the initial forest consists of edges in $\Disjunion_{C':\textnormal{child of }C}\msf(C')$. 

In other words, let $\widehat{C}$ denote the graph obtained from
$C$ by replacing  each the small child $C'$ by $\msf(C')$. Let
$D_{1}$ be the graph obtained from $\widehat{C}$ by contracting
each connected component in $\Disjunion_{C':\textnormal{child of }C}\msf(C')$
into a single node. By the property of Kruskal's algorithm, we have
$\msf(C)\cap E^{C}=\msf(D_{1})$.

Let $D_{2}$ be the graph obtained from $\widehat{C}$ by contracting,
for each large child $C'$ of $\widehat{C}$, the set $V(C')-W_{0}^{C'}$
into a single node. Using the fact that $C'[V(C')-W_{0}^{C'}]$ is
connected, we know that $V(C')-W_{0}^{C'}$ is a subset of a connected
component in $\msf(C')$. That is, $D_{1}$ can be obtained from $D_{2}$
by further contracting nodes. By \ref{fact:contract}, $\msf(D_{1})\subseteq\msf(D_{2})$.
Since we know $\msf(D_{1})\subseteq E^{C}$, we have $\msf(D_{1})\subseteq\msf(D_{2})\cap E^{C}$. 

Observe that $\overline{C}(W^{C})$ is exactly the graph that can
be obtained from $D_{2}$ by removing the nodes in $W^{C}=\Disjunion_{C':\textnormal{large child of }C}W_{0}^{C'}$.
Let $E'$ be the edges in $D_{2}$ with some endpoint incident to
$W^{C}$. So $\msf(D_{2}-E')=\msf(\overline{C}(W^{C}))$. Having all
these, we can conclude
\begin{align*}
\msf(C)\cap E^{C} & \subseteq\msf(D_{2})\cap E^{C}\\
 & =\msf((D_{2}-E')\disjunion E')\cap E^{C}\\
 & \subseteq(\msf(D_{2}-E')\cup E')\cap E^{C} & \mbox{by \ref{fact:sparsify}}\\
 & =(\msf(\overline{C}(W^{C}))\cup E')\cap E^{C}\\
 & \subseteq\msf(\overline{C}(W^{C}))\cup\incident^{C}(W^{C}) & \mbox{as }E'\cap E^{C}=\incident^{C}(W^{C})\\
 & =\msf(\bigcup_{i=1,2,3}E(\overline{C}_{i}(W^{C})))\cup\incident^{C}(W^{C}) & \mbox{ by }\bigcup_{i=1,2,3}E(\overline{C}_{i}(W^{C}))=E(\overline{C}(W^{C}))\\
 & \subseteq\bigcup_{i=1,2,3}\msf(\overline{C}_{i}(W^{C}))\cup\incident^{C}(W^{C}) & \mbox{by \ref{fact:sparsify}.}
\end{align*}
\end{proof}
\begin{claim}
For $i=1,2,3$, $\msf(\overline{C}_{i}(W^{C}))\cup\incident^{C}(W^{C})\subseteq\msf(\overline{C}_{i}(P^{C}))\cup\incident^{C}(P^{C})$.
\label{claim:can prune too much}\end{claim}
\begin{proof}
Let $E'=E(\overline{C}_{i}(W^{C}))-E(\overline{C}_{i}(P^{C}))$. Note
that $E'\subseteq\incident^{C}(P^{C})-\incident^{C}(W^{C})$. Indeed,
for any edge $e\in E(\overline{C}_{i}(W^{C}))-E(\overline{C}_{i}(P^{C}))$,
$e$ must be a $C$-own edge in $\overline{C}_{i}$ that is incident
to $P^{C}$ because $e$ is removed if $P^{C}$ is pruned. Also,
$e$ is not incident to $W^{C}$ because $e$ is not removed if $W^{C}$
is pruned. So we have

\begin{align*}
\msf(\overline{C}_{i}(W^{C})) & =\msf(\overline{C}_{i}(P^{C})\cup E')\\
 & \subseteq\msf(\overline{C}_{i}(P^{C}))\cup E' & \mbox{by \ref{fact:sparsify}}\\
 & \subseteq\msf(\overline{C}_{i}(P^{C}))\cup(\incident^{C}(P^{C})-\incident^{C}(W^{C})).
\end{align*}
Applying union of $\incident^{C}(W^{C})$ on both sides completes
the claim.
\end{proof}
By the above two claims, we have that 
\begin{align*}
\msf(C)\cap E^{C} & \subseteq\bigcup_{i=1,2,3}\msf(\overline{C}_{i}(W^{C}))\cup\incident^{C}(W^{C}) & \mbox{by \ref{claim:pruning work}}\\
 & \subseteq\bigcup_{i=1,2,3}\msf(\overline{C}_{i}(P^{C}))\cup\incident^{C}(P^{C}) & \mbox{by \ref{claim:can prune too much}}\\
 & \subseteq H' & \mbox{as }\overline{C}_{i}(P^{C})=\overline{C}_{i}\mbox{ by definition,}
\end{align*}
which completes the proof.
\end{proof}
By \ref{lem:can work in H'} and \ref{lem:msf in H'}, this implies
\ref{lem:can work in H}. That is, $\msf(H) = \msf(G)$.

\subsubsection{Sketch Graph is Extremely Sparse}\label{sec:sketch sparse}

The goal here is to prove the following:
\begin{lem}
	$|E(H)-\msf(H)|=O(n/\gamma)$.
	%with probability at least $1-p$. Moreover,
	%the bound holds with certainty before the first update.
	\label{lem:few non-tree edge in H}
\end{lem}
To prove this, we need to define some definitions. Fix $i\in\{1,2,3\}$.
Let $G'_{i}=(V,M_{\smalltext}\disjunion\Disjunion_{C:\largetext}E_{i}^{\overline{C}})$.
Recall that $E_{i}^{\overline{C}}$ is the $C$-own edges in $\overline{C}$
incident whose $(i-1)$ endpoints are incident to super nodes in $\overline{C}$.
Note that $G'_{i}$ is a subgraph of $G'$. We can also define a corresponding
hierarchical decomposition $\cH_{i}$ of $G'_{i}$. For each small
cluster $C$ in $\cH$, let $C_{i}=(V(C),\msf(C))=(V(C),M_{\smalltext}[V(C)])$
be a small cluster in $\cH_{i}$. For each large cluster $C$ in $\cH$,
let $C_{i}$ be a large cluster in $\cH_{i}$ where $V(C_{i})=V(C)$
and the set of $C_{i}$-own edges is $E^{C_{i}}=E_{i}^{\overline{C}}$.
From this definition, for every cluster $C\in\cH$, there is a corresponding
cluster $C_{i}\in\cH_{i}$. For each large cluster $C\in\cH$, $\overline{C}_{i}$
is a subgraph of the compressed cluster $\overline{C}$. We have the
following relation between $\overline{C}_{i}$ and $C_{i}$: 
\begin{lem}
	For any large cluster $C\in\cH$ and $i\in\{1,2,3\}$, $\msf(\overline{C}_{i})\subseteq\msf(C_{i})$.\label{lem:C_i bar in C_i}\end{lem}
\begin{proof}
	Observe that $\overline{C}_{i}$ can be obtained from $C_{i}$ by
	contracting each large child $C'_{i}$ of $C_{i}$ into a single node.
	By \ref{fact:contract}, $\msf(\overline{C}_{i})\subseteq\msf(C_{i})$. \end{proof}
\begin{lem}
	For any cluster $C_{i}\in\cH_{i}$, $\msf(C_{i})\subseteq\msf(G'_{i})$
	.\label{lem:C_i in G'_i}\end{lem}
\begin{proof}
	By \ref{thm:MSF decomposition} and the fact that $G'_{i}$ is a subgraph
	of $G'$. We have the following: for any cluster $C_{i}\in\cH_{i}$,
	$\msf(C_{i})=\Disjunion_{C_{i}':\textnormal{child of }C_{i}}\msf(C_{i}')\disjunion(\msf(C_{i})\cap E^{C_{i}})$.
	In particular, $\msf(C_{i}')\subseteq\msf(C_{i})$ for any child cluster
	$C'_{i}$ of $C_{i}$. Therefore, $\msf(C_{i})\subseteq\msf(G'_{i})$
	because $G'_{i}$ is the root cluster of $\cH_{i}$. \end{proof}
\begin{lem}
	$|\bigcup_{C:\largetext}\msf(\overline{C}_{i})\setminus M_{\smalltext}|\le\frac{n}{s_{low}/3}+T$.\label{lem:few non tree M_123}\end{lem}
\begin{proof}
	By \ref{lem:C_i bar in C_i} and \ref{lem:C_i in G'_i}, we have $\bigcup_{C:\largetext}\msf(\overline{C}_{i})\subseteq\bigcup_{C:\largetext}\msf(C_{i})\subseteq\msf(G'_{i})$.
	So it suffices to bound $|\msf(G'_{i})\setminus M_{\smalltext}|$.
	Next, observe that $M_{\smalltext}=\Disjunion_{C_{i}:\smalltext}\msf(C_{i})\subseteq\msf(G'_{i})$
	by the definition of small $C_{i}$ and \ref{lem:C_i in G'_i}. Therefore,
	$|\msf(G'_{i})\setminus M_{\smalltext}|$ is at most the number of
	connected components in $M_{\smalltext}$. 
	
	Before the first edge deletion, we have that all small clusters are
	connected and each small cluster has at least $s_{low}/3$ nodes by
	\ref{thm:MSF decomposition}. So there are at most $\frac{n}{s_{low}/3}$
	connected components in $M_{\smalltext}$ at that time. After $T$
	edge deletions, the number of connected components can be increased
	by at most $T$. So $|\msf(G'_{i})\setminus M_{\smalltext}|\le\frac{n}{s_{low}/3}+T$.
\end{proof}

\paragraph{Proof of \ref{lem:few non-tree edge in H}.}

Now, we can bound the number of the non-tree edges of $H$.
\begin{proof}
	Recall that $E(H)=E^{\neq}(w)\cup M_{\smalltext}\cup\bigcup_{C:\largetext}(\incident^{C}(P^{C})\cup\bigcup_{i=1,2,3}\msf(\overline{C}_{i})\cup J^{C})$.
	First, by \ref{thm:MSF decomposition}, $|E^{\neq}(w)|\le\alpha d\gamma n=O(n/\gamma)$.
	Next, by \ref{lem:few non tree M_123}, $|\bigcup_{i=1,2,3;C:\largetext}\msf(\overline{C}_{i})\setminus M_{\smalltext}|=O(\frac{n}{s_{low}/3}+T)=O(n/\gamma)$. 
	
	%With probability $1-p$, all non-root clusters $C$ is such that $\phi(C)\ge\alpha_{0}$
	%by \ref{thm:MSF decomposition}. 
	%Assume that it is the case. 
	Next, for each edge update on each large cluster, $\pruning$ spends time
	by at most $\pi$ by the definition of $\pi$. Therefore, $\sum_{C:\largetext}|P^{C}|\le T\times\pi d$
	because, for each deletion of an edge $e$, $e$ is contained in at
	most $d$ clusters as the depth of $\cH$ is at most $d$ by \ref{thm:MSF decomposition},
	and, for each large cluster $C$ whose edge is deleted, $|P^{C}|$
	can grow by at most $\pi$. Hence, $|\bigcup_{C:\largetext}\incident^{C}(P^{C})|=O(\sum_{C:\largetext}|P^{C}|)=O(T\pi d)=O(n/\gamma).$
	Finally, we bound $|\bigcup_{C:\largetext}J^{C}|$. By definition,
	$J^{C}$ contains edges that are removed from $\msf(\overline{C}_{i})$,
	over all $\overline{C}_{i}$ and $i\in\{1,2,3\}$, but are not deleted
	from $G$ yet. So $J^{C}\subseteq\incident^{C}(P^{C})$. Hence $|\bigcup_{C:\largetext}J^{C}|\le|\bigcup_{C:\largetext}\incident^{C}(P^{C})|=O(n/\gamma)$
	as well. 
%	From these, with probability $1-p$, we have $|E(H)\setminus M_{\smalltext}|=O(n/\gamma)$,
%	and hence $|E(H)\setminus\msf(H)|=O(n/\gamma)$ at any time. Moreover,
%	as $\bigcup_{C:\largetext}\incident^{C}(P^{C})\cup J^{C}=\emptyset$
%	before the updates, we have $|E(H)\setminus\msf(H)|=O(n/\gamma)$
%	with certainty before the first update.
\end{proof}

%% file: MSF_time.tex
\subsection{Running Time }\label{sec:MSF_time}

In this section, we assume again that the algorithm does not fails. 
%(Actually, we only need that no instance of $\cA$ or $\cA_{few}$ fails i.e. $\cA$
%or $\cA_{few}$ takes time more than the time bound which is guaranteed to hold.) 
Under this assumption, we analyze the running time of the algorithm.
We bound the preprocessing time in \ref{sec:time_pre}, the time
needed for maintaining the sketch graph $H$ itself in \ref{sec:time_sketch}, 
and finally, in \ref{sec:time_msf(H)}, the time needed for maintaining $\msf(H)$ in $H$ 
that changes more than one edge per time step.

Recall the following notations. An instance $\cA(m_{0},p_{0})$, for any
$m_{0}$ and $p_{0}$, has preprocessing time $t_{pre}(m_{0},p_{0})$ and deletion time $t_{u}(m_{0},p_{0})$.
Also, an instance $\cA_{few}(m_{1},k_{1},B_{1},p_{1})$, for any $m_{1}$, $k_1$, $B_1$, and $p_{1}$, 
has preprocessing time $t{}_{pre}^{few}(m_{1},k_{1},B_{1},p_{1})$, batch insertion time 
$t{}_{ins}^{few}(m_{1},k_{1},B_{1},p_{1})$, and deletion time $t{}_{del}^{few}(m_{1},k_{1},B_{1},p_{1})$.
%**ADD THAT preprocessing time is with certainty***

\subsubsection{Preprocessing}\label{sec:time_pre}

Next, we bound the preprocessing time, which in turn is needed for
bounding the update time later in \ref{lem:final update time}.
\begin{lem}
Given an $n$-node $m$-edge graph $G$ with max degree 3 and a parameter
$p$, \ref{alg:MSF prep} takes $O(m^{1+O(\sqrt{\log\log n/\log n})}\log\frac{1}{p})$
time.\label{lem:final prep time}\end{lem}
\begin{proof}
Consider each step in \ref{alg:MSF prep}. In Step 1, we just run
the $\msf$ decomposition which takes $\tilde{O}(nd\gamma\log\frac{1}{p})=\tilde{O}(n\gamma^{2}\log\frac{1}{p})$
by \ref{thm:MSF decomposition}. In Step 2, we initialize $\cA(C,p)$
for each small cluster $C$. This takes time

\[
\sum_{C:\smalltext}t{}_{pre}(|E(C)|,p)\le\frac{n}{s_{high}}\cdot t{}_{pre}(O(s_{high}),p)=\gamma\cdot t{}_{pre}(O(n/\gamma),p).
\]
where the inequality is because $t_{pre}(m,p)$ is at least linear
in $m$, and $|E(C)|\le s_{high}$ for each small cluster $C$.
Step 3 takes total time $\tilde{O}(n\log(1/p)\cdot d)=\tilde{O}(m\gamma \log(1/p))$ because $\pruning$ from \ref{thm:pruning detect failure} 
initializes on a large cluster $C$ in time $\tilde{O}(|V(C)| \log(1/p))$ and the depth of the decomposition is $d$.
In Step 4.a, the total time for constructing all compressed clusters
$\overline{C}$ and $\overline{C}_{1},\overline{C}_{2},\overline{C}_{3}$
is just $O(n)$.

In Step 4.b, for each large cluster $C$, we initialize $\cA_{few}(\overline{C}_{1},1,p)$.
Note that the set of non-tree edges in $\overline{C}_{1}$ is contained
in $E(\overline{C}_{1})-M_{\smalltext}(C)=E_{1}^{\overline{C}}$ by
\ref{fact:edge of compressed}. So the initialization takes $t_{pre}^{few}(|E(\overline{C}_{1})|,|E_{1}^{\overline{C}}|,1,p)$
time. In total this takes time 
\begin{align*}
\sum_{C:\largetext}t_{pre}^{few}(|E(\overline{C}_{1})|,|E_{1}^{\overline{C}}|,1,p) & =\sum_{C:\largetext}(t{}_{pre}(O(|E_{1}^{\overline{C}}|),O(p/\log n))+\tilde{O}(|E(\overline{C}_{1})|)) & \mbox{by \ref{thm:reduc restricted dec}}\\
 & =\tilde{O}(n)+\sum_{C:\largetext}t{}_{pre}(O(|E_{1}^{\overline{C}}|),O(p/\log n))\\
 & =\tilde{O}(n)+\sum_{1\le i\le d}\sum_{C:\text{large, level-}i}t{}_{pre}(O(|E_{1}^{\overline{C}}|),O(p/\log n))\\
 & \le\tilde{O}(n)+\sum_{1\le i\le d}t{}_{pre}(O(n/\gamma),O(p/\log n))\\
 & =\tilde{O}(n)+\gamma\cdot t_{pre}(O(n/\gamma),O(p/\log n)) & \mbox{by }d=\gamma,
\end{align*}
where the inequality follows because $\sum_{C:\text{large, level-}i}|E_{1}^{\overline{C}}|\le\sum_{C:\text{large, level-}i}|E^{C}|\le n/(d-2)+\alpha\gamma n=O(n/\gamma)$
by \ref{thm:MSF decomposition} and $t_{pre}(m,p)$ is at least linear
in $m$. 

In Step 4.c, by \ref{lem:MSF in multigraph} and \ref{lem:non-tree cover},
the total time for initializing $\cA_{2}(\overline{C}_{2})$ and $\cA_{3}(\overline{C}_{3})$
over all large clusters $C$, is $\tilde{O}(n)$ because compressed
cluster are edge-disjoint and $\cA_{2}$ and $\cA_{3}$ have near-linear
preprocessing time.

In Step 5, we initialize $\cA_{few}(H,B,p)$. This takes time 
\begin{align*}
 & t_{pre}^{few}(|E(H)|,|E(H)-\msf(H)|,B,p)\\
 & =t_{pre}(O(|E(H)-\msf(H)|),O(p/\log n))+\tilde{O}(|E(H)|\log n)\\
 & =t_{pre}(O(n/\gamma),p)+\tilde{O}(n) & \mbox{by \ref{lem:few non-tree edge in H}.}
\end{align*}
Note that $m=\Theta(n)$. Now, we conclude that the total preprocessing
time is 

\[
t_{pre}(m,p)=\tilde{O}(m\gamma^{2}\log\frac{1}{p})+O(\gamma)\times t_{pre}(O(m/\gamma),O(p/\log m)).
\]
To solve this recurrence, we use the following fact:
\begin{fact}
Let $f(n)$ and $g(n)$ be a function where $g(n)=\Omega(n)$. If
$f(n)\le c\cdot a\cdot f(n/a)+g(n)$, then $f(n)=\tilde{O}(g(n)\cdot c^{\log_{a}n})$.\label{prop:solve recursion}
\end{fact}
Recall that $\gamma=n^{O(\sqrt{\log\log n/\log n})}$. Let $d_{0}=\log_{O(\gamma)}m=O(\frac{\log m}{\sqrt{\log m\log\log m}})=O(\sqrt{\frac{\log m}{\log\log m}})$.
After solving the recurrence, we have 
\begin{align*}
t_{pre}(m,p) & =\tilde{O}(m\gamma^{2}\log\frac{\log^{d_{0}}m}{p}\times c_{0}^{d_{0}}) & \mbox{for some constant }c_{0}\\
 & =\tilde{O}(m\gamma^{2}\log\frac{1}{p}\times m^{O(\sqrt{1/\log m\log\log m})})\\
 & =O(m^{1+O(\sqrt{\log\log m/\log m})}\log\frac{1}{p})
\end{align*}

\end{proof}

\subsubsection{Maintaining the Sketch Graph}\label{sec:time_sketch}

We bound the time for maintaining the sketch graph $H$. For convenience,
we first show the following lemma.
\begin{lem}
For some large cluster $C$ and $i\in\{2,3\}$, $\cA_{i}(\overline{C}_{i})$
takes $\tilde{O}(\gamma)$ time to update $\msf(\overline{C}_{i})$
for each edge update in $E(\overline{C}_{i})$.\label{lem:2 and 3 fast}\end{lem}
\begin{proof}
For $i=2$, by \ref{prop:super node few} the set of super nodes $S_{\overline{C}}$
has size $|S_{\overline{C}}|=O(\gamma)$. So \ref{lem:non-tree cover},
$\cA_{2}(\overline{C}_{2})$ has $\tilde{O}(\gamma)$ update time.
For $i=3$, by \ref{rem:nodes in C_3}, the graph that $\cA_{3}(\overline{C}_{3})$
actually runs on is induced by the set of edges in $E_{3}^{\overline{C}}$,
and this graph has $O(\gamma)$ nodes. So each update takes $\tilde{O}(\gamma)$
time by \ref{lem:MSF in multigraph}. 
\end{proof}
The next lemma bounds the time for maintaining the sketch graph $H$.
\begin{comment}
CAN SHOW: just 1 edge deletion!
\end{comment}

\begin{lem}
Suppose that the algorithm does not fail. For each edge deletion in $G$, $H$ can
be updated in time $t_{u}(n/\gamma,p)+t_{del}^{few}(n,O(n/\gamma),1,p)+t{}_{ins}^{few}(n,O(n/\gamma),1,p)+\tilde{O}(\pi\gamma^{2})$.
Moreover, there are at most $2$ edge deletions and $B=O(\pi\gamma)$
edge insertions in $H$.\label{lem:maintain H}\end{lem}
\begin{proof}
Recall that $E(H)=E^{\neq}(w)\cup M_{\smalltext}\cup\bigcup_{C:\largetext}(\incident^{C}(P^{C})\cup\bigcup_{i=1,2,3}\msf(\overline{C}_{i})\cup J^{C})$.
Suppose that we are given an edge deletion in $G$.

Given the edge deletion in $G$, there is at most one small cluster
$C$ where $\msf(C)\subseteq M_{\smalltext}$ is changed. $\msf(C)$
can be updated by $\cA(C,p)$ in time $t_{u}(|E(C)|,p)=t_{u}(n/\gamma,p)$.
In $M_{\smalltext}$, there are at most 1 edge deletion and at most
1 edge insertion.

As we assume that no instance of $\pruning$ fails, 
$\bigcup_{C:\largetext}\incident^{C}(P^{C})$ can be updated
in time $O(\pi d)=O(\pi\gamma)$. This is
because 1) $\cH$ has depth at most $d$ by \ref{thm:MSF decomposition}
and so each edge is contained in at most $d$ clusters, and 2) for
each large cluster $C$ whose edge is deleted, $\pruning$ spends
time at most $\pi$. Therefore, the size of $\bigcup_{C:\largetext}\incident^{C}(P^{C})$
can grow by at most $O(\pi\gamma)$ as well.

Next, we bound the time for maintaining $\bigcup_{i=1,2,3;C:\largetext}\msf(\overline{C}_{i})$.
For $i=1$, we have that there is at most one compressed cluster $\overline{C}$
where the deleted edge $e\in E(\overline{C}_{1})$ because of edge-disjointness
by \ref{prop:compressed edge disj}. Observe that $E(\overline{C}_{1})$
is determined only by $E^{C}$ and $M_{\smalltext}(C)$, and not $P^{C}$.
If $e\in M_{\smalltext}(C)$, then this generates at most one edge
deletion and at most one insertion in $E(\overline{C}_{1})$. Else,
$e\in E^{C}$, then there is one deletion in $E(\overline{C}_{1})$.
So the total cost spent by $\cA_{few}(\overline{C}_{1},1,p)$ is at most
\[
t{}_{del}^{few}(|E(\overline{C}_{1})|,|E_{1}^{\overline{C}}|,1,p)+t{}_{ins}^{few}(|E(\overline{C}_{1})|,|E_{1}^{\overline{C}}|,1,p)=t{}_{del}^{few}(n,O(n/\gamma),1,p)+t{}_{ins}^{few}(n,O(n/\gamma),1,p).
\]

For $i=2,3$, we have that $E(\overline{C}_{i})$ depends also on
$P^{C}$. Hence, there are at most $O(\pi\gamma)$ edge updates in
$\bigcup_{C:\largetext}E(\overline{C}_{i})$.
For each edge update in $E(\overline{C}_{i})$ for some large cluster
$C$, $\cA_{i}(\overline{C}_{i})$ takes $\tilde{O}(\gamma)$ time
to update $\msf(\overline{C}_{i})$ by \ref{lem:2 and 3 fast}. Therefore
the total time for updating $\bigcup_{i=2,3;C:\largetext}\msf(\overline{C}_{i})$
is $O(\pi\gamma)\times\tilde{O}(\gamma)=\tilde{O}(\pi\gamma^{2})$.
The time for updating $\bigcup_{C:\largetext}J^{C}$ is subsumed by
other steps. Therefore, the total update time is at most $t_{u}(n/\gamma,p)+t^{few}_{del}(n,O(n/\gamma),1,p)+t^{few}_{ins}(n,O(n/\gamma),1,p)+\tilde{O}(\pi\gamma^{2})$. 

To bound the edge changes in $H$, by \ref{prop:few deletion in H},
there is no edge removed from $\bigcup_{C:\largetext}(\incident^{C}(P^{C})\cup\bigcup_{i=1,2,3}\msf(\overline{C}_{i})\cup J^{C})$
except the deleted edge itself. So there are at most $2$ edge deletions
(from edges in $E^{\neq}(w)$ or $M_{\smalltext}$) and $B=O(\pi\gamma)$
edge insertions in $H$.
\end{proof}

\subsubsection{Maintaining $\protect\msf$ of the Sketch Graph}\label{sec:time_msf(H)}

Finally, we bound the time to maintain $\msf(H)$ which is the same
as $\msf(G)$ by \ref{lem:can work in H}.
\begin{lem}
Suppose that the algorithm does not fail.
The algorithm for \ref{lem:dec MST final} has  update time
$$O(m^{O(\log\log\log m/\log\log m)}\log\frac{1}{p}).$$\label{lem:final update time}\end{lem}
\begin{proof}
By \ref{lem:maintain H}, we need to spend time at most 
\[
t_{u}(n/\gamma,p)+t{}_{del}^{few}(n,O(n/\gamma),1,p)+t{}_{ins}^{few}(n,O(n/\gamma),1,p)+\tilde{O}(\pi\gamma^{2})
\]
for maintaining the sketch graph $H$ itself. Moreover, there are
only $2$ edge deletions and $B=O(\pi\gamma)$ edge insertions in $H$.
Given these updates, we can update $\msf(H)$ using $\cA_{few}(H,B,p)$
in time 
\begin{eqnarray*}
 &  & 2t{}_{del}^{few}(|E(H)|,|E(H)-\msf(H)|,B,p)+t{}_{ins}^{few}(|E(H)|,|E(H)-\msf(H)|,B,p)\\
 & = & 2t{}_{del}^{few}(n,O(n/\gamma),B,p)+t{}_{ins}^{few}(n,O(n/\gamma),B,p)
\end{eqnarray*}
by \ref{lem:few non-tree edge in H}. 
Note that $m=\Theta(n)$. Write $k=\Theta(m/\gamma)$. We have that the total
update time to maintain $\msf(H)$ is

\begin{eqnarray*}
 &  & t_{u}(m,p)\\
 & \le & t_{u}(k,p)+3t{}_{del}^{few}(O(m),k,B,p)+2t{}_{ins}^{few}(O(m),k,B,p)+\tilde{O}(\pi\gamma^{2})\\
 & = & O(\frac{B\log k}{k}\cdot t{}_{pre}(O(k),O(p/\log m))+B\log^{2}m+\frac{k\log k}{T(k)}+\log k\cdot t{}_{u}(O(k),O(p/\log m)))+\tilde{O}(\pi\gamma^{2})\\
 & = & O(\log k\cdot t{}_{u}(O(k),O(p/\log m)))+O(m^{O(\log\log\log m/\log\log m)}\log\frac{1}{p}).
\end{eqnarray*}
The first equality is because of \ref{thm:reduc restricted dec}.
To show the last equality, note first that by \ref{lem:final prep time},
\begin{align*}
\frac{B\log k}{k}\cdot t{}_{pre}(O(k),p') & =O(\frac{B\log k}{k}\cdot k^{1+O(\sqrt{\log\log m/\log m})}\log\frac{1}{p})\\
 & =\tilde{O}(\pi\gamma k^{O(\sqrt{\log\log m/\log m})}\log\frac{1}{p}) & B=O(\pi\gamma)\\
 & =O(m^{O(\log\log\log m/\log\log m)}\log\frac{1}{p}).
\end{align*}
Next, note that $T(m)=m/(3\pi d\gamma)=m^{1-\Theta(\log\log\log m/\log\log m)}$.
So $T(k)=k^{1-\Theta(\log\log\log k/\log\log k)}$ and we have $\frac{k\log k}{T(k)}=\frac{k\log k}{k^{1-\Theta(\log\log\log k/\log\log k)}}=O(m^{O(\log\log\log m/\log\log m)})$.
Also, $O(B\log^{2}m+\pi\gamma^{2})=O(m^{O(\log\log\log m/\log\log m)})$.
Therefore, we have 
\[
t_{u}(m,p)=O(\log m\cdot t{}_{u}(O(m/\gamma),O(p/\log m)))+O(m^{O(\log\log\log m/\log\log m)}\log\frac{1}{p}).
\]
So solve this recurrence, note that $d_{0}=\log_{O(\gamma)}m=O(\frac{\log m}{\sqrt{\log m\log\log m}})=O(\sqrt{\frac{\log m}{\log\log m}})$
is the depth of the recursion. After solving the recurrence, we have
\begin{eqnarray*}
t_{u}(m,p) & = & O(\log^{d_{0}}m)\times O(m^{O(\log\log\log m/\log\log m)}\log\frac{\log^{d_{0}}m}{p})\\
 & = & O(m^{O(\log\log\log m/\log\log m)}\log\frac{1}{p}).
\end{eqnarray*}

\end{proof}

\subsubsection{Wrapping Up}\label{sec:wrap up}

As a last step, we show that the algorithm fails with low probability.
\begin{lem}
	\label{lem:fail low prob}For each update, the algorithm fails with
	probability $O(\gamma p)$.\end{lem}
\begin{proof}
	There are two types of events that cause the algorithm to fail: 1)
	``$\cA$ or $\cA_{few}$ fails'': some instance of $\cA$ or $\cA_{few}$
	takes time longer than the guaranteed time bound, 2) ``$\pruning$
	fails'': some instance of $\pruning$ reports failure. Now, we fix
	some time step and will bound the probability of the occurrence of
	each type of events. 
	
	First, let $E_{1}$ be the event that some instance of $\cA$ or $\cA_{few}$
	fails. We list the instances of $\cA$ and $\cA_{few}$ first. Consider
	\ref{alg:MSF prep}. For each small cluster $C$, there is the instance
	$\cA(C,p)$ (in Step 2). For each large cluster $C$, there is the
	instance $\cA_{few}(\overline{C}_{1},1,p)$ (in Step 4.a). Lastly,
	there is the instance $\cA_{few}(H,B,p)$ on the sketch graph $H$.
	Now, these instances can fail every time we feed the update operations
	to them. So we list how we feed the update operations to them. Look
	inside the proof of \ref{lem:maintain H}. Given an edge $e$ to be
	deleted, there is at most one small cluster $C$ where we feed the
	edge deletion of $e$ to $\cA(C,p)$. Also, there is at most one
	compressed cluster $\overline{C}$ that we need to feed the update
	to $\cA_{few}(\overline{C}_{1},1,p)$. There are at most one insertion
	and one deletion. For the sketch graph $H$, by \ref{lem:maintain H}
	there are at most two edge deletions and one batch of edge insertion
	fed to $\cA_{few}(H,B,p)$. In total, there are $O(1)$ many operations
	that we feed to the instances of $\cA$ or $\cA_{few}$. Each time,
	an instance can fail with probability at most $p$ by the definition
	of the parameter $p$ in $\cA(C,p)$, $\cA_{few}(\overline{C}_{1},1,p)$
	and $\cA_{few}(H,B,p)$. So $\Pr[E_{1}]=O(p)$. 
	
	Second, let $E_{2}$ be the event that some instance of $\pruning$
	reports failure. Let $E'_{2}$ be the event that all large clusters
	$C$ are such that $\phi(C)=\Omega(\alpha_{0})$. We have that
	\begin{eqnarray*}
		\Pr[E_{2}] & = & \Pr[E_{2}\mid\neg E'_{2}]\Pr[\neg E'_{2}]+\Pr[E_{2}\mid E'_{2}]\Pr[E'_{2}]\\
		& \le & \Pr[\neg E'_{2}]+\Pr[E_{2}\mid E'_{2}].
	\end{eqnarray*}
	By \ref{thm:MSF decomposition}, $\Pr[\neg E'_{2}]\le p$. Next, for
	each large cluster $C$, if $\phi(C)=\Omega(\alpha_{0})$, then the
	instance of $\pruning$ on $C$ fails at some step with probability
	at most $p$ by \ref{thm:pruning detect failure}. Moreover, there
	are only $O(d)=O(\gamma)$ many large clusters that are updated for each step.
	By union bound, $\Pr[E_{2}\mid E'_{2}]\le O(\gamma p)$. This implies
	that $\Pr[E_{2}]=O(\gamma p)$. This conclude that the algorithm
	fails with probability at most $\Pr[E_{1}]+\Pr[E_{2}]=O(\gamma p)$
	at each step.
\end{proof}

Finally, we conclude the proof of \ref{lem:dec MST final} 
and which implies \ref{thm:dyn MST final}, our main result.

\paragraph{Proof of \ref{lem:dec MST final}. }

By \ref{lem:can work in H}, we have that the sketch graph $H$ is
such that $\msf(H)=\msf(G)$. As the instance $\cA_{few}(H,B,p)$
maintains $\msf$ in $H$, we conclude that $\msf(G)$ is correctly
maintained. The algorithm has preprocessing time 
\[
t_{pre}(m,p)=O(m^{1+O(\sqrt{\log\log m/\log m})}\log\frac{1}{p})
\]
by \ref{lem:final prep time}. Given a sequence of edge deletions
of length 
\[
T(m)=m/(3\pi d\gamma)=\Theta(m^{1-O(\log\log\log m/\log\log m)}),
\]
the algorithm take time 
\[
t_{u}(m,p)=O(m^{O(\log\log\log m/\log\log m)}\log\frac{1}{p})
\]
for each update with probability $1-O(\gamma p)$, by \ref{lem:final update time,lem:fail low prob}. 
As noted before, we have $m=\Theta(n)$
throughout the sequence of updates. Finally, we can obtain the proof
of \ref{lem:dec MST final} by slightly adjusting the parameter $p$
so that update time bound holds with the probability $1-p$ instead
of $1-O(\gamma p)$.

%% file: openproblems.tex
\section{Open Problems}

\paragraph{Dynamic $\protect\msf$.}

First, it is truly intriguing whether there is a \emph{deterministic} algorithm
that is as fast as our algorithm. The current best update time of
deterministic algorithms is still $\tilde{O}(\sqrt{n})$ \cite{Frederickson85,EppsteinGIN97,Kejlberg-Rasmussen16}
(even for dynamic connectivity). Improving this bound to $O(n^{0.5-\Omega(1)})$
will already be a major result. Secondly, can one improve the $O(n^{o(1)})$
update time to $O(\polylog(n))$? There are now several barriers in
our approach and this improvement should require new ideas. Lastly, it is
also very interesting to simplify our algorithm.

\paragraph{Expander-related Techniques.}

The combination of the expansion decomposition and dynamic expander
pruning might be useful for other dynamic graph problems. 
%Can this technique be used to improve the $\tilde{O}(\sqrt{n})$ update time
%of dynamic min cut algorithm by Thorup \cite{Thorup07mincut} to $O(n^{0.5-\Omega(1)})$
% update time? 
Problems whose static algorithms are based on low-diameter
decomposition (e.g. low-stretch spanning tree) are possible candidates.
Indeed, it is conceivable that the expansion decomposition together
with dynamic expander pruning can be used to maintain low diameter decomposition
under edge updates, but additional work maybe required.

\paragraph{Worst-case Update Time Against Adaptive Adversaries.}

Among major goals for dynamic graph algorithm are (1) to reduce gaps between \emph{worst-case} and ~\emph{amortized} update
time, and (2) to reduce gaps between update time of algorithms that work against \emph{adaptive} adversaries
and those that require \emph{oblivious} adversaries. Upper bounds known for
the former case (for both goals) are usually much higher than those for
the latter. However, worst-case bounds are crucial in real-time applications,
and being against adversaries is often needed when algorithms are
used as subroutines of static algorithms. Note that of course deterministic
algorithms always work against adaptive adversaries.

%
%The following are the major goals for dynamic graph algorithms in general:
%to reduce gaps between 1) \emph{worst-case} vs.~\emph{amortized} update
%time, and 2) update time of algorithms against \emph{adaptive} adversaries
%vs.~algorithms against \emph{oblivious} adversaries. The bounds for
%the former in both goals are usually much higher than the bounds for
%the latter. However, worst-case bounds are crucial in real-time applications,
%and being against adversaries is often needed when algorithms are
%used as subroutines in static algorithms. Note that of course deterministic
%algorithms are against adaptive adversaries.

The result in this paper is a step towards both goals. The best amortized
bound for dynamic $\msf$ is $O(\polylog(n))$ \cite{HolmLT01,HolmRW15}.
For dynamic $\sf$ problem, the result by \cite{KapronKM13,GibbKKT15}
implies the current best algorithm against oblivious adversaries with
$O(\polylog(n))$ worst-case update time. Our dynamic $\msf$ algorithm
is against adaptive adversaries and has $O(n^{o(1)})$ worst-case
update time. This significantly reduces the gaps on both cases.

It is a challenging goal to do the same for other fundamental problems.
For example, dynamic 2-edge connectivity has $O(\polylog(n))$
amortized update time \cite{HolmLT01} but only $O(\sqrt{n})$ worst-case
bound \cite{Frederickson97,EppsteinGIN97}. Dynamic APSP has $\tilde{O}(n^{2})$
amortized bound \cite{DemetrescuI03} but only $\tilde{O}(n^{2+2/3})$
worst-case bound \cite{AbrahamCK17}. There are fast algorithms against
oblivious adversaries for dynamic maximal matching \cite{BaswanaGS11},
spanner \cite{BaswanaKS12}, and cut/spectral sparsifier \cite{AbrahamDKKP16}.
It will be exciting to have algorithms against adaptive adversaries with comparable update time for these problems.
%if one can make one of them against adaptive adversaries.

%% file: local_pruning.tex
\section{Reduction from One-shot Expander Pruning to LBS Cuts\label{sec:LBS cut to Prune}}

In this section, we show the proof of \ref{lem:reduc to LBS}. 
\begin{thm}
[Restatement of \ref{lem:reduc to LBS}]\label{thm:reduc to LBS full}Suppose
there is a $(c_{size}(\sigma),c_{con}(\sigma))$-approximate LBS cut
algorithm with running time $t_{LSB}(n,vol(A),\alpha,\sigma)$ when
given $(G,A,\sigma,\alpha)$ as inputs where $G=(V,E)$ is an $n$-node
graph, $A\subset V$ is a set of nodes, $\sigma$ is an overlapping
parameter, and $\alpha$ is a conductance parameter. Then, there is
a one-shot expander pruning algorithm as in \ref{thm:local pruning}
with input $(G,D,\alpha_{b},\epsilon)$ that has \emph{time limit}
\[
\overline{t}=O((\frac{|D|}{\alpha_{b}})^{\epsilon}\cdot\frac{c_{size}(\alpha_{b}/2)}{\epsilon}\cdot t_{LSB}(n,\frac{\Delta|D|}{\alpha_{b}},\alpha_{b},\alpha_{b}))
\]
and \emph{conductance guarantee} 
\[
\alpha=\frac{\alpha_{b}}{5c_{con}(\alpha_{b}/2)^{1/\epsilon-1}}.
\]
More precisely, there is an algorithm $\cA$ that can do the following:
\begin{itemize}
\item $\cA$ is given $G$,$D,\alpha_{b},\epsilon$ as inputs: $G=(V,E)$
is an $n$-node $m$-edge graph with maximum degree $\Delta$, $\alpha_{b}$
is a conductance parameter, $\epsilon\in(0,1)$ is a parameter, and
$D$ is a set of edges where $D\cap E=\emptyset$ where $|D|=O(\alpha_{b}^{2}m/\Delta)$.
Let $G_{b}=(V,E\cup D)$. 
\item Then, in time $\overline{t}=O((\frac{|D|}{\alpha_{b}})^{\epsilon}\cdot\frac{c_{size}(\alpha_{b}/2)}{\epsilon}\cdot t_{LSB}(n,\frac{\Delta|D|}{\alpha_{b}},\alpha_{b},\alpha_{b}))$,
$\cA$ either reports $\phi(G_{b})<\alpha_{b}$, or output a set of
\emph{pruning nodes }$P\subset V$. Moreover, if $\phi(G_{b})\ge\alpha_{b}$,
then we have

\begin{itemize}
\item $vol_{G}(P)\le2|D|/\alpha_{b}$, and
\item a pruned graph $H=G[V-P]$ has high conductance: $\phi(H)\ge\alpha=\frac{\alpha_{b}}{5c_{con}(\alpha_{b}/2)^{1/\epsilon-1}}.$.
\end{itemize}
\end{itemize}
\end{thm}
Observe that $\epsilon$ is a trade-off parameter such that, on one
hand when $\epsilon$ is small, the algorithm is fast but has a bad
conductance guarantee in the output, on the other hand when $\epsilon$
is big, the algorithm is slow but has a good conductance guarantee.

\subsection{The Reduction\label{sec:local decomp alg}}

Throughout this section, let $\epsilon$ be the parameter and let
$G$,$D,p,\alpha_{b}$ be the inputs of the algorithm where $G=(V,E)$
is an $n$-node graph, $D$ is a set of edges where $D\cap E=\emptyset$.
$p$ is a failure probability parameter, and $\alpha_{b}$ is a conductance
parameter where $\alpha_{b}<\frac{1}{\gamma^{\omega(1)}}$. We call
$G_{b}=(V,E\cup D)$ the \emph{before graph}. We want to compute the
set of \emph{pruning nodes }$P\subset V$ with properties according
to \ref{thm:local pruning}. 

We now define some more notations. Let $A$ be the set of endpoints
of $D$. Let $\cA_{cut}$ be the deterministic algorithm for finding
LBS cuts from \ref{thm:LBS cut alg}. We set the overlapping parameter
$\sigma=\alpha_{b}/2$ for $\cA_{cut}$. Let $c_{size}=c_{size}(\sigma),c_{con}=c_{con}(\sigma)$
be the approximation ratios of $\cA_{cut}$.

Let $\bar{s}_{1},\dots,\bar{s}_{L}$ be such that $\bar{s}_{1}=2|D|/\alpha_{b}+1$,
$\bar{s}_{L}\le1$, and $\bar{s}_{\ell}=\bar{s}_{\ell-1}/(\bar{s}_{1})^{\epsilon}$
for $1<\ell<L$. Hence, $L\le1/\epsilon$. We denote $\alpha=\frac{\alpha_{b}}{5c_{con}^{L-1}}$.
Let $\alpha_{1},\dots,\alpha_{L}$ be such that $\alpha_{L}=\alpha$
and $\alpha_{\ell}=\alpha_{\ell+1}c_{con}$ for $\ell<L$. Hence,
$\alpha_{L}<\dots<\alpha_{2}<\alpha_{1}=\alpha_{b}/5<\alpha_{b}/4$.

For any graphs $H=(V_{H},E_{H})$, $I=(V_{I},E_{I})$, and a number
$\ell$, the main procedure $\decomp(H,I,\ell)$ is defined as in
\ref{alg:local decomp}. For any $\alpha'$ and $B\subset V_{H}$,
recall that $\opt(H,\alpha')$ is the size of the largest $\alpha'$-sparse
cut $S$ in $H$ where $|S|\le|V_{H}-S|$, and $\opt(H,\alpha',B,\sigma)$
is the size of the largest $\alpha'$-sparse $(B,\sigma)$-overlapping
cut $S$ in $H$ where $|S|\le|V_{H}-S|$. By definition, $\opt(H,\alpha')\ge\opt(H,\alpha',B,\sigma)$.

The algorithm is simply to run $\decomp(G,G,1)$ with time limit $\bar{t}$.
If $\decomp(G,G,1)$ takes time more than $\overline{t}$, then we
reports that $\phi(G_{b})<\alpha_{b}$ (FAIL).

\begin{algorithm}
\caption{\label{alg:local decomp}$\protect\decomp(H,I,\ell)$ where $H=(V_{H},E_{H})$
and $I=(V_{I},E_{I})$}

\begin{enumerate}
\item Set $B_{H}=(A\cup A_{H})\cap V_{H}$ where $A_{H}$ is the set of
endpoints of edges in $\partial_{G}(V_{H})$
\item If $vol_{H}(V_{H}-B_{H})<\frac{3}{\sigma}vol_{H}(B_{H})$, then report
$\phi(G_{b})<\alpha_{b}$ (FAIL). 
\item If $\ell=L$, then return.
\item If $\cA_{cut}(H,\alpha_{\ell},B_{H},\sigma)$ reports $\opt(H,\alpha_{\ell}/c_{con},B_{H},\sigma)=0$,
i.e. there is no $(\alpha_{\ell}/c_{con})$-sparse $(B_{H},\sigma)$-overlapping
cut, then return. 
\item Else, $\cA_{cut}(H,\alpha_{\ell},B_{H},\sigma)$ outputs an $\alpha_{\ell}$-sparse
cut $S$ in $H$ where $\opt(H,\alpha_{\ell}/c_{con},B_{H},\sigma)/c_{size}\le vol(S)\le vol(V_{H})/2$.

\begin{enumerate}
\item If $|S|\ge\bar{s}_{\ell+1}/c_{size},$ then include $S$ into pruning
set $P$ and recurse on $\decomp(H[V_{H}\setminus S],I,\ell)$. 
\item Else, recurse on $\decomp(H,H,\ell+1)$. \end{enumerate}
\end{enumerate}
\end{algorithm}

\subsubsection{Upper Bounding $vol(P)$}

In this section, we prove that if $\phi(G_{b})\ge\alpha_{b}$ and
the algorithm does not fail by other reasons, then we have $vol_{G}(P)\le2|D|/\alpha_{b}$.
We will show that the algorithm might fail only if $\phi(G_{b})<\alpha_{b}$
in \ref{sec:Validity}.

Let us list all sets of nodes $P_{1},\dots,P_{t}$ that are outputted
by \ref{alg:local decomp} in either Step 2 or Step 5.a and constitute
the pruning set $P=\bigcup_{i}P_{i}$. Note that $P_{i}\cap P_{i'}=\emptyset$
for any $i,i'$. The sets $P_{1},\dots,P_{t}$ is ordered by the time
they are outputted. Let $H{}_{1},\dots,H{}_{t}$ the corresponding
subgraphs such that $P_{i}$ is ``cut from'' $H_{i}$, i.e. $H{}_{1}=G$,
$H{}_{2}=G[V-P{}_{1}],\dots,H{}_{t}=G[V-\bigcup_{i=1}^{t-1}P_{i}]$.
Note the following fact: 
\begin{fact}
For any $i\le t$, $\phi_{H_{i}}(P_{i})<\alpha_{b}/4$.\label{claim:sparse for t-1}\end{fact}
\begin{proof}
Since all $P_{i}$'s are returned in Step 5.a and $\alpha_{\ell}\le\alpha_{1}<\alpha_{b}/4$.
\end{proof}
Next, we have the following:
\begin{prop}
$\delta_{G}(P)\le\sum_{i=1}^{t}\delta_{H{}_{i}}(P_{i})$\label{fact:sum crossing edges}\end{prop}
\begin{proof}
We will prove that $\partial_{G}(P)\subseteq\bigcup_{i=1}^{t}\partial_{H_{i}}(P_{i})$.
Let $(u,v)\in\partial_{G}(P)$. Suppose that $u\in P_{j}$ for some
$j\le t$. Then $v\in V-P\subseteq V(H{}_{j})-P_{j}$ because $V(H_{j})=V-\bigcup_{i=1}^{j-1}P_{i}$.
Therefore, $(u,v)\in\partial_{H{}_{j}}(P_{j})\subset\bigcup_{i=1}^{t}\partial_{H{}_{i}}(P_{i})$.\end{proof}
\begin{prop}
For any $t'\le t$, if $vol_{G}(\bigcup_{i=1}^{t'}P{}_{i})\le vol(G)/2$,
then $\phi_{G}(\bigcup_{i=1}^{t'}P_{i})<\alpha_{b}/4$.\label{prop:small then sparse}\end{prop}
\begin{proof}
We have
\begin{align*}
\phi_{G}(\bigcup_{i=1}^{t'}P_{i}) & =\frac{\delta_{G}(\bigcup_{i=1}^{t'}P_{i})}{vol_{G}(\bigcup_{i=1}^{t'}P_{i})} & \mbox{as }vol(\bigcup_{i=1}^{t'}P{}_{i})\le vol(G)/2\\
 & =\frac{\delta_{G}(\bigcup_{i=1}^{t'}P_{i})}{\sum_{i=1}^{t'}vol_{G}(P_{i})} & \mbox{as }P_{i}\mbox{'s are disjoint}\\
 & \le\frac{\sum_{i=1}^{t}\delta_{H{}_{i}}(P_{i})}{\sum_{i=1}^{t'}vol_{H_{i}}(P_{i})} & \mbox{by \ref{fact:sum crossing edges}}\\
 & \le\max_{i\le t'}\frac{\delta_{H{}_{i}}(P_{i})}{vol_{H_{i}}(P_{i})}\\
 & =\max_{i\le t'}\phi_{H{}_{i}}(P_{i})<\alpha_{b}/4 & \mbox{by \ref{claim:sparse for t-1}.}
\end{align*}

\end{proof}
The following lemma is the key observation:
\begin{lem}
Suppose that $\phi(G_{b})\ge\alpha_{b}$. If a cut $S$ is $G$ where
$vol(S)\le vol(G)/2$ is such that $\phi_{G}(S)<\alpha_{b}/2$, then
$vol_{G}(S)\le2|D|/\alpha_{b}$.\label{thm:sparse cut small}\end{lem}
\begin{proof}
Suppose otherwise that $vol(S)>2|D|/\alpha_{b}$. Then, as $G=G_{b}-D$,
we have 
\[
\delta_{G}(S)\ge\delta_{G_{b}}(S)-|D|>\alpha_{b}vol(S)-\alpha_{b}vol(S)/2=\alpha_{b}vol(S)/2,
\]
which means, $\phi_{G}(S)>\alpha_{b}/2>\alpha'$, a contradiction.
\end{proof}
Next, we argue against a corner case where $vol(P)>vol(G)/2$.
\begin{lem}
If $\phi(G_{b})\ge\alpha_{b}$ and $vol_{G}(V)=\omega(|D|/\alpha_{b})$,
then $vol(P)\le vol(G)/2$.\label{lem:less than half}\end{lem}
\begin{proof}
Suppose that $vol(P)>vol(G)/2$. We will show that $vol_{G}(V)=O(|D|/\alpha_{b})$
which is a contradiction. Let $k\le t$ be the such that $vol(\bigcup_{i=1}^{k-1}P_{i})\le vol_{G}(V)/2<vol(\bigcup_{i=1}^{k}P_{i})$.
Let $P_{<k}=\bigcup_{i=1}^{k-1}P_{i}$. We partition the vertices
into 3 sets: $P_{<k},P_{k}$ and $Q=V-(P_{<k}\cup P_{k})$. Note that
$H_{k}=G[V-P_{<k}]$.

First, we list some properties of $P_{<k}$. Because $vol(P_{<k})<vol_{G}(V)/2$,
by \ref{prop:small then sparse}, we have $\phi_{G}(P_{<k})<\alpha_{b}/4$
and hence, by \ref{thm:sparse cut small}, $vol_{G}(P_{<k})<2|D|/\alpha_{b}\ll vol_{G}(V)/6$.

Next, we list some properties of $P_{k}$. We have $\phi_{H_{k}}(P_{k})<\alpha_{b}/4$
and 
\begin{align*}
vol_{G}(P_{k}) & \le E_{G}(P_{<k},P_{k})+vol_{H_{k}}(P_{k})\\
 & \le vol_{G}(P_{<k})+vol_{H_{k}}(Q)\\
 & \le vol_{G}(P_{<k})+vol_{G}(Q).
\end{align*}

Last, we list properties of $Q$. We have that $\delta_{G}(Q)=\delta_{G}(P_{<k}\cup P_{k})\le\delta_{G}(P_{<k})+\delta_{H_{k}}(P_{k})$
using the same argument as in \ref{fact:sum crossing edges}. We claim
that $vol_{G}(P_{<k})+vol_{G}(P_{k})\le2vol_{G}(Q)$. Because 
\begin{align*}
vol_{G}(V) & =vol_{G}(P_{<k})+vol_{G}(P_{k})+vol_{G}(Q)\\
 & \le2vol_{G}(P_{<k})+2vol_{G}(Q)\\
 & \le vol(G)/3+2vol_{G}(Q),
\end{align*}
and so we have $vol_{G}(Q)\ge vol_{G}(V)/3$ and hence $vol_{G}(P_{<k})+vol_{G}(P_{k})\le\frac{2}{3}vol_{G}(V)\le2vol_{G}(Q)$.
Therefore, we have the following: 
\begin{align*}
\phi_{G}(Q) & =\frac{\delta_{G}(Q)}{vol_{G}(Q)}\le\frac{\delta_{G}(P_{<k})+\delta_{H_{k}}(P_{k})}{(vol_{G}(P_{<k})+vol_{G}(P_{k}))/2}\le2\max\{\phi_{G}(P_{<k}),\phi_{H_{k}}(P_{k})\}<\alpha_{b}/2.\\
\end{align*}
This means that $vol_{G}(Q)<2|D|/\alpha_{b}$ by \ref{thm:sparse cut small}.
So we can conclude the contradiction: 
\[
vol_{G}(V)=vol_{G}(P_{<k})+vol_{G}(P_{k})+vol_{G}(Q)=O(|D|/\alpha_{b}).
\]

\end{proof}
By the above lemma, we immediately have a strong bound on $vol(P)$.
Note that $|D|<\alpha_{b}^{2}m/30\Delta$ implies that $vol_{G}(V)=\omega(|D|/\alpha_{b})$
\begin{cor}
If $\phi(G_{b})\ge\alpha_{b}$, then $vol_{G}(P)\le2|D|/\alpha_{b}$.\label{thm:bound prune vol}\end{cor}
\begin{proof}
By \ref{lem:less than half}, $vol(P)\le vol(G)/2$ and hence $\phi_{G}(P)<\alpha_{b}/4$
by \ref{prop:small then sparse}, this implies that $vol(P)\le2|D|/\alpha_{b}$
by \ref{thm:sparse cut small}.
\end{proof}

\subsubsection{Upper bounding $vol_{H}(B_{H})$\label{sec:Validity}}

First, we prove that $|D|<\alpha_{b}^{2}m/30\Delta$ implies that
$\frac{3}{\sigma}vol_{H}(B_{H})\le vol_{H}(V_{H}-B_{H})$ unless $\phi(G_{b})<\alpha_{b}$.
\begin{prop}
If $\phi(G_{b})\ge\alpha_{b}$, then $vol_{H}(B_{H})\le4\Delta|D|/\alpha_{b}.$\label{prop:bound B_H}\end{prop}
\begin{proof}
We have
\begin{align*}
vol_{H}(B_{H}) & \le vol_{H}(A)+vol_{H}(A_{H})\\
 & \le2\Delta|D|+\Delta|E_{G}(P,V-P)|\\
 & \le2\Delta|D|+\Delta vol_{G}(P)\\
 & \le2\Delta|D|+2\Delta|D|/\alpha_{b} & \mbox{by \ref{thm:bound prune vol}}\\
 & \le4\Delta|D|/\alpha_{b}
\end{align*}
\end{proof}
\begin{lem}
Suppose that $|D|<\alpha_{b}^{2}m/30\Delta$. If $\phi(G_{b})\ge\alpha_{b}$,
then the condition in Step 2 of \ref{alg:local decomp} never holds.\label{lem:no fail Step 2}\end{lem}
\begin{proof}
Suppose that the condition in Step 2 holds. Let $P$ be the pruning
set that \ref{alg:local decomp} outputted so far. We have that $H=G[V_{H}]=G[V-P]$
such that $vol_{H}(V_{H}-B_{H})<\frac{3}{\sigma}vol_{H}(B_{H})$ where
$B_{H}=(A\cup A_{H})\cap V_{H}$, $A$ is the endpoints of $D$, and
$A_{H}$ is the set of endpoints of edges in $\partial_{G}(V_{H})=\partial_{G}(P)$.
This implies that 

\[
vol_{H}(V_{H})=vol_{H}(V_{H}-B_{H})+vol_{H}(B_{H})<(1+\frac{3}{\sigma})vol_{H}(B_{H}).
\]
Together, we have that $vol_{H}(V-P)=vol_{H}(V_{H})\le\frac{16\Delta}{\sigma}(|D|/\alpha_{b})$
by \ref{prop:bound B_H}. This implies 
\begin{align*}
vol_{G}(V) & =vol_{G}(P)+\delta_{G}(P)+vol_{H}(V-P)\\
 & \le2vol_{G}(P)+vol_{H}(V-P)\\
 & \le60\Delta|D|/\alpha_{b}^{2}. & \mbox{by \ref{thm:bound prune vol} and \ensuremath{\sigma}=\ensuremath{\alpha_{b}}/2}.
\end{align*}
But this is contradiction because it means $|D|\ge\alpha_{b}^{2}m/30\Delta$.
\end{proof}
This implies that the parameters for $\cA_{cut}$ are valid when it
is called. 
\begin{lem}
Whenever $\cA_{cut}(H,\alpha_{\ell},B_{H},\sigma)$ is called, we
have that $\sigma\ge\frac{3vol(B)}{vol(V_{H}-B_{H})}$ satisfying
the requirement for $\cA_{cut}$ as stated in \ref{thm:LBS cut alg}.\end{lem}
\begin{proof}
Observe that $\cA_{cut}$ can be called only when the condition in
Step 2 of \ref{alg:local decomp} is false: $vol(V_{H}-B_{H})\ge\frac{3}{\sigma}vol(B_{H})$.
That is, $\sigma\ge\frac{3vol(B_{H})}{vol(V_{H}-B_{H})}$ and $4vol(B_{H})\le vol(V_{H}-B_{H})$.
\end{proof}

\subsubsection{Lower Bounding Conductance}

Next, given that $\phi(G_{b})\ge\alpha_{b}$, we would like to prove
an important invariant: if $\decomp(H,I,\ell)$ is called, then $\opt(I,\alpha_{\ell})<\bar{s}_{\ell}$.
In order to prove this, we need two lemmas.
\begin{lem}
Suppose that $\phi(G_{b})\ge\alpha_{b}$. If $\decomp(H,I,1)$ is
called, then $\opt(I,\alpha_{1})<2|D|/\alpha_{b}+1=\bar{s}_{1}$.\label{lem:invaraint 1}\end{lem}
\begin{proof}
Observe that when $\ell=1$, we have $I=G$. \ref{thm:sparse cut small}
implies that $\opt(G,\alpha_{1})\le\opt(G,\alpha_{b}/2)\le2|D|/\alpha_{b}$.
\end{proof}
Recall that, in an induced subgraph $H=G[V_{H}]$ in $G$, we denote
$B_{H}=(A\cup A_{H})\cap V_{H}$ where where $A$ is the endpoints
of $D$ and $A_{H}$ is the set of endpoints of edges in $\partial_{G}(V_{H})$.
\begin{lem}
Suppose that $\phi(G_{b})\ge\alpha_{b}$. $\opt(H,\alpha')=\opt(H,\alpha',B_{H},\sigma)$
for any $\alpha'<\alpha_{b}/2$ and induced subgraph $H=G[V_{H}]$.\label{thm:sparse must overlap}\end{lem}
\begin{proof}
In words, we need to prove that any $\alpha'$-sparse cut $S\subset V_{H}$
where $|S|\le|V_{H}-S|$ must be $(B_{H},\sigma)$-overlapping in
$H$.

First, consider any cut edge $(u,v)\in\partial_{G_{b}}(S)$ in the
before graph $G_{b}$ where $u\in S$. We claim that either $(u,v)\in\partial_{H}(S)$
or $u\in A\cup A_{H}$. Indeed, if $u\notin A\cup A_{H}$, i.e. $u$
is not incident to any edge in $D$ nor $\partial_{G}(V_{H})$, so
all edges incident to $u$ are inside $H$, and hence $(u,v)\in\partial_{H}(S)$.
It follows that $\delta_{G_{b}}(S)\le\delta_{H}(S)+vol(S\cap(A\cup A_{H}))$.

Suppose that there is an $\alpha'$-sparse cut $S\subset V_{H}$ which
is not $(B_{H},\sigma)$-overlapping, i.e. $vol(S\cap(A\cup A_{H}))<\sigma vol(S)=\frac{\alpha_{b}}{2}vol(S)$.
Then we have that
\[
\delta_{H}(S)\ge\delta_{G_{b}}(S)-vol(S\cap(A\cup A_{H}))>\alpha_{b}vol(S)-\alpha_{b}vol(S)/2=\alpha_{b}vol(S)/2.
\]
That is, $\phi_{H}(S)\ge\alpha_{b}/2>\alpha'$, which is a contradiction.
\end{proof}
Now, we can prove \textbf{the main invariant}:
\begin{lem}
\label{prop:invariant local decomp}Suppose that $\phi(G_{b})\ge\alpha_{b}$.
If $\decomp(H,I,\ell)$ is called, then the invariant $\opt(I,\alpha_{\ell})<\bar{s}_{\ell}$
is satisfied. \end{lem}
\begin{proof}
When $\ell=1$, $\opt(I,\alpha_{\ell})<\bar{s}_{1}$ by \ref{lem:invaraint 1}.
In particular, the invariant is satisfied when $\decomp(G,G,1)$.
The invariant for $\decomp(H[V_{H}\setminus S],I,\ell)$ which is
called in Step 5.a is the same as the one for $\decomp(H,I,\ell)$,
and hence is satisfied by induction. 

Finally, we claim that the invariant is satisfied when $\decomp(H,H,\ell+1)$
is called, i.e., $\opt(H,\alpha_{\ell+1})<\bar{s}_{\ell+1}$. By Step
5.a, $|S|<\bar{s}_{\ell+1}/c_{size}$. By Step 5, $\opt(H,\alpha_{\ell+1},B_{H},\sigma)/c_{size}\le|S|$
as $\alpha_{\ell+1}=\alpha_{\ell}/c_{con}$. Since $H$ is induced
by $V_{H}$ and $\alpha_{\ell+1}\le\alpha_{1}<\alpha_{b}/2$ satisfying
the conditions in \ref{thm:sparse must overlap}, we have $\opt(H,\alpha_{\ell+1})=\opt(H,\alpha_{\ell+1},B_{H},\sigma)$.
Therefore, $\opt(H,\alpha_{\ell+1})\le c_{size}|S|<\bar{s}_{\ell+1}$
as desired. 
\end{proof}
Finally, we bound the conductance of the components of $G^{d}$.
\begin{lem}
\label{lem:expansion local decomp}Suppose that $\phi(G_{b})\ge\alpha_{b}$.
The pruned graph $H=G[V-P]$ has conductance $\phi(H)\ge\alpha$.\end{lem}
\begin{proof}
$H$ is in either Step 3 or 4 in \ref{alg:local decomp}. First, if
$H$ is returned in Step 3, then $\decomp(H,H,L)$ was called. By
the invariant, we have $\opt(H,\alpha_{L})<\bar{s}_{L}\le1$, i.e.
there is no $\alpha_{L}$-sparse cut in $H$. As $\alpha_{L}=\alpha$,
$\phi(H)\ge\alpha$. Second, if $H$ is returned in Step 4, then $\cA_{cut}(H,\alpha_{\ell},B,\sigma)$
reports that $\opt(H,\alpha_{\ell}/c_{con},B_{H},\sigma)=0$. As $\alpha_{\ell}/c_{con}<\alpha_{b}/2$,
$\opt(H,\alpha_{\ell}/c_{con})=0$ by \ref{thm:sparse must overlap}.
That is, $\phi(H)\ge\alpha_{\ell}/c_{con}\ge\alpha_{L}=\alpha$.
\end{proof}
Now, it is left to analyze the running time.

\subsubsection{Running time}

In this section, we prove that if $\phi(G_{b})\ge\alpha_{b}$, then
$\decomp(G,G,1)$ takes at most $\overline{t}$ time. In other words,
if $\decomp(G,G,1)$ takes more that $\bar{t}$ time, then $\phi(G_{b})<\alpha_{b}$
and the algorithm will just halt and report that $\phi(G_{b})<\alpha_{b}$
(FAIL).

To analyze the running time, we need some more notation. We define
the recursion tree $\mathcal{T}$ of $\decomp(G,G,1)$. Actually,
$\cT$ is a path.
\begin{itemize}
\item Each node of $\mathcal{T}$ represents the parameters of the procedure
$\decomp(H,I,\ell)$. 
\item $(G,G,1)$ is the root node. 
\item For each $(H,I,\ell)$, if $\decomp(H,I,\ell)$ returns the pruned
graph $H$, then $(H,I,\ell)$ is a leaf. 
\item If $\decomp(H,I,\ell)$ recurses on $\decomp(H[V_{H}\setminus S],I,\ell)$,
then there is an edge $((H,I,\ell),(H[V_{H}\setminus S],I,\ell))\in\cT$
and is called a \emph{right edge}. 
\item If $\decomp(H,I,\ell)$ recurses on $\decomp(H,H,\ell+1)$, then there
is an edge $((H,I,\ell),(H,H,\ell+1))$ in $\cT$ and is called a
\emph{down edge}. 
\end{itemize}
For any $n$-node graph $G$, $\alpha$ is a conductance parameter,
$A$ is a set of nodes in $G$, and $\sigma$ is an overlapping parameter,
let $t_{LSB}(n,vol(A),\alpha,\sigma)$ denote the running time of
the LSB cut algorithm $\cA_{cut}(G,\alpha,A\sigma)$.
\begin{lem}
\label{lem:depth implies time local}Suppose that $\phi(G_{b})\ge\alpha_{b}$.
If the depth of the recursion tree $\mathcal{T}$ is $d_{\cT}$, then
the total running time is $O(d_{\cT}\times t_{LSB}(n,\frac{\Delta|D|}{\alpha_{b}},\alpha_{b},\alpha_{b}))$.\end{lem}
\begin{proof}
At any level of recursion, the running time on $\decomp(H,I,\ell)$,
excluding the time spent in the next recursion level, is just the
running time of $\cA_{cut}(H,\alpha_{\ell},B_{H},\sigma)$ which is
$t_{LSB}(|V_{H}|,vol(B_{H}),\alpha_{\ell},\sigma)\le O(t_{LSB}(n,\frac{\Delta|D|}{\alpha_{b}},\alpha_{b},\alpha_{b}))$
because $|V_{H}|\le n$, $vol(B_{H})=O(\frac{\Delta|D|}{\alpha_{b}})$
\ref{prop:bound B_H}, $\alpha_{\ell}\le\alpha_{b}$, and $\sigma=\alpha_{b}/2$.
So the total running time is $O(d_{\cT}\times t_{LSB}(n,\frac{\Delta|D|}{\alpha_{b}},\alpha_{b},\alpha_{b}))$
is if there are $d_{\cT}$ levels.
\end{proof}
Now, we bound the depth $d_{\cT}$ of $\cT$. Recall $L\le1/\epsilon$
and $\bar{s}_{1}=2|D|/\alpha_{b}+1$.
\begin{lem}
\label{lem:bound depth local}Suppose that $\phi(G_{b})\ge\alpha_{b}$.
$\cT$ contains at most $L$ down edges, and $Lc_{size}\bar{s}_{1}^{\epsilon}$
right edges. That is, the depth of $\cT$ is $d_{\cT}=Lc_{size}\bar{s}_{1}^{\epsilon}=O((\frac{|D|}{\alpha_{b}})^{\epsilon}\cdot\frac{c_{size}}{\epsilon})$.\end{lem}
\begin{proof}
There are at most $L$ down edges in $\cT$ because Step 3 in \ref{alg:local decomp}
always terminates the recursion when $\ell=L$. Next, it suffices
to prove that there cannot be $k=c_{size}\bar{s}_{1}^{\epsilon}$
right edges between any down edges in $\cT$. Let $P=(H_{1},I,\ell),\dots,(H_{k},I,\ell)$
be a path of $\cT$ which maximally contains only right edges. Note
that $I=H_{1}$ because $(H_{1},I,\ell)$ must be either a root or
a deeper endpoint of a down edges. Suppose for contradiction that
$|P|\ge k$. 

For each $i$, let $S_{i}$ be the cut such that $H_{i+1}=H_{i}[V_{H_{i}}\setminus S_{i}]$
and $\phi_{H_{i}}(S_{i})<\alpha_{\ell}$. Since $\{S_{i}\}_{i\le k}$
are mutually disjoint. We conclude $\phi_{H_{1}}(\bigcup_{i=1}^{k}S_{i})<\alpha_{\ell}$
using the same argument as in \ref{prop:small then sparse}. However,
we also have that $|S_{i}|\ge\bar{s}_{\ell+1}/c_{size}$, for all
$i$, and hence $|\bigcup_{i=1}^{k}S_{i}|%=\sum_{i=1}^{k}|S_{i}|
\ge k\bar{s}_{\ell+1}/c_{size}\ge\bar{s}_{1}^{\epsilon}\bar{s}_{\ell+1}=\bar{s}_{\ell}$. So $\bigcup_{i=1}^{k}S_{i}$ contradicts the invariant for $\decomp(H_{1},I,\ell)$,
where $I=H_{1}$, which says $\opt(H_{1},\alpha_{\ell})<\bar{s}_{\ell}$.
Note that the invariant must hold by \ref{prop:invariant local decomp}.\end{proof}
\begin{cor}
Suppose that $\phi(G_{b})\ge\alpha_{b}$. $\decomp(G,G,1)$ runs in
time $\overline{t}=O((\frac{|D|}{\alpha_{b}})^{\epsilon}\cdot\frac{c_{size}}{\epsilon}\cdot t_{LSB}(n,\frac{\Delta|D|}{\alpha_{b}},\alpha_{b},\alpha_{b}))$.\label{cor:local decomp time}\end{cor}
\begin{proof}
We have $d_{\cT}=O(\frac{|D|^{\epsilon}}{\alpha_{b}^{1+\epsilon}\epsilon})$
by \ref{lem:bound depth local}. By \ref{lem:depth implies time local},
we have that $\decomp(G,G,1)$ runs in time $O((\frac{|D|}{\alpha_{b}})^{\epsilon}\cdot\frac{c_{size}}{\epsilon}\cdot t_{LSB}(n,\frac{\Delta|D|}{\alpha_{b}},\alpha_{b},\alpha_{b}))$. 
\end{proof}
Now, we can conclude the main theorem. 
\begin{proof}
[Proof of \ref{thm:reduc to LBS full}]Suppose that $\phi(G_{b})\ge\alpha_{b}$.
By \ref{thm:bound prune vol}, we have that the pruning set $P$ has
small volume $vol_{G}(V)\le2|D|/\alpha_{b}$. By \ref{lem:expansion local decomp},
the pruned graph $H=G[V-P]$ has high conductance: $\phi(C)\ge\alpha=\frac{\alpha_{b}}{5c_{con}(\alpha_{b}/2)^{1/\epsilon-1}}$.
Finally, \ref{cor:local decomp time} $\decomp(G,G,1)$ runs in time
$\overline{t}=O((\frac{|D|}{\alpha_{b}})^{\epsilon}\cdot\frac{c_{size}(\alpha_{b}/2)}{\epsilon}\cdot t_{LSB}(n,\frac{\Delta|D|}{\alpha_{b}},\alpha_{b},\alpha_{b}))$. 

When $\decomp(G,G,1)$ reports failure, we claim that $\phi(G_{b})<\alpha_{b}$.
Indeed, given that $|D|=O(\alpha_{b}^{2}m/\Delta)$, we have that
$\decomp(G,G,1)$ does not fails in Step 2 by \ref{lem:no fail Step 2}.
So $\decomp(G,G,1)$ can return FAIL only when its running time exceeds
$\overline{t}$ which can only happens when $\phi(G_{b})<\alpha_{b}$.\end{proof}

%% file: omitted.tex
\section{Omitted Proofs}

\subsection{Proof of \ref{thm:simplifed unit flow} (Unit Flow)\label{sec:proof unit flow}}

The proof is basically by adjusting and simplifying parameters of
the following statement from \cite{HenzingerRW17}. \footnote{We note that the original statement of \cite{HenzingerRW17} guarantees
some additional properties. We do not state them here as they are
not needed. Moreover, the result of the algorithm is a flow in the
graph $G(U)$ which is the graph $G$ where the \emph{capacity }of
each edge is $U$, for a given parameter $U$. We adapt the original
statement to use the notion of \emph{congestion }instead. It is easy
to see that a flow in $G(U)$ is a flow in $G$ with congestion $U$
and vice versa.}
\begin{lem}
[Theorem 3.1 and Lemma 3.1 in \cite{HenzingerRW17}]\label{thm:orig unit flow}There
exists an algorithm \emph{called Unit Flow }which takes the followings
as input: a graph $G=(V,E)$ with $m$ edges (with parallel edges
but no self loop), positive integers $h'$, $F'$, and $U$, a source
function $\Delta$ such that $\Delta(v)\le F'\cdot\deg(v)$ for all
$v\in V$, and a sink function $T$ where $T(v)=\deg(v)$ for all
$v\in V$. In time $O(F'h'\total{\Delta})$, the algorithm returns
a source-feasible preflow $f$ with congestion at most $U$. Moreover,
one of the followings holds.
\begin{enumerate}
\item $\total{ex_{f}}=0$ i.e. $f$ is a flow.
\item $\total{ex_{f}}\le\total{\Delta}-2m$. More specifically, for all
$v\in V$, $\deg(v)$ units of supply is absorbed at $v$. 
\item A set $S\subseteq V$ is returned, where $S$ is such that $\forall v\in S$,
$\deg(v)\le f(v)\le F'\cdot\deg(v)$ and $\forall v\notin S$, $f(v)\le\deg(v)$.
Furthermore, if $h\ge\log m$, then $\phi_{G}(S)\le\frac{20\log2m}{h'}+\frac{F'}{U}$.
\end{enumerate}
\end{lem}
\begin{comment}
NOT NEEDED: If $h=\Omega(\log m'\log\log m')$ for $m'\ge m$, let
$K$ be the smaller side of $(S,\overline{S})$, then $\phi_{G}(K)\le\frac{\log m+1-\left\lceil \log(vol(K)\right\rceil }{50\log m'}+\frac{w}{U}.$ 
\end{comment}
By restricting to only when $\total{\Delta}\le2m$ and $h'\ge\log m$,
we can simplify \ref{thm:orig unit flow} to \ref{thm:simplifed unit flow-1}
below. Note that, when $T(v)=\deg(v)$, we have $ex_{f}(v)=\max\{f(v)-T(v),0\}=ex_{f}(v)=\max\{f(v)-\deg(v),0\}$.
\begin{lem}
\label{thm:simplifed unit flow-1}There exists an algorithm \emph{called
Unit Flow }which takes the followings as input: a graph $G=(V,E)$
with $m$ edges (with parallel edges but no self loop), positive integers
$h'\geq\log m$, $F'$, and $U$, a source function $\Delta$ such
that $\Delta(v)\le F'\cdot\deg(v)$ for all $v\in V$ and $\total{\Delta}\le2m$,
and a sink function $T$ where $T(v)=\deg(v)$ for all $v\in V$.
In time $O(F'h'\total{\Delta})$, the algorithm returns a source-feasible
preflow $f$ with congestion at most $U$. Moreover, one of the followings
holds.
\begin{enumerate}
\item $\total{ex_{f}}=0$ i.e. $f$ is a flow.
\item A cut $S$ is returned where $\phi_{G}(S)\le\frac{20\log(2m)}{h'}+\frac{F'}{U}$.
Moreover, $\forall v\in S$: $ex_{f}(v)\le(F'-1)\deg(v)$ and $\forall v\notin S$:
$ex_{f}(v)=0$.
\end{enumerate}
\end{lem}
To get \ref{thm:simplifed unit flow}, set $F=F'$ , $h=\frac{h'}{41\log(2m)}$
and $U=2hF$. So $\frac{20\log(2m)}{h'}+\frac{F'}{U}<\frac{1}{2h}+\frac{1}{2h}=\frac{1}{h}$.
Note that the condition $h'\geq\log m$ becomes $h\geq1.$ Also note
that in the input of \ref{thm:Extended Unit Flow} (thus \ref{thm:simplifed unit flow}),
$\total{\Delta}\le|T(\cdot)|\leq2m$, thus the condition that $\total{\Delta}\le2m$
in \ref{thm:simplifed unit flow-1} can be dropped. This completes
the proof of \ref{thm:simplifed unit flow}.